\documentclass[11pt,a4paper]{article}
\usepackage{amsmath, amssymb, amsfonts, amsthm}
\usepackage{hyperref}

\newtheorem{theorem}{Theorem}[section]
\newtheorem{proposition}[theorem]{Proposition}

\newcommand{\be}{\begin{equation}}
\newcommand{\ee}{\end{equation}}
\newcommand{\bey}{\begin{eqnarray}}
\newcommand{\eey}{\end{eqnarray}}
\newcommand{\norm}[1]{\|#1\|}

\newcommand{\PP}{{\mathbb P }}

\newcommand{\eps}{\varepsilon}

\newcommand{\bx}{\textbf{x}}

\newcommand{\ph}{\varphi}
\newcommand{\Ph}{\Phi}

\newcommand{\g}{\gamma}
\newcommand{\e}{\varepsilon}
\newcommand{\s}{\sigma}

\newcommand{\cU}{{\cal U}}

\newcommand{\bR}{{\mathbb R}}

\newcommand{\bN}{{\mathbb N}}
\newcommand{\bZ}{{\mathbb Z}}

\newcommand{\tr}{\mbox{Tr}}

\newcommand{\wt}{\widetilde}

\newcommand{\cF}{{\cal F}}

\newcommand{\cE}{{\cal E}}

\newcommand{\cV}{{\cal V}}

\newcommand{\cK}{{\cal K}}
\newcommand{\cH}{{\cal H}}
\newcommand{\cL}{{\cal L}}
\newcommand{\cN}{{\cal N}}

\newcommand{\cJ}{{\cal J}}

\newcommand{\D}{\Delta}

\newcommand{\h}{\mathfrak{h}}
\newcommand{\G}{{\Gamma}}

\input epsf

\newcommand{\donothing}[1]{}

\title{Effective Evolution Equations from Quantum Dynamics}

\author{Niels Benedikter$^{*}$, Marcello Porta$^{\dag}$ and Benjamin Schlein$^{\dag}$\\ \vspace{.3cm} \\ $^{*}$Department of Mathematical Sciences, University of Copenhagen\\
Universitetsparken 5, 2100 K\o{}benhavn \O, Denmark \\ \vspace{.1cm} \\ $^{\dag}$Mathematics Department, University of Z\"urich, \\ Winterthurerstrasse 190, 8057 Z\"urich, Switzerland \\ \vspace{.05cm} }

\begin{document}

\maketitle

\begin{abstract}
In these notes we review the material presented at the summer school on ``Mathematical Physics, Analysis and Stochastics'' held at the University of Heidelberg in July 2014. We consider the time-evolution of quantum systems and in particular the rigorous derivation of effective equations approximating the many-body Schr\"odinger dynamics in certain physically interesting regimes. 
\end{abstract}

\section{Introduction}
\setcounter{equation}{0}
\label{sec:intro}

Systems of interest in physics and other natural sciences can be described at the microscopic and the macroscopic level. Microscopically, a system is described in terms of its elementary constituents and their fundamental interactions. 
While such a description is very accurate, it is typically not well-suited for computations because of the large number of degrees of freedom. Examples of microscopic theories include Newton's theory of classical mechanics, Schr\"odinger's quantum mechanics, quantum electrodynamics and Einstein's general relativity\footnote{Of course we are not claiming these theories to be absolutely fundamental from the view of a physicist. It would be more correct to consider them as different levels between fundamental and effective, and which theory we call effective and which fundamental depends on the pair we are looking at. For example, we could also consider Newtonian mechanics as a macroscopic theory arising as an effective theory from quantum mechanics. On the next level we could view nonrelativistic quantum mechanics as an effective theory arising from the Standard Model.}. On the other hand, a macroscopic description of the system does not resolve the constituents and only takes into account effective interactions. It focuses on macroscopically observable quantities which arise from the collective behavior of the system and are of interest for the observer. Such a description is less accurate but it is much more accessible to computations. Examples of macroscopic theories are Boltzmann's kinetic theory of gases, the Navier-Stokes and the Euler equations of hydrodynamics, the Hartree and Hartree-Fock theory, the BCS theory of superconductors and superfluids, the Ginzburg-Landau theory, the Gross-Pitaevskii theory of Bose-Einstein condensation and the Vlasov theory of plasma physics. 

Because of the great importance of effective macroscopic theories for making qualitative and quantitative predictions about the behavior of physically interesting systems, a key goal of statistical mechanics is to understand their  emergence from microscopic theories in appropriate scaling regimes (also called limits, even though we often think of the parameter as being large but finite). Here mathematical physics can and should play a central role to put the effective theories, which are often obtained merely by heuristic and phenomenological arguments,  on solid grounds and understand the range and the limits of their validity. 

Let us present two examples of physical systems which can be described by effective equations that can be rigorously derived from microscopic theories in appropriate limits. 

\medskip

\emph{Large atoms and molecules.} We consider a quantum-mechanical system of $N$ electrons and $M$ nuclei of charges $Z_1, \dots , Z_M > 0$ located at positions $R_1, \dots , R_M \in \bR^3$. We assume the system to be neutral, i.\,e.\ $N = \sum_{i=1}^M Z_i$. For simplicity we work in the Born-Oppenheimer approximation, i.\,e.\ we keep the nuclei fixed (only the electrons are dynamical particles in this approximation). At zero temperature the system is in its ground state, with energy  
\begin{equation}\label{eq:ham-min} E (N) = \min_{\psi \in L^2_a (\bR^{3N}) : \| \psi \| = 1} \langle \psi, H_{N} \psi \rangle \end{equation}
where $H_{N}$ denotes the Hamilton operator
\begin{equation}\label{eq:ham-TF} H_{N} = \sum_{j=1}^N  \left[-\Delta_{x_j} - \sum_{i=1}^M \frac{Z_i}{|x_j - R_i|}\right] + \sum_{i<j}^N \frac{1}{|x_i - x_j|} + \sum_{i<j}^M \frac{Z_i Z_j}{|R_i - R_j|} .
\end{equation}
Notice that $H_{N}$ acts on the subspace $L^2_a (\bR^{3N})$ of $L^2 (\bR^{3N})$ consisting of all functions that are antisymmetric with respect to permutations of the $N$ electrons. (Of course $E (N)$ and $H_{N}$ also depend on $M$, on the charges $Z_1, \dots , Z_M$ and on the positions $R_1, \dots , R_M$). Observe that the last term on the r.\,h.\,s.\ of (\ref{eq:ham-TF}) is just a constant representing the interaction among the nuclei. Already for $N \simeq 20$ it is extremely difficult to compute the ground state energy $E (N)$ numerically since the eigenvalue equation one has to solve is a partial differential equation in $3N$ coupled variables.

Thomas and Fermi postulated already in the early stages of quantum mechanics at the end of the 1920s that the ground state energy $E (N)$ can be approximated by 
\begin{equation}\label{eq:TF-min} E (N) \simeq E_{TF} (N) = \inf_{\rho \geq 0, \| \rho \|_1 = N} \cE_{TF} (\rho) \end{equation}
with the Thomas-Fermi functional 
\[ \cE_{TF} (\rho) =  \frac{3}{5} c_{TF} \int \rho^{5/3} (x) dx - \sum_{i=1}^M Z_i \int \frac{\rho (x)}{|x - R_i|} dx + \frac{1}{2} \int \frac{\rho (x) \rho (y)}{|x-y|} dx dy + \sum_{i<j}^M \frac{Z_i Z_j}{|R_i - R_j|}. \]
Notice that on the r.\,h.\,s.\ of (\ref{eq:TF-min}) we are looking for a function $\rho \in L^1 (\bR^3)$; as a consequence, in terms of numerical computations, the minimization problem (\ref{eq:TF-min}) is much simpler than the original problem (\ref{eq:ham-min}) despite the fact that the resulting Euler-Lagrange equations are nonlinear. 

In \cite{LSi} Lieb and Simon proved that the approximation (\ref{eq:TF-min}) becomes exact in the limit of large $N$. More precisely, they showed that, while $E(N)$ and $E_{TF} (N)$ are both proportional to $N^{7/3}$, 
\[ |E (N) - E_{TF} (N)| \leq C N^{7/3-1/30} \]
for an appropriate constant $C>0$. This gives a mathematically rigorous derivation of the Thomas-Fermi theory, and it tells us how big $N$ should be in order for $E_{TF} (N)$ to be a good approximation of the true quantum energy (later, better bounds have been obtained \cite{Hu, SiWe1, SiWe2}: one knows that the error with respect to Thomas-Fermi theory is of the order $N^{2}$ in the limit of large $N$). 

\medskip

\emph{Kinetic theory of dilute gases.} Consider now a gas of $N$ classical particles moving according to Newton's equations 
\begin{equation}\label{eq:new-bolt} \begin{split} \dot{x}_j (t) &= v_j (t) \\ \dot{v}_j (t)  &= - \sum_{i \not = j}^N \nabla V (x_i (t) - x_j (t)) \end{split} \end{equation}
for $j=1, \dots, N$. Here $V$ is a short range (compactly supported), regular potential. Eq.\ (\ref{eq:new-bolt}) is a system of $6N$ coupled ordinary differential equations. Given appropriate initial data, it is known to have a unique solution for all $t \in \bR$. 
However, since the number of particles $N$ is typically extremely large, it is almost impossible to deduce from (\ref{eq:new-bolt}) interesting qualitative or quantitative properties of the solution. 

At the beginning of the twentieth century Boltzmann, based on clever heuristic arguments, proposed to describe the dynamics of the gas by the nonlinear equation 
\begin{equation}\label{eq:boltz} \partial_t f_t (x,v) + v \cdot \nabla_x f_t (x,v) = \int dv' \int_{S^2} d\omega \, B (v-v' ; \omega) \left( f_t (x,v_{\text{out}}) f_t (x, v'_{\text{out}}) - f_t (x,v) f_t (x,v') \right) 
\end{equation}
for the phase-space density $f_t (x,v)$, which should measure the number of particles at time $t$ that are located close to $x \in \bR^3$ and have velocity close to $v \in \bR^3$. Here 
\[ \begin{split} v_{\text{out}} &= v + \omega \cdot (v'-v) \omega \\
v'_{\text{out}} &= v'-\omega \cdot (v'-v) \omega \end{split} \]
are the velocity of two particles emerging from the collision of two particles with velocities $v,v'$ and collision vector $\omega \in S^2$. Boltzmann's equation is a partial differential equation in only six variables, and is therefore much more accessible to computations than (\ref{eq:new-bolt}). 

For several years after the work of Boltzmann, his equation was not accepted by the physics community. In contrast to Newton's equations, Boltzmann's equation (\ref{eq:boltz}) is not time-reversal invariant. It took a while to understand that this fact does not contradict the validity of Boltzmann's equation, but only restricts the set of initial data for which (\ref{eq:boltz}) is a good approximation.

In \cite{Gr} Grad realized that Boltzmann's equation becomes correct in the low-density limit, where the density of particles $\rho$ is very small, $N$ is very large, and $N \rho^2$ is fixed of order one (the low-density or Boltzmann-Grad limit). Finally Lanford proved in \cite{Lan} that indeed in this limit Boltzmann's equation can be rigorously derived from Newton's equations, at least for sufficiently short times (recently \cite{PSS,GST} have extended the validity of (\ref{eq:boltz}) to a larger class of interaction potentials). 

\medskip

\emph{Plan of the notes.} Notice that the first example we discussed above (the Thomas-Fermi theory for large atoms and molecules) is based on many-body quantum mechanics while the second one (kinetic theory of gases in the low density limit) is based on classical Newtonian mechanics. Observe, moreover, that in the first example we were interested in an equilibrium property of the system (its ground state energy) while in the second case we considered a non-equilibrium problem (the time-evolution of the gas). 

In these notes we will focus on the derivation of time-dependent effective theories (non-equilibrium question) approximating many-body quantum dynamics. 
In the rest of this section, we will briefly recall the main properties of many-body quantum systems and their time evolution. In Section \ref{sec:mf} we will then introduce the mean-field regime for bosonic systems and we will explain how the many-body dynamics can be approximated by the Hartree equation in this limit. In Section \ref{sec:cohe} we will present a method, based on the use of coherent states and inspired by \cite{Hepp,GV}, to rigorously prove the convergence towards the Hartree dynamics. The fluctuations around the Hartree equation will be considered in Section \ref{sec:fluc-bos}. In Section \ref{sec:GP} we will discuss a more subtle regime, in which the many-body evolution can be approximated by the nonlinear Gross-Pitaevskii equation. In Section  \ref{sec:mf_fermions} we will discuss fermionic systems (characterized by antisymmetric wave functions). The fermionic mean-field regime is naturally linked with a semiclassical regime, and we will prove that the evolution of approximate Slater determinants can be approximated by the nonlinear Hartree-Fock equation.  Finally, in Section \ref{sec:mixed}, we will consider the same fermionic mean-field regime but this time we will focus on mixed quasi-free initial data approximating thermal states at positive temperature. 

\medskip

\emph{Wave functions and observables.} We will describe quantum systems of $N$ particles in three dimensions through a complex-valued wave function $\psi_N \in L^2 (\bR^{3N}, dx_1, \dots dx_N)$ with $\| \psi_N \|_2 = 1$. The arguments $(x_1, \dots , x_N) \in \bR^{3N}$ of the wave function $\psi_N$ describe the position of the $N$ particles. Since it does not play an important role in our analysis, we always neglect the spin of the particles.

Observables of the quantum system are associated with self-adjoint operators on the Hilbert space $L^2 (\bR^{3N}, dx_1 \dots dx_N)$. If $A$ is such an operator, the spectral theorem allows us to write it as 
\[ A = \int \lambda\, dE_A (\lambda), \]
where $dE_A (\lambda )$ is the spectral measure associated with $A$. Assuming for simplicity that $A$ has purely discrete spectrum, we find $A = \sum_j \lambda_j |\ph_j \rangle \langle \ph_j|$, where $\lambda_j$ are the eigenvalues and $\ph_j$ the eigenvectors of $A$. The expectation of the observable $A$ in the state $\psi_N$ is then given by the inner product 
\[ \langle \psi_N, A \psi_N \rangle = \sum_{j} \lambda_j |\langle \ph_j, \psi \rangle|^2 .\]
This expression for the expectation of $A$ leads to the following interpretation: the eigenvalues $\lambda_j$ are the possible outcomes of a measurement of $A$ while $|\langle \ph_j, \psi \rangle|^2$ is the probability that a measurement of $A$ produces the outcome $\lambda_j$. The positions of the $N$ particles are associated with multiplication operators; hence $|\psi_N (x_1, \dots , x_N)|^2$ is the probability density for finding particles close to $(x_1, \dots , x_N) \in \bR^{3N}$. The momenta of the particles are associated with differential operators $p_j = -i \nabla_{x_j}$. Hence, if $\widehat{\psi}_N$ denotes the Fourier transform of $\psi_N$, $|\widehat{\psi}_N (p_1, \dots , p_N)|^2$ is the probability density for finding particles with momenta close to $(p_1,\dots , p_N) \in \bR^{3N}$. The fact that $\psi_N$ determines the probability distribution of all observables of the system is an important feature of quantum mechanics (it leads for example to Heisenberg's uncertainty principle).

\medskip

\emph{Statistics.} In the following we will study systems of $N$ indistinguishable particles. In this case, there are important restrictions on the behavior of the wave function with respect to permutations. There are two different types of particles. Bosons are characterized by wave functions which are symmetric with respect to permutation, i.\,e.\ $\psi_N (x_{\pi(1)}, \dots, x_{\pi(N)}) = \psi_N (x_1, \dots , x_N)$ for all permutations $\pi \in S_N$. Fermions, on the other hand, are characterized by antisymmetric wave functions, i.\,e.\ $\psi_N (x_{\pi(1)} , \dots , x_{\pi(N)}) = \sigma_\pi \psi_N (x_1, \dots , x_N)$, where $\sigma_\pi$ is the sign of the permutation $\pi \in S_N$\footnote{Bosonic and fermionic wave functions are the only one-dimensional irreducible representations of the permutation group $S_N$; an irreducible representation of $S_N$ with dimension larger than one would lead to redundancies in the description of the system and is therefore meaningless, from the point of view of physics.}. In these notes, we will consider both the time-evolution of bosonic and fermionic systems; as we will see, the different behavior with respect to permutations has important consequences on the dynamics. 

\medskip

\emph{Dynamics.} The time evolution of quantum systems is governed by the Schr\"odinger equation 
\begin{equation}\label{eq:schr0}
i\partial_t \psi_{N,t} = H_N \psi_{N,t} 
\end{equation}
for the $N$-particle wave function $\psi_{N,t}$ (the subscript $t$ indicates the time dependence of $\psi_{N,t}$). On the r.\,h.\,s.\ of (\ref{eq:schr0}), $H_N$ is a self-adjoint operator on $L^2 (\bR^{3N}, dx_1 \dots dx_N)$ known as the Hamilton operator (or Hamiltonian) of the system. In these notes, we will consider Hamilton operators with two-body interactions, having the form
\begin{equation}\label{eq:HN} H_N = \sum_{j=1}^N \left( -\Delta_{x_j} + V_{\text{ext}} (x_j) \right) + \lambda \sum_{i<j}^N V(x_i - x_j) \end{equation}
for appropriate functions $V_{\text{ext}}$ (the external potential) and $V$ (the interaction potential) and for a coupling constant $\lambda \in \bR$ (which we introduce here for later convenience). The sum of the Laplace operators is the kinetic part of the Hamiltonian and generates the evolution of free particles. 

The Schr\"odinger equation (\ref{eq:schr0}) is linear and can always be solved by $\psi_{N,t} = e^{-i H_N t} \psi_{N,0}$, where $e^{-iH_N t}$ denotes the unitary one-parameter group generated by the self-adjoint operator $H_N$. This implies that the well-posedness of the Schr\"odinger equation is not an issue (although sometimes it can be difficult to prove the self-adjointness of the Hamilton operator). Nevertheless, for $N \gg 1$, it is essentially impossible to extract useful qualitative and quantitative information from (\ref{eq:schr0}) that goes beyond the existence and uniqueness of solutions. According to the philosophy outlined above, we look for simpler effective equations which approximate (\ref{eq:schr0}) in interesting regimes. 

\medskip

\emph{Acknowledgments.} We would like to thank Manfred Salmhofer and Christoph Kopper for organizing the summer school ``Mathematical Physics, Analysis and Stochastics'' and for encouraging us to write up these notes. The work of M.\ Porta and of B.\ Schlein has been supported by the ERC grant MAQD-240518. N. Benedikter has been partially supported by the ERC grant CoMBoS-239694 and by the ERC Advanced grant 321029.

\section{Mean-field regime for bosonic systems}
\setcounter{equation}{0}
\label{sec:mf}

One of simplest non-trivial regimes in which it is possible to approximate the many-body dynamics by an effective equation is the mean-field limit for bosonic systems. In the mean-field regime, particles experience a large number of weak collisions, whose cumulative effect can be approximated by an average mean-field potential. To realize the mean-field regime, we consider a system of $N$ bosons with a Hamilton operator of the form (\ref{eq:HN}), in the limit of large $N$ and small coupling constant $\lambda$, with $N \lambda$ of order one. This last condition guarantees that the total force on each particle is of order one, and therefore comparable with the inertia. In other words, we consider the dynamics generated by the mean-field Hamiltonian 
\begin{equation}\label{eq:HNmf} H_N = \sum_{j=1}^N \left( -\Delta_{x_j} + V_{\text{ext}} (x_j) \right) + \frac{1}{N} \sum_{i<j}^N V(x_i -x_j) \end{equation}
acting on the Hilbert space $L^2_s (\bR^{3N})$, the subspace of $L^2 (\bR^{3N})$ consisting of permutation symmetric functions, in the limit of large $N$. 

\medskip

\emph{Initial data.} The choice of the initial data is dictated by physics. In typical experiments a Bose gas is initially trapped by an external confining potential. To study the dynamics of the gas out of equilibrium, we consider the reaction of the system to a change of the external field. In other words, we are going to consider initial data given by equilibrium states of a Hamiltonian of the form (\ref{eq:HNmf}), with $V_{\text{ext}}$ modeling the external traps. In particular, at zero temperature, we are interested in initial data close to the ground state of (\ref{eq:HNmf}). Under appropriate assumptions on the interaction potential $V$ it is known that the ground state $\psi_{N}^{\text{gs}}$ of (\ref{eq:HNmf}) can be approximated, in the limit of large $N$, by a factorized wave function; i.\,e.\ $\psi_N^{\text{gs}} \simeq \ph^{\otimes N}$, for an appropriate $\ph \in L^2 (\bR^3)$ ($\ph$ is the minimizer of the Hartree energy functional). We are interested, therefore, in the solution of the Schr\"odinger equation 
\begin{equation}\label{eq:schr-mf} i\partial_t \psi_{N,t} = H_N \psi_{N,t} \end{equation}
for approximately factorized initial data $\psi_{N,0} \simeq \ph^{\otimes N}$. Note that the external potential in the operator $H_N$ appearing on the r.\,h.\,s.\ of (\ref{eq:schr-mf}) is typically different from the external potential in the trapping Hamiltonian, whose ground state is approximated by $\ph^{\otimes N}$ (otherwise, the dynamics would be trivial). For example the initial data would be taken as an approximation to the ground state in a harmonic trap, and the evolution of this initial data would be studied after switching off the trap, i.\,e.\ $V_{\text{ext}} = 0$.

\medskip

\emph{Hartree equation.} Of course, since $H_N$ is an interacting Hamiltonian, the dynamics does not preserve the factorization of the many-body wave function. Still, since the interaction is weak, we can expect factorization to be approximately (in a sense to be specified later) preserved, in the limit of large $N$. In other words, we can expect that for $N \gg 1$
\begin{equation}\label{eq:fact-t} \psi_{N,t}  (x_1, \dots , x_N ) \simeq \prod_{j=1}^N \ph_t (x_j) \end{equation}
for an evolved one-particle wave function $\ph_t$. Assuming (\ref{eq:fact-t}), it is easy to derive a self-consistent equation for the evolution of the one-particle wave function $\ph_t$. In fact, factorization of the $N$-particle wave function means, in probabilistic terms, that the particles are distributed in space according to the density $|\ph|^2$, independently of each other. The law of large numbers then suggests that the total potential experienced, say, by the $j$-th particle can be approximated by
\[ \frac{1}{N} \sum_{i \not = j} V (x_i -x_j) \simeq \int V(x_j - y) |\ph_t (y)|^2 dy = (V*|\ph|^2 ) (x_j). \]
Hence, $\ph_t$ must satisfy the Hartree equation
\begin{equation}\label{eq:hartree0} i\partial_t \ph_t = (- \Delta + V_{\text{ext}}) \ph_t + (V*|\ph_t|^2) \ph_t 
\end{equation}
where the many-body interaction has been replaced by the effective one-particle potential $V*|\ph_t|^2$, making $\lvert \varphi_t\rvert^2$ a nonlinear equation.

\medskip

\emph{Reduced densities.}  To explain in which sense we can expect (\ref{eq:fact-t}) to hold true, we introduce the notion of reduced density matrices (also known as reduced densities). The one-particle reduced density associated with the wave function $\psi_{N,t}$ is defined by
\[ \gamma^{(1)}_{N,t} = N\, \tr_{2,3,\dots , N} |\psi_{N,t} \rangle \langle \psi_{N,t} |, \]
where $|\psi_{N,t} \rangle \langle \psi_{N,t}|$ denotes the orthogonal projection onto $\psi_{N,t}$ and $\tr_{2, \dots , N}$ is the partial trace over the last $(N-1)$ particles. In other words, the one-particle reduced density $\gamma_{N,t}^{(1)}$ is defined as the non-negative trace-class operator on $L^2 (\bR^3)$ with integral kernel
\[ \gamma^{(1)}_{N,t} (x;y) = N \int dx_2 \dots dx_N \, \psi_{N,t} (x,x_2, \dots , x_N) \overline{\psi}_{N,t} (y,x_2, \dots , x_N). \]
Notice that we chose the normalization $\tr \, \gamma^{(1)}_{N,t} = N$.

Analogously, for $k=2, 3, \dots , N$, we can define the $k$-particle reduced density associated with $\psi_{N,t}$ by
\[ \gamma^{(k)}_{N,t} = \binom{N}{k} \, \tr_{k+1, \dots , N} \lvert\psi_{N,t} \rangle \langle \psi_{N,t} \rvert .\]
The integral kernel of the $k$-particle density matrix is given by
\begin{equation}\label{eq:red-k} \begin{split}\gamma^{(k)}_{N,t} &(x_1, \dots , x_k ; y_1, \dots , y_k) \\ = &\binom{N}{k} \int dx_{k+1} \dots dx_N \, \psi_{N,t} (x_1, \dots, x_k , x_{k+1} , \dots , x_N) \overline{\psi}_{N,t} (y_1, \dots , y_k, x_{k+1}, \dots,  x_N). \end{split} \end{equation}
The normalization is such that $\tr\, \gamma^{(k)}_{N,t} = \binom{N}{k}$. 

Clearly, for $k < N$, the $k$-particle reduced density $\gamma^{(k)}_{N,t}$ does not contain the full information about the system. Still, $\gamma^{(k)}_{N,t}$ is enough to compute the expectation of any $k$-particle observable: Let $J^{(1)}$ be an operator on the one-particle space $L^2 (\bR^3)$, and denote by $J^{(1)}_i = 1 \otimes \dots  \otimes J^{(1)} \otimes \dots \otimes 1$ the operator on $L^2 (\bR^{3N})$ acting like $J^{(1)}$ on the $i$-th particle and trivially on the other $(N-1)$ particles. We write $d\Gamma (J^{(1)}) = \sum_{i=1}^N J^{(1)}_i$. Then
\[ \langle \psi_{N,t} , d\Gamma (J^{(1)}) \psi_{N,t} \rangle = \sum_{i=1}^N \langle \psi_{N,t} , J^{(1)}_i \psi_{N,t} \rangle = N \tr J^{(1)}_i |\psi_{N,t} \rangle \langle \psi_{N,t} | = \tr \, J^{(1)} \, \gamma^{(1)}_{N,t}. \]
Similarly, if $J^{(k)}$ is an operator on the $k$-particle space $L^2 (\bR^{3k})$ and if we denote by $J^{(k)}_{i_1, \dots, i_k}$ the operator acting like $J^{(k)}$ on the $k$ particles $i_1, \dots, i_k$, we have
\[ \langle \psi_{N,t} , \sum_{\{i_1, \dots , i_k\}} J^{(k)}_{i_1, \dots , i_k} \psi_{N,t} \rangle = \tr \, J^{(k)} \, \gamma^{(k)}_{N,t}  \]
where the sum on the l.\,h.\,s.\ runs over all sets of $k$ different indices $\{ i_1, \dots , i_k \}$ chosen among $\{ 1, \dots ,N \}$.

\medskip

\emph{Convergence of reduced densities.} It turns out that reduced densities provide the appropriate language to describe the convergence of the many-body quantum evolution towards the Hartree dynamics. Under appropriate assumptions on the external potential $V_{\text{ext}}$ and, more importantly, on the interaction potential $V$, one can show the  convergence of the reduced densities associated with the solution $\psi_{N,t}$ of the Schr\"odinger equation (\ref{eq:schr-mf}) towards orthogonal projections onto products of solutions of the Hartree equation (\ref{eq:hartree0}). More precisely, consider a sequence of wave functions $\psi_N \in L_s^2 (\bR^{3N})$ with reduced density $\gamma^{(1)}_N$ satisfying  
\begin{equation}\label{eq:conv-t0} \frac{1}{N} \gamma^{(1)}_N \to |\ph \rangle \langle \ph| \quad (N \to \infty) \end{equation}
for a $\ph \in L^2 (\bR^3)$. Let $\psi_{N,t} = e^{-i H_N t} \psi_N$ be the solution of the Schr\"odinger equation (\ref{eq:schr-mf}) with initial data $\psi_N$. Then we expect, and under appropriate assumptions on $V_{\text{ext}},V$ we can show that 
\begin{equation}\label{eq:conv-mf} \frac{1}{N} \gamma^{(1)}_{N,t} \to |\ph_t \rangle \langle \ph_t| \quad (N \to \infty).\end{equation}
Here $\ph_t$ denotes the solution of the Hartree equation (\ref{eq:hartree0}) with initial data $\ph_0 = \ph$. The convergence in (\ref{eq:conv-mf}) can be understood in the trace-class topology. In fact, since the limit is a rank-one projection, weak convergence implies convergence in the trace norm\footnote{First, by testing the difference $N^{-1} \gamma^{(1)}_{N,t} - |\ph_t \rangle \langle \ph_t|$ against $|\ph_t \rangle \langle \ph_t|$ it implies Hilbert-Schmidt convergence. Then, since $|\ph_t \rangle \langle \ph_t|$ is a rank-one  projection, the operator $N^{-1} \gamma^{(1)}_{N,t} - |\ph_t \rangle \langle \ph_t|$ has exactly one negative eigenvalue. (If there were two linearly independent eigenvectors $\xi_1,\xi_2$ with negative eigenvalue, one could find a linear combination $\xi$ such that $\langle \xi, \gamma^{(1)}_{N,t}\xi\rangle < 0$.) Since $\tr\, \gamma^{(1)}_{N,t} - |\ph_t \rangle \langle \ph_t| = 0$, its absolute value is equal to the sum of all positive eigenvalues, and therefore the trace norm is twice the operator norm.}. Moreover, notice that convergence of the one-particle reduced density towards a rank-one orthogonal projection also implies convergence of higher order reduced densities in the limit $N \to \infty$ (the argument is outlined in \cite{LS}, after Theorem 1), 
\begin{equation}\label{eq:gk} \frac{1}{\binom{N}{k}} \gamma^{(k)}_{N,t} \to |\ph_t \rangle \langle \ph_t|^{\otimes k}. \end{equation}
The convergence here (and in \eqref{eq:conv-mf}) is for fixed $t$ and $k$.
Eq.\ (\ref{eq:conv-mf}) and (\ref{eq:gk}) explain in which sense one should understand the approximate factorization (\ref{eq:fact-t}). (For any operator $A$ we use the notation $A^{\otimes k} = \bigotimes_{i=1}^k A$ for its $k$-fold tensor product.)

\medskip

\emph{Some results.} The first rigorous results of the form (\ref{eq:conv-mf}) have been obtained for smooth interactions by Hepp \cite{Hepp} and for singular potentials by Ginibre and Velo \cite{GV}. For bounded potential, a proof of the convergence (\ref{eq:conv-mf}), based on the analysis of the so-called BBGKY hierarchy, has been given by Spohn \cite{Sp} (see next paragraph). The BBGKY technique has been later extended to potentials with Coulomb singularity in \cite{EY,BGM} and, for semi-relativistic bosons, in \cite{ES}. A new approach giving a precise estimate on the rate of the convergence towards the Hartree dynamics has been developed in \cite{RS} for potentials with Coulomb singularities (and further improved in \cite{CLS}); this approach, based on the ideas of \cite{Hepp,GV}, will be presented in Section \ref{sec:cohe}. Other bounds on the rate of convergence towards the Hartree evolution have been established in \cite{KP} (based on ideas proposed in \cite{P}) and recently in \cite{AFP}. In \cite{FKS,FKP} the convergence of the many-body evolution has been interpreted as a Egorov-type theorem. In \cite{AN1,AN2} the authors study the propagation of the Wigner measure in the bosonic mean-field limit. Next order corrections to the Hartree dynamics have been considered in \cite{GMM1,GMM2}, leading to a better approximation of the many-body evolution. (A related problem is the study of the fluctuations around the Hartree evolution, which will be discussed in Section \ref{sec:cohe}.) Instead of a fixed interaction $V$, it is also interesting to consider $N$-dependent potentials, scaling like $V_N (x) = N^{3\alpha} V(N^\alpha x)$ (in the three dimensional case) and converging towards a delta-function in the limit of large $N$. In this case (assuming $\alpha < 1$; for $\alpha = 1$, one recovers instead the Gross-Pitaevskii regime, which will also be discussed in Section \ref{sec:cohe}), the many-body evolution can be approximated by a nonlinear Schr\"odinger equation with a local cubic nonlinearity. Results in this direction have been obtained in \cite{AGT,AB,CH2} in the one-dimensional setting, in \cite{KSS} in the two-dimensional case and in \cite{ESY1} in three dimensions. It is also possible to start from Hamiltonians with three-body interactions; in this case, the evolution can be approximated by a quintic nonlinear Schr\"odinger equation; see \cite{CP,XC}.

\medskip

\emph{The BBGKY approach.} The main idea of the BBGKY approach, which was first applied to many-body quantum systems in the mean-field regime in \cite{Sp}, is to study directly the evolution of the reduced densities defined in (\ref{eq:red-k}). To explain this idea it is convenient to normalize the reduced densities associated with the solution $\psi_{N,t}$ of the Schr\"odinger equation, defining, for $k=1,\dots , N$, 
\[ \wt{\gamma}^{(k)}_{N,t} = \frac{1}{\binom{N}{k}} \gamma^{(k)}_{N,t}. \]
The new density matrices are normalized so that $\tr \, \wt\gamma^{(k)}_{N,t} = 1$ for all $N \in \bN$ and all $k=1, \dots, N$. {F}rom the Schr\"odinger equation for $\psi_{N,t}$ it is easy to derive evolution equations for the family $\{ \wt\gamma^{(k)} \}_{k=1}^N$. It turns out that the evolution of the $N$ reduced densities is governed by a hierarchy\footnote{i.\,e.\ the equation for $\wt{\gamma}^{(k)}_{N,t}$ depends on $\wt{\gamma}^{(k+1)}_{N,t}$} of $N$ coupled equations, known as the BBGKY hierarchy (BBGKY stands for Bogoliubov-Born-Green-Kirkwood-Yvon):
\begin{equation}\label{eq:BBGKY}
\begin{split} 
i\partial_t \wt\gamma^{(k)}_{N,t} = \; & \sum_{j=1}^k \left[ -\Delta_{x_j} , \wt\gamma^{(k)}_{N,t} \right]  + \frac{1}{N} \sum_{i<j}^k \left[ V(x_i -x_j), \wt\gamma^{(k)}_{N,t} \right] \\ &+ \frac{N-k}{N} \sum_{j=1}^k \tr_{k+1} \left[ V(x_j - x_{k+1}), \wt\gamma^{(k+1)}_{N,t} \right] \end{split} \end{equation}
where we use the convention that $\wt\gamma^{(N+1)}_{N,t} = 0$ and where $\tr_{k+1}$ denotes the partial trace over the degrees of freedom of the $(k+1)$-st particle. Here we also introduced the commutator, defined by $[A,B] = AB-BA$ on a domain suited to the operators $A$ and $B$. Notice that the second term on the r.\,h.\,s.\ of (\ref{eq:BBGKY}) describes the interaction among the first $k$ particles in the system. The last term on the r.\,h.\,s.\ of (\ref{eq:BBGKY}), on the other hand, corresponds to the interaction of these $k$ particles with the other $(N-k)$ particles. At least formally, the BBGKY hierarchy (\ref{eq:BBGKY}) converges, in the limit of large $N$, towards the infinite hierarchy 
\begin{equation}\label{eq:inf-hier} i\partial_t \wt\gamma^{(k)}_{\infty,t} = \; \sum_{j=1}^k \left[ -\Delta_{x_j} , \wt\gamma^{(k)}_{\infty,t} \right] + \sum_{j=1}^k \tr_{k+1} \left[ V(x_j - x_{k+1}), \wt\gamma^{(k+1)}_{\infty,t} \right] \, .
\end{equation}
It is simple to check that this infinite hierarchy has a factorized solutions $\wt\gamma_{\infty,t}^{(k)} = |\ph_t \rangle \langle \ph_t|^{\otimes k}$, given by products of the solution of the Hartree equation (\ref{eq:hartree0}). This observation suggests a general strategy to show the convergence (\ref{eq:gk}) of the reduced densities towards projections onto products of solutions of the Hartree equation. The strategy consists of three steps:
\begin{itemize}
\item \emph{Compactness:} first, one needs to prove the compactness of the sequence (in $N$) of families $\Gamma_{N,t} = \{ \wt\gamma_{N,t}^{(k)} \}_{k=1}^N$ with respect to an appropriate weak topology. Compactness implies in particular the existence of at least one limit point $\Gamma_{\infty,t} = \{ \wt\gamma_{\infty,t}^{(k)} \}_{k \geq 1}$.
\item \emph{Convergence.} Secondly, one needs to characterize limit points of the sequence $\Gamma_{N,t}$ as solutions of the infinite hierarchy (\ref{eq:inf-hier}). In other words, one has to show that any limit point $\Gamma_{\infty,t}$ of the sequence $\Gamma_{N,t}$ satisfies (\ref{eq:inf-hier}).
\item \emph{Uniqueness.} Finally, one has to prove the uniqueness of the solution of the infinite hierarchy. This implies immediately that the sequence $\Gamma_{N,t}$ converges, since every compact sequence with at most one limit point converges. Moreover, since we know that the factorized densities $\wt\gamma_{\infty,t}^{(k)} = |\ph_t \rangle \langle \ph_t|^{\otimes k}$ are a solution of (\ref{eq:inf-hier}), uniqueness also implies that $\wt{\gamma}_{N,t} \to |\ph_t \rangle \langle \ph_t |^{\otimes k}$ for all $k \in \bN$ (the argument proves convergence with respect to the weak topology with respect to which one showed compactness in the first step; however, since the limit is a rank-one projection, weak convergence immediately implies convergence in the trace norm). 
\end{itemize}

The most difficult of the three steps is the proof of the uniqueness of the solution of the infinite hierarchy. Let us illustrate how to prove uniqueness in the case of a bounded interaction potential $V \in L^\infty (\bR^3)$. To this end, we rewrite the infinite hierarchy (\ref{eq:inf-hier}) in integral form as 
\begin{equation}\label{eq:wtg-int} \wt\gamma_{\infty,t}^{(k)} = \cU^{(k)} (t) \wt\gamma_{\infty,0}^{(k)} + \frac{1}{i} \int_0^t ds \, \cU^{(k)} (t-s) B^{(k)} \wt\gamma^{(k+1)}_{\infty,s} .
\end{equation}
Here we defined the action of the free evolution $\cU^{(k)} (t)$ on a $k$-particle density $\gamma^{(k)}$ by 
\[ \cU^{(k)} (t) \gamma^{(k)} = e^{-i \sum_{j=1}^k \Delta_{x_j} t} \gamma^{(k)} e^{i\sum_{j=1}^k \Delta_{x_j} t} \]
and the action of the collision operator $B^{(k)}$ on a $(k+1)$-particle density $\gamma^{(k+1)}$ by  
\begin{equation}\label{eq:coll} B^{(k)} \gamma^{(k+1)} = \sum_{j=1}^k \tr_{k+1} \left[ V(x_j - x_{k+1}) , \gamma^{(k+1)} \right].
\end{equation}
Notice that the collision operator maps $(k+1)$-particle density matrices into $k$-particle density matrices. It is important to observe how the free evolution and the collision operator affect the trace norm. On the one hand, we clearly have 
\begin{equation}\label{eq:Uk} \| \cU^{(k)} \gamma^{(k)} \|_{\text{tr}} = \| \gamma^{(k)} \|_{\text{tr}} \,.
\end{equation}
On the other hand, for bounded interaction potentials, we find
\begin{equation} \label{eq:Bk} \begin{split} \| B^{(k)} \gamma^{(k+1)} \|_{\text{tr}} &\leq \sum_{j=1}^k \left[ \tr \, \left| V(x_j -x_{k+1}) \gamma^{(k+1)} \right| + \tr \left| \gamma^{(k+1)} V(x_j - x_{k+1}) \right| \right]  \\ & \leq 2k \| V \|_\infty \tr | \gamma^{(k+1)} | \\ &\leq 2 k \| V \|_\infty \| \gamma^{(k+1)} \|_{\text{tr}} \, .\end{split} 
\end{equation}
Here we used that $\tr \lvert AB\rvert \leq \norm{A} \tr \lvert B \rvert$ for any bounded operator $A$ and any trace-class operator $B$. (In the same way $\tr \lvert AB\rvert \leq \norm{B} \tr \lvert A \rvert$ if $B$ is bounded and $A$ trace class.)

Iterating (\ref{eq:wtg-int}), we obtain the $n$-th order Dyson series
\begin{equation} \label{eq:wtg-dyson} 
\begin{split} \wt\gamma_{\infty,t}^{(k)} & = \cU^{(k)} (t) \wt\gamma_{\infty,0}^{(k)} \\
&\quad + \sum_{m=1}^{n-1} \frac{1}{i^m} \int_0^t ds_1 \ldots \int_0^{s_{m-1}} ds_m \\ &\hspace{2cm} \times  \cU^{(k)} (t-s_1) B^{(k)} \cdots \cU^{(k+m-1)}(s_{m-1} -s_m) B^{(k+m-1)} \cU^{(k+m)} (s_m) \wt\gamma^{(k+m)}_{\infty,0} \\
&\quad+ \frac{1}{i^n} \int_0^t ds_1 \ldots \int_0^{s_{n-1}} ds_n \\ &\hspace{2cm} \times  \cU^{(k)} (t-s_1) B^{(k)} \cU^{(k+1)} (s_1 - s_2) \cdots \cU^{(k+n-1)} (s_{n-1} -s_n) B^{(k+n-1)} \wt\gamma^{(k+n)}_{\infty,s_n} .
\end{split} \end{equation}
Suppose now that the families $\Gamma_{1,\infty,t} = \{ \wt\gamma_{1,\infty,t}^{(k)} \}_{k \geq 0}$ and $\Gamma_{2, \infty,t} = \{ \wt\gamma_{2, \infty,t}^{(k)} \}_{k\geq 0}$ are two solutions of the infinite hierarchy, written in the integral form (\ref{eq:wtg-int}), with the same initial data. Expanding both solutions as in (\ref{eq:wtg-dyson}), and taking the difference (so that all fully expanded terms cancel), we find
\[ \begin{split}
\tr \, \left| \wt\gamma_{1,\infty,t}^{(k)} - \wt\gamma_{2,\infty,t}^{(k)} \right| \leq \; &\int_0^t ds_1 \dots \int_0^{s_{n-1}} ds_n \, \| \cU^{(k)} (t-s_1) B^{(k)} \cdots B^{(k+n-1)} \wt\gamma_{1,\infty,s_n}^{(k+n)}  \|_{\text{tr}} \\ &+  \int_0^t ds_1 \dots \int_0^{s_{n-1}} ds_n \, \| \cU^{(k)} (t-s_1) B^{(k)} \cdots B^{(k+n-1)} \wt\gamma_{2,\infty,s_n}^{(k+n)}  \|_{\text{tr}}.
\end{split} \]
Applying iteratively the bounds (\ref{eq:Uk}) and (\ref{eq:Bk}), we conclude that
\[ \tr \, \left| \gamma_{1,\infty,t}^{(k)} - \gamma_{2,\infty,t}^{(k)} \right| \leq 2 \frac{|t|^n}{n!} k (k+1) \cdots (k+n-1) (2 \| V \|_\infty)^n \leq 2^k (4 \| V \|_\infty |t| )^n, \]
where we used the normalization $\tr \, \gamma^{(k+n)}_{1,\infty,t} = \tr \, \gamma^{(k+n)}_{2,\infty,t} = 1$ and the bound 
\[ \binom{n+k-1}{k-1} \leq 2^{k+n-1}. \]
For $|t| \leq (8 \| V \|_\infty)^{-1}$, we obtain that for all $n \in \bN$
\[\tr \, \left| \wt\gamma_{1,\infty,t}^{(k)} - \wt\gamma_{2,\infty,t}^{(k)} \right|  \leq 2^{k-n}.\]
Since the l.\,h.\,s.\ is independent of $n$, it must vanish. This proves that $\wt\gamma_{1,\infty,t}^{(k)} = \wt\gamma_{2,\infty,t}^{(k)}$ for all $|t| \leq (8 \| V \|_\infty)^{-1}$. Since this argument only uses the normalization $\tr \, \wt\gamma^{(k+n)}_{1,\infty,t} = \tr \, \wt\gamma^{(k+n)}_{2,\infty,t} = 1$, which hold for all $t \in \bR$, it can be iterated to prove uniqueness of the solution of the infinite hierarchy for all $t \in \bR$. 

For a Coulomb potential $V(x) = \pm 1/ |x|$, the proof we outlined above can be modified by introducing a different norm for density matrices (this approach was first used in \cite{EY}). For a $k$-particle density $\wt\gamma^{(k)}$, acting on $L^2 (\bR^{3k})$, we define the Sobolev-type norm 
\begin{equation}\label{eq:sob-norm} \| \wt\gamma^{(k)} \|_{H^{(k)}_1} = \tr \left| S_1 \dots S_k \wt\gamma^{(k)} S_k \dots S_1 \right| \end{equation}
where $S_j = (1-\Delta_{x_j})^{1/2}$. Since the Coulomb potential is bounded with respect to the kinetic energy in the sense that as operators
\[ \pm \frac{1}{|x|} \leq C (1-\Delta) \]
one can show that the collision operator $B^{(k)}$ defined as in (\ref{eq:coll}), but now with $V(x) = \pm 1/|x|$, satisfies
\[  \| B^{(k)} \wt\gamma^{(k+1)} \|_{H_1^{(k)}} \leq Ck \| \wt\gamma^{(k+1)} \|_{H_1^{(k+1)}} .\]
This bound replaces \eqref{eq:Bk}. Using this new bound one can prove, similarly as explained above for the case of bounded potentials, that any two solutions $\Gamma_{1,\infty,t} = \{ \wt\gamma^{(k)}_{1,\infty,t} \}_{k\geq 1}$ and $\Gamma_{2,\infty,t} = \{ \wt\gamma_{2,\infty,t}^{(k)} \}_{k \geq 1}$ of the infinite hierarchy with the same initial data satisfy 
\[  \left\|  \wt\gamma_{1,\infty,t}^{(k)} - \wt\gamma_{2,\infty,t}^{(k)} \right\|_{H_1^{(k)}}  \leq C^{k-n} \left[ \| \wt\gamma_{1,\infty,t}^{(k+n)} \|_{H_1^{(k+n)}} +\| \wt\gamma_{2,\infty,t}^{(k+n)} \|_{H_1^{(k+n)}} \right] . \]
To conclude uniqueness, here we need to show a-priori estimates of the form
\begin{equation}\label{eq:apri} \left\| \wt\gamma^{(k)}_{\infty,t} \right\|_{H^{(k)}_1} \leq C^k \end{equation}
for all $t \in \bR$ and $k \in \bN$, valid for any limit point $\Gamma_{\infty,t} = \{ \wt\gamma^{(k)}_{\infty,t} \}_{k \geq 1}$ of the sequence of families $\Gamma_{N,t} = \{ \wt\gamma_{N,t}^{(k)} \}_{k=1}^N$ of densities associated with the solution of the Schr\"odinger equation. (This step was not needed in the case of bounded potentials because the trace norm trivially remains uniformly bounded.) To obtain a-priori bounds of the form (\ref{eq:apri}), one can use energy conservation; for details, see \cite{EY}. 

\section{Coherent States Approach}
\setcounter{equation}{0}
\label{sec:cohe}

In this section, we illustrate the method developed in \cite{RS}, based on the original ideas of \cite{Hepp,GV}, to prove the convergence of the many-body evolution towards the Hartree dynamics. This method is based on a representation of the many-boson system on the Fock space and on the study of the time-evolution of a special class of initial data, known as coherent states. With respect to the BBGKY approach that was presented in the last section, the analysis of the evolution of coherent states allows us to obtain precise bounds on the rate of convergence towards the Hartree dynamics. This is an important point, since real systems have a large but, of course, finite number of particles. Bounds on the rate of convergence are therefore crucial to establish whether the Hartree equation is a good approximation to the dynamics of a given boson gas. 

\medskip

\emph{Fock space.} The bosonic Fock space over $L^2 (\bR^3)$ is defined as the direct sum
\[ \cF = \bigoplus_{n \geq 0} L^2_s (\bR^{3n}), \]
where $L^2_s (\bR^{3n})$ denotes the subspace of $L^2 (\bR^{3n})$ consisting of all permutation symmetric functions. On $\cF$, we can define the inner product
\[ \langle \Psi, \Ph \rangle_\cF = \sum_{n\geq 0} \langle \psi^{(n)} , \ph^{(n)} \rangle \]
We denote by $\| . \|_{\cF}$ the corresponding norm. We always use vectors $\Psi \in \cF$ normalized such that \[ \| \Psi \|^2_\cF = \sum_{n \geq 0 } \| \psi^{(n)} \|_2^2 = 1.\]  
The Fock space allows us to describe states of the system where the number of particles is not fixed. The vector $\Psi = \{ \psi^{(n)} \}_{n\geq 0} \in \cF$ describes a state with probability $\| \psi^{(n)} \|_2^2$ has $n$ particles. In particular, states with exactly $N$ particles are embedded in the Fock space; they are described by vectors $\{ 0, 0, \dots,  \psi_N , 0, \dots \} \in \cF$ having only one non-zero component. These vectors are eigenvectors of the number of particles operator $\cN$ defined by
\[ (\cN \Psi)^{(n)} = n \psi^{(n)} \]
for any $\Psi = \{ \psi^{(n)} \}_{n \geq 0}$ such that 
\[ \sum_{n \geq 0} n^2 \| \psi^{(n)} \|_2^2 < \infty \, . \]
The vacuum vector $\Omega = \{ 1, 0 , \dots \} \in \cF$ plays a special role; it is an eigenvector of $\cN$ with eigenvalue zero and describes a state with no particles at all. 

\medskip

\emph{Creation and annihilation operators.} For any one-particle wave-function $f \in L^2 (\bR^3)$ we define the creation operator $a^* (f)$ and the annihilation operator $a(f)$ by setting
\[ \begin{split} 
(a^* (f) \Psi)^{(n)} (x_1, \dots , x_n) &= \frac{1}{\sqrt{n}} \sum_{j=1}^n f(x_j) \psi^{(n-1)} (x_1, \dots , x_{j-1}, x_{j+1}, \dots , x_n) \\
(a (f) \Psi)^{(n)} (x_1, \dots , x_n) &= \sqrt{n+1} \int dx \overline{f (x)} \psi^{(n+1)} (x, x_1, \dots , x_n). \end{split} \]
The interpretation is straightforward: $a^* (f)$ creates a new particle with wave function $f$, while $a(f)$ annihilates such a particle. Creation and annihilation operators are closed densily defined operators on $\cF$; moreover, $a^* (f)$ is the adjoint of $a(f)$ (as the notation suggests). Notice also that $a^* (f)$ is linear in $f$, while $a(f)$ is antilinear. Creation and annihilation operators satisfy canonical commutation relations
\begin{equation}\label{eq:ccr} [a(f) , a^* (g) ] = \langle f,g \rangle , \quad [a(f), a(g)] = [ a^* (f) , a^* (g) ] = 0.
\end{equation}
Despite the bosonic creation and annihilation operators being unbounded, usually domain questions are unproblematic since one can take the domain of a sufficiently large power of the number operator to easily make sense of most expressions and in particular of their commutators --- see \eqref{eq:aa-bd}. (One has to be more careful in this respect when working with the operator-valued distributions below.)
It is also useful to introduce operator-valued distributions $a_x^*$ and $a_x$ which formally create and annihilate a particle at 
position $x \in \bR^3$. It terms of these distributions 
\[ a^* (f) = \int f(x) a^*_x dx , \quad a(f) = \int \overline{f (x)} \, a_x dx \]
and 
\begin{equation}\label{eq:CCR} [ a_x, a_y^* ] = \delta (x- y) , \quad [ a_x , a_y ] = [a_x^* , a_y^*]=0 .
\end{equation}
Notice that, if $\Psi = \{ \psi^{(n)} \}_{n\in \bN}$, 
\begin{equation}\label{eq:axpsi} (a_x \Psi)^{(n)}  (x_1, \dots, x_n) = \sqrt{n+1} \, \psi^{(n+1)} (x, x_1, \dots , x_n). \end{equation}
Hence
\[ \begin{split} \langle \Psi, \cN \Ph \rangle &= \sum_{n\geq 0} n \int dx_1 \dots dx_n \overline{\phi}^{(n)} (x_1, \dots , x_n)  \ph^{(n)} (x_1, \dots, x_n) \\ &= \sum_{n\geq 1} \int dx dx_1 \dots dx_{n-1} \overline{(a_x \Psi)}^{(n-1)} (x_1,\dots , x_{n-1}) (a_x \Ph)^{(n-1)} (x_1, \dots , x_{n-1}) \\ &= \int dx \, \langle a_x \Psi, a_x \Ph \rangle \end{split} \]
and in this sense we can write 
\begin{equation}\label{eq:Naa} \cN = \int dx \, a_x^* a_x \, .
\end{equation}

This expression for $\cN$ suggests that, although creation and annihilation operators are unbounded operators, they can be bounded with respect to the square root of the number of particles operator. We have 
\begin{equation}
\begin{split}
\pm \langle \Psi, (a(f) + a^* (f)) \Psi \rangle  &\leq 2 | \langle \Psi, a(f) \Psi \rangle| \\ &\leq 2 \sum_{n\geq 0} \sqrt{n+1} \int dx dx_1 \dots dx_n \overline{\psi^{(n)} (x_1, \dots , x_n)} \, \overline{f(x)} \, \psi^{(n+1)} (x, x_1, \dots , x_n) \\ &\leq 2 \sum_{n \geq 0} \sqrt{n+1} \| f \| \| \psi^{(n)} \| \| \psi^{(n+1)} \| 
\end{split}
\end{equation}
and therefore in the sense of forms 
\begin{equation}\label{eq:aa-bd-form}
\pm  (a(f) + a^* (f)) \leq 4 \| f \| \cN^{1/2}.
\end{equation}
In norm, we have the bounds
\begin{equation} \label{eq:aa-bd} \begin{split} 
\| a(f) \Psi \|_{\cF} &\leq \| f \| \| \cN^{1/2} \Psi \|_{\cF} \quad \text{and } \\
\| a^* (f) \Psi \|_{\cF} &\leq \|f \| \| (\cN+1)^{1/2} \Psi \|_{\cF}
\end{split} 
\end{equation}
for any $f \in L^2 (\bR^3)$. To prove the first bound in (\ref{eq:aa-bd}), we observe that by the Cauchy-Schwarz inequality
\begin{equation}\label{eq:afpsi} \| a(f) \Psi \|_\cF \leq \int  dx \, |f(x)| \| a_x \Psi \|_{\cF} \leq \| f \| \left( \int dx \, \| a_x \Psi \|_{\cF}^2 \right)^{1/2} = \| f \| \| \cN^{1/2} \Psi \| .
\end{equation}
The second estimate in (\ref{eq:aa-bd}) follows from 
\[ \| a^* (f) \Psi \|_{\cF}^2 = \langle \Psi, a(f) a^* (f) \Psi \rangle = \langle \Psi , a^* (f) a(f) \Psi \rangle + \| f \|^2 \| \Psi \|_{\cF}^2 = \| a(f) \Psi \|^2 + \| f \|^2 \| \Psi \|^2 \] 
and (\ref{eq:afpsi}).  

Similarly to (\ref{eq:Naa}), we can express the second quantization of any one-particle operator in terms of the operator-valued distributions $a_x^*$ and $a_x$: Let $J^{(1)}$ be an operator on the one-particle space $L^2 (\bR^3)$. The second quantization of $J^{(1)}$ is the operator $d\Gamma (J^{(1)})$ on $\cF$ defined by the requirement that 
\[ (d\Gamma (J^{(1)}) \Psi)^{(n)} = \sum_{i=1}^n J^{(1)}_i \psi^{(n)} \] 
where $J^{(1)}_i$ denotes the operator acting on $L^2 (\bR^{3n})$ as $J^{(1)}$ on the $i$-th particle and as the identity on the other $(n-1)$ particles. If the one-particle operator $J^{(1)}$ has the integral kernel $J^{(1)} (x;y)$, we can write 
\[ \begin{split} \langle \Ph, d\Gamma (J^{(1)}) \Psi \rangle  &= \sum_{n \geq 1} \sum_{j=1}^n \langle \ph^{(n)}, J^{(1)}_j \psi^{(n)} \rangle \\ &= \sum_{n \geq 1} n \int dx dy dx_2 \dots dx_n \overline{\ph^{(n)} (x,x_2,\dots , x_n)} J^{(1)} (x;y) \psi^{(n)} (y, x_2, \dots, x_n) \\ &= \sum_{n \geq 1} \int dx dy dx_2 \dots dx_n\, J^{(1)} (x;y) \overline{(a_x \Ph)^{(n-1)} (x_2, \dots , x_n)} \,  (a_y \Psi)^{(n-1)} (x_2, \dots , x_n) \\
&= \int dx dy \, J^{(1)} (x;y) \, \langle a_x \Ph , a_y \Psi \rangle .
\end{split} \]
Thus
\begin{equation}
\label{eq:DGJ} 
d\Gamma (J^{(1)}) = \int dx dy \, J^{(1)} (x;y) a_x^* a_y. \end{equation}
Since the number of particles operator is the second quantization of the identity, $\cN = d\Gamma (1)$, the last expression is consistent with (\ref{eq:Naa}). Notice that if $J^{(1)}$ is bounded, then its second quantization $d\Gamma (J^{(1)})$ is (although generally unbounded) bounded with respect to the number of particles operator, i.e.
\begin{equation}\label{eq:dGJN} \lvert \langle \psi, d\Gamma (J^{(1)}) \psi\rangle \rvert \leq \| J^{(1)} \| \langle \psi, \cN\psi\rangle \quad \text{for all }\psi \in \cF. 
\end{equation}
and also 
\[ \| d\Gamma (J^{(1)}) \Psi \| \leq \| J^{(1)} \| \| \cN \Psi \| \, .  \]

We can extend (\ref{eq:DGJ}) to operators involving more particles. Let $J^{(k)}$ be an operator on the $k$-particle Hilbert space $L^2_{s} (\bR^{3k})$. We define the second quantization $d\Gamma (J^{(k)})$ of $J^{(k)}$ by
\begin{equation}\label{eq:dGJk} (d\Gamma (J^{(k)}) \Psi)^{(n)} = \sum_{ \{i_1, \dots , i_k\}} J^{(k)}_{i_1, \dots , i_k} \psi^{(n)} 
\end{equation}
where the sum runs over all sets $\{ i_1, \dots , i_k \}$ 
of $k$ different indices in $\{1,\dots , n\}$, and where $J^{(k)}_{i_1, \dots , i_k}$ denotes the operator on $L^2 (\bR^{3n})$ acting as $J^{(k)}$ on the particles $i_1, \dots, i_k$ and as the identity on the other $(n-k)$ particles. If $J^{(k)}$ has the integral kernel  $J^{(k)}(x_1, \dots , x_k ; y_1, \dots , y_k)$, we find, similarly to (\ref{eq:DGJ}), 
\begin{equation}\label{eq:DGJke} d\Gamma (J^{(k)}) = \int dx_1 \dots dx_k dy_1\dots dy_k \, J^{(k)}(x_1, \dots, x_k ; y_1, \dots, y_k) \, a_{x_1}^* \dots a_{x_k}^* a_{y_k} \dots a_{y_1} .
\end{equation}
Similarly to (\ref{eq:dGJN}), we find
\[ \lvert \langle \psi, d\Gamma (J^{(k)}) \psi\rangle\rvert \leq \| J^{(k)} \| \, \cN (\cN-1) \dots (\cN-k+1) .\]

\medskip

\emph{Reduced densities.} We define the one-particle reduced density associated with a normalized Fock space vector $\Psi \in \cF$ as the operator $\gamma^{(1)}\!:L^2(\bR^3) \to L^2(\bR^3)$ such that
\[ \tr \, J^{(1)} \gamma^{(1)} = \langle \Psi, d\Gamma (J^{(1)}) \Psi \rangle\]
for all one-particle observables $J^{(1)}$.
{F}rom (\ref{eq:DGJ}), we find that the integral kernel of the one-particle reduced density $\gamma^{(1)}$ is given by
\begin{equation}
\label{eq:g1xy} \gamma^{(1)} (x;y) = \langle \Psi , a_y^* a_x \Psi \rangle .
\end{equation}
We define the $k$-particle reduced density associated with the Fock space vector $\Psi\in \cF$ as the operator $\gamma^{(k)}$ on $L^2(\bR^{3k})$ such that 
\[ \tr\, J^{(k)} \gamma^{(k)} = \langle \Psi, d\Gamma (J^{(k)}) \Psi \rangle \]
for every $k$-particle observable $J^{(k)}$. Using (\ref{eq:DGJke}) we find that the integral kernel of the 
$k$-particle reduced density $\gamma^{(k)}$ is
\[ \gamma^{(k)} (x_1, \dots , x_k ; y_1, \dots , y_k) = \langle \Psi, a_{y_1}^* \dots a_{y_k}^* a_{x_k} \dots a_{x_1} \Psi \rangle.\] 
Notice that $\gamma^{(k)}$ is normalized such that $\tr \, \gamma^{(k)} = \langle \Psi, \cN (\cN-1) \dots (\cN-k+1) \Psi \rangle$. 

\medskip

\emph{Representation of the Hamiltonian in second quantization.} We define the Hamilton operator $\cH_N$ on the Fock space $\cF$ by $(\cH_N \Psi)^{(n)} = \cH_N^{(n)} \psi^{(n)}$ with
\begin{equation}\label{eq:HN1} \cH_N^{(n)} = \sum_{j=1}^n -\Delta_{x_j} + \frac{1}{N} \sum_{i<j}^n V(x_i -x_j). \end{equation}
Since it does not play an important role in our analysis, we neglect here the external potential (but we could also add it, requiring only straightforward changes). The operator $\cH_N$ leaves the number of particles invariant, i.\,e.\ it commutes with the number of particles operator $\cN$. When restricted to the $N$-particle subspace (consisting of the vectors of the form $\{ 0, 0, \dots, 0, \psi_N , 0, \dots \}$), $\cH_N$ coincides precisely with the Hamilton operator (\ref{eq:HNmf}) defined in the previous section (of course up to the external potential that we neglect here). Notice that the parameter $N$ appearing in the Hamiltonian (\ref{eq:HN1}) is not yet related with the number of particles in the system which in Fock space is arbitrary. Of course, later on we will restrict our attention to initial data in $\cF$ that has approximately $N$ particles; this is important to make sure that we are in the same mean-field regime discussed in Section \ref{sec:mf}.

We can write the Hamiltonian $\cH_N$ in terms of the operator-valued distributions $a_x$ and $a_x^*$. {F}rom (\ref{eq:axpsi}) we find that 
\begin{equation}\label{eq:HNaa} \cH_N = \int dx \, \nabla_x a_x^* \nabla_x a_x + \frac{1}{2N} \int dx dy V(x-y) a_x^* a_y^* a_y a_x \, . \end{equation}
The r.\,h.\,s.\ should be understood in the sense of forms, similarly to (\ref{eq:Naa}). The fact that $\cH_N$ commutes with the number of particles operator $\cN$ is evident from (\ref{eq:HNaa}), because, in every summand, the number of creation operators matches exactly the number of annihilation operators. In particular, it follows that 
\[ e^{-i \cH_N t} \{ 0, \dots , 0, \psi_N, 0, \dots \} = \{ 0, \dots, 0, e^{-i \cH^{(N)}_N t} \psi_N , 0, \dots \}.\]
On the r.\,h.\,s.\ we recover exactly the mean-field evolution discussed in the previous section. So what have we gained by switching to the Fock space representation of the bosonic system? The answer is that now we have more freedom in the choice of the initial data; in particular, we can consider initial data having a number of particles which is not fixed, i.\,e.\ is a superposition of different particle numbers. We will make use of this freedom by considering a special class of initial data known as coherent states.  

\medskip

\emph{Coherent states.} For $f \in L^2 (\bR^3)$, we introduce the Weyl operator 
\[ W(f) = e^{a(f) - a^* (f)}. \]
The coherent state with wave function $f \in \bR^3$ is defined as\footnote{To evaluate $W(f)\Omega$ we made use of the Baker-Campbell-Hausdorff formula for operators $A,B$ with the property that $[[A,B],A] = [[A,B],B]=0$:
\[e^{A+B} = e^{-\frac{1}{2}[A,B]} e^A e^B.\] Since $a(f)\Omega=0$ one then simply has to expand the exponential of $a^*(f)$.}
\begin{equation}\label{eq:coh-def} 
W(f) \Omega = e^{-\| f \|^2/2} e^{a^*(f)} \Omega = e^{-\| f \|^2/2} \left\{ 1, f , \frac{f^{\otimes 2}}{\sqrt{2!}}, \dots, \frac{f^{\otimes n}}{\sqrt{n!}}, \dots , \right\}. 
\end{equation}
Since $W(f)$ is unitary,
\[ W(f)^* = W^{-1} (f) = W(-f), \]
coherent states are always normalized, $\| W(f) \Omega \| = 1$ for all $f \in L^2 (\bR^3)$. Weyl operators generate shifts of creation and annihilation operators, i.\,e.
\begin{equation} \label{eq:Wshift} \begin{split} W(f)^* a(g) W(f) &= a(g) + \langle g, f \rangle, \\
W(f)^* a^* (g) W(f) &= a^* (g) + \langle f, g \rangle. \end{split} 
\end{equation} 
This implies that coherent states are eigenvectors of all annihilation operators, i.e.
\[ a(g) W(f) \Omega = W(f) W^* (f) a(g) W(f) \Omega = W(f) \left(a(g) + \langle g,f \rangle \right) \Omega = \langle f,g \rangle W(f) \Omega.\]
{F}rom (\ref{eq:coh-def}), we see that coherent states do not have a fixed number of particles. Instead they are a linear combination of states with all possible numbers of particles. The expectation of the number of particles operator in the state $W(f) \Omega$ is 
\[ \begin{split} \langle W(f) \Omega, \cN W(f) \Omega \rangle &= \int dx \langle W(f) \Omega, a_x^* a_x W(f) \Omega \rangle \\ &= \int dx \langle W(f) \Omega, (a_x^* + \overline{f}(x)) ( a_x + f(x)) \Omega \rangle = \| f \|^2. \end{split}\]  
More precisely the number of particles in the coherent state $W(f) \Omega$ is a Poisson random variable with mean and variance $\| f \|^2$. 

\medskip

\emph{Dynamics of coherent states.} We are interested in the evolution of coherent initial states. For $\ph \in L^2 (\bR^3)$ with $\| \ph \| =1$, we consider the coherent state $W(\sqrt{N} \ph) \Omega$. The expectation of the number of particles operator is $N$; for this reason we expect to be in the mean-field regime, in which many-body interactions can be effectively approximated by an average potential. This is made rigorous by the following theorem.
\begin{theorem}\label{thm:1}
Let $V$ be a measurable function satisfying the operator inequality 
\begin{equation}\label{eq:ass0} V^2 (x) \leq C (1-\Delta) \end{equation}
for some $C > 0$, and let $\ph \in H^1 (\bR^3)$. Let $\{\xi_N\}_{N\in\bN}$ be a sequence of vectors in $\cF$ with $\| \xi_N \| = 1$ and such that for some $C > 0$ \begin{equation}\label{eq:xiNass} \langle \xi_N, \cN \xi_N \rangle + \frac{1}{N} \langle \xi_N, \cN^2 \xi_N \rangle  \leq C \quad \forall N \in \bN.
\end{equation}
Let $\gamma_{N,t}^{(1)}$ be the one-particle reduced density associated with $\psi_{N,t} = e^{-i\cH_N t} W(\sqrt{N} \ph) \xi_N$. Then there exist constants $D,K > 0$ such that 
\begin{equation}
\label{eq:thm1} \tr \, \left| \gamma_{N,t}^{(1)} - N |\ph_t \rangle \langle \ph_t | \right| \leq D e^{K |t|} 
\end{equation}
for all $t \in \bR$ and all $N \in \bN$. Here $\ph_t$ denotes the solution of the nonlinear Hartree equation 
\begin{equation}\label{eq:hartree2}
i\partial_t \ph_t = -\Delta \ph_t + (V*|\ph_t|^2) \ph_t
\end{equation} 
with initial data $\ph_0 = \ph$. 
\end{theorem}

\emph{Remarks:}
\begin{itemize}
\item The trace-norm bound \eqref{eq:thm1} is of order one; this is to be compared to $\tr \gamma_{N,t}^{(1)} = \tr N \lvert \varphi_t\rangle \langle \varphi_t \rvert = N$. In this sense, \eqref{eq:thm1} establishes a relative rate of convergence of order $1/N$.
\item The assumption $V^2 (x) \leq C (1-\Delta)$ holds for potentials with a Coulomb singularity, since Hardy's inequality implies that 
\[ \int dx \, \frac{|\ph (x)|^2}{|x|^2} \leq \| \ph \|_{H^1}^2. \]
In fact, the result (\ref{eq:thm1}) can be easily extended to potentials $V \in L^2 (\bR^3) + L^\infty (\bR^3)$, although in general the time dependence on the r.\,h.\,s.\ of (\ref{eq:thm1}) may be worse.
\item The condition (\ref{eq:xiNass}) guarantees that the initial deviations from the coherent state $W(\sqrt{N} \ph) \Omega$, which are described by the vector $\xi_N$, are small compared to the number of particles in the initial data $W(\sqrt{N} \ph) \xi_N$, which has expectation $N$. 
\end{itemize}

\medskip

\emph{Sketch of the proof of Theorem \ref{thm:1}.} According to (\ref{eq:g1xy}), the one-particle reduced density associated with the Fock space vector $\Psi_{N,t}$ has the integral kernel 
\begin{equation}\label{eq:gNker} \gamma_{N,t}^{(1)} (x;y) = \langle \Psi_{N,t} , a_y^* a_x \Psi_{N,t} \rangle. 
\end{equation}
We define the fluctuation vector $\xi_{N,t}$ at time $t$ by setting
\begin{equation}\label{eq:xit} 
\Psi_{N,t} = e^{-i\cH_N t} W(\sqrt{N} \ph) \xi_N = W(\sqrt{N} \ph_t) \xi_{N,t} 
\end{equation}
where $\ph_t$ is the solution of the Hartree equation (\ref{eq:hartree2}). Equivalently, we can write
\[ \xi_{N,t} = \cU_N (t;0) \xi_N \]
where we introduced the fluctuation dynamics
\begin{equation}\label{eq:fldyn}
\cU_N (t;s) = W(\sqrt{N} \ph_t)^* e^{-i\cH_N t} W(\sqrt{N} \ph_s) .
\end{equation}
Plugging the r.\,h.\,s.\ of (\ref{eq:xit}) in (\ref{eq:gNker}), we obtain
\[\begin{split} \gamma_{N,t}^{(1)} (x;y) &= \langle W(\sqrt{N} \ph_t) \xi_{N,t}, a_y^* a_x W(\sqrt{N} \ph_t) \xi_{N,t} \rangle \\ &= \langle \xi_{N,t}, (a_y^* + \sqrt{N} \overline{\ph}_t (y) ) ( a_x + \sqrt{N} \ph_t (x)) \xi_{N,t} \rangle \\ &= N \overline{\ph}_t (y) \ph_t (x) + \sqrt{N} \overline{\ph}_t (y) \langle \xi_{N,t} , a_x \xi_{N,t} \rangle + \sqrt{N} \ph_t (x) \langle \xi_{N,t} , a_y^* \xi_{N,t} \rangle + \langle \xi_{N,t} , a_y^* a_x \xi_{N,t} \rangle. \end{split} \]
Integrating against a one-particle observable $J^{(1)}$ with kernel $J^{(1)} (x;y)$, we find 
\begin{equation}
\label{eq:trJg} 
\begin{split} \tr \, J^{(1)} \left(\gamma^{(1)}_{N,t} - N |\ph_t \rangle \langle \ph_t| \right) = \; &\sqrt{N} \left\langle \xi_{N,t}, \left[a (J^{(1)} \ph_t) + a^* (J^{(1)} \ph_t) \right] \xi_{N,t} \right\rangle + \langle \xi_{N,t}, d\Gamma (J^{(1)}) \xi_{N,t} \rangle .
\end{split} 
\end{equation}
According to (\ref{eq:dGJN}) and (\ref{eq:aa-bd-form}), the expectations of the operators $a(J^{(1)} \ph_t) + a^* (J^{(1)} \ph_t)$ and $d\Gamma (J^{(1)})$ can both be bounded by the expectation of the number of particles operator, in the state $\xi_{N,t} = \cU_N (t;0) \xi_N$. The next proposition is taken from \cite[Prop.\ 3.1]{BSS} but similar estimates have already been proven in \cite{RS,CLS}. The proposition extends also to higher moments of the number of particles operator.
\begin{proposition}\label{prop:growth}
Let $V$ be a measurable function satisfying the operator inequality (\ref{eq:ass0}) and let $\ph \in H^1 (\bR^3)$. Then there exist constants $D,K > 0$ such that 
\[ \langle \cU_N (t;0) \psi, \cN \cU_N (t;0) \psi \rangle \leq D e^{K |t|} \, \langle \psi, \left[ \cN + \frac{1}{N} \cN^{2} \right] \psi \rangle  \]
for any $\psi \in \cF$. Here $\cU_N$ is defined as in (\ref{eq:fldyn}). 
\end{proposition}
While Proposition \ref{prop:growth} immediately implies the desired bound for the second term on the r.\,h.\,s.\ of (\ref{eq:trJg}), some more work is required to get rid of the additional $\sqrt{N}$ factor in the first term on the r.\,h.\,s.\ of (\ref{eq:trJg}). We will omit the details which can be found in \cite{RS,CLS}, and just write the final estimate
\[ \left| \tr \, J^{(1)} \left(\gamma^{(1)}_{N,t} - N |\ph_t \rangle \langle \ph_t| \right) \right| \leq D \| J^{(1)} \| e^{K|t|}. \]
Since the space of trace class operators $\cL^1 (L^2 (\bR^3))$, equipped with the trace norm, is the dual of the space of compact operators, equipped with the operator norm, the last bound implies (\ref{eq:thm1}).

\medskip

\emph{Sketch of the proof of Proposition~\ref{prop:growth}.} Let us briefly discuss the main ideas needed to bound the growth of the number of fluctuations. We notice, first of all, that the fluctuation dynamics $\cU_N (t;s)$ satisfies the Schr\"odinger type equation 
\[ i\partial_t \cU_N (t;s) = \cL_N (t) \cU_N (t;s) \]
with the time-dependent generator
\[ \cL_N (t) = \left[ i\partial_t W^* (\sqrt{N} \ph_t) \right] W (\sqrt{N} \ph_t) + W^* (\sqrt{N} \ph_t) \cH_N W (\sqrt{N} \ph_t) .\]
The time derivative is
\footnote{A systematic way to calculate to calculate time derivatives of the form $(\partial_t e^{-A(t)})e^{A(t)}$ (with $A(t)$ a sufficiently regular family of operators) is as follows. 
Start by writing \[\begin{split}(\partial_t e^{-A(t)})e^{A(t)} & = \lim_{h\to 0}\frac{1}{h} \left(e^{-A(t+h)}-e^{-A(t)}\right) e^{A(t)} = \lim_{h \to 0} \frac{1}{h} \int_0^1d\lambda \frac{d}{d \lambda}\left(e^{-A(t+h)\lambda}e^{A(t)\lambda}\right)\\
&= - \lim_{h \to 0} \int_0^1 d \lambda\, e^{-A(t+h)\lambda} \frac{A(t+h)-A(t)}{h} e^{A(t)\lambda} = - \int_0^1d\lambda\,e^{-A(t)\lambda}\dot A(t) e^{A(t)\lambda}. \end{split}
\] In the integral then use the Baker-Campbell-Hausdorff formula $e^{-A}Be^A = B - \int_0^1 d \rho\, e^{-\rho A}[A,B]e^{\rho A}$ for operators $A,B$. The commutator $[A,B]$ in our case is a complex number, so the expression is easy to evaluate. (Otherwise it could be helpful to iterate the last integral formula to obtain a series.)}
\[ \left[ i\partial_t W^* (\sqrt{N} \ph_t) \right] W(\sqrt{N} \ph_t) = C_{1,N} (t) - \sqrt{N} \left( a^* (i\partial_t \ph_t) + a(i\partial_t \ph_t) \right) \]
and, using (\ref{eq:Wshift}), that 
\begin{equation*} \begin{split} 
W^* (\sqrt{N} \ph_t) \cH_N W (\sqrt{N} \ph_t) = & \; C_{2,N} (t) + \sqrt{N} \left( a^* (-\Delta \ph_t + (V*|\ph_t|^2) \ph_t) + a(-\Delta \ph_t + (V*|\ph_t|^2) \ph_t) \right) \\
&\; +\int dx \nabla_x a_x^* \nabla_x a_x + \int dx (V*|\ph_t|^2) (x) a_x^* a_x \\ &+ \int dx dy V(x-y) \ph_t (x) \overline{\ph}_t (y) a_x^* a_y \\ &+ \int dx dy V(x-y) \left(\ph_t (x) \ph_t (y) a_x^* a_y^* + \overline{\ph}_t (x) \overline{\ph}_t (y) a_x a_y \right) \\
&+\frac{1}{\sqrt{N}} \int dx dy V(x-y) a_x^* \left( \ph_t (y) a_y^* + \overline{\ph}_t (y) a_y \right) a_x \\
&+\frac{1}{2N} \int dx dy V(x-y) a_x^* a_y^* a_y a_x \end{split} 
\end{equation*}
for appropriate constants\footnote{We speak of constants since these are not operators, but rather just complex numbers. Of course they depend on time and on $N$.} $C_{1,N} (t)$ and $C_{2,N} (t)$. Since $\ph_t$ satisfies the Hartree equation, \emph{the contributions of order $\sqrt{N}$ (which are linear in $a^*_x$ and $a_x$) cancel exactly}, and therefore (up to a constant contribution, which can be absorbed in the phase) we obtain
\begin{equation}\label{eq:cLN}
\begin{split}
\cL_N (t) = &\; \int dx \nabla_x a_x^* \nabla_x a_x + \int dx (V*|\ph_t|^2) (x) a_x^* a_x + \int dx dy V(x-y) \ph_t (x) \overline{\ph}_t (y) a_x^* a_y \\ &+ \int dx dy V(x-y) \left(\ph_t (x) \ph_t (y) a_x^* a_y^* + \overline{\ph}_t (x) \overline{\ph}_t (y) a_x a_y \right) \\
&+\frac{1}{\sqrt{N}} \int dx dy V(x-y) a_x^* \left( \ph_t (y) a_y^* + \overline{\ph}_t (y) a_y \right) a_x \\
&+\frac{1}{2N} \int dx dy V(x-y) a_x^* a_y^* a_y a_x .
\end{split} 
\end{equation} 
We notice that $\cL_N (t)$ contains terms (the second and third line) which do not commute with the number of particles operator. This means that, as expected, the fluctuation dynamics $\cU_N (t)$, in contrast with the original dynamics $e^{-i\cH_N t}$, does not preserve the number of particles. Nevertheless we can control the growth of the expectation of $\cN$. We compute
\[ \begin{split} 
i \partial_t \langle \cU_N (t;0) \xi_N, &\cN \cU_N (t;0) \xi_N \rangle \\ = \; &\langle \cU_N (t;0) \xi_N, [ \cN , \cL_N (t) ] \cU_N (t;0) \xi_N \rangle \\ = \; &2i \text{Im }  \int dx dy V(x-y) \ph_t (x) \ph_t (y) \langle \cU_N (t;0) \xi_N , [ \cN, a_x^* a_y^* ]  \cU_N (t;0) \xi_N \rangle \\ &+ \frac{2i}{\sqrt{N}} \text{Im } \int dx dy V(x-y) \ph_t (x) \langle \cU_N (t;0) \xi_N , [\cN , a_x^* a_y^* a_x]  \cU_N (t;0) \xi_N \rangle. \end{split} \]
Using the canonical commutation relations (\ref{eq:CCR}) we find
\begin{equation}\label{eq:ddtN}\begin{split} 
\partial_t \langle \cU_N (t;0) \xi_N, \cN \cU_N (t;0) \xi_N \rangle = \; &4\text{Im} \int dx dy V(x-y) \ph_t (x) \ph_t (y) \langle \cU_N (t;0) \xi_N, a_x^* a_y^* \cU_N (t;0) \xi_N \rangle \\ &+ \frac{2}{\sqrt{N}} \text{Im } \int dx dy V(x-y) \ph_t (y) \langle \cU_N (t;0) \xi_N, a_x^* a_y^* a_x \cU_N (t;0) \xi_N \rangle \\ = :\; & \text{I} + \text{II}.\end{split} 
\end{equation}
The first term can be bounded using (\ref{eq:aa-bd}):
\[ \begin{split} |\text{I}| &= 4 \left| \int dx \ph_t (x) \langle a_x \cU_N (t;0) \xi_N, a^* (V(x-\cdot ) \ph_t) \cU_N (t;0) \xi_N \rangle \right| \\ &\leq 4 \int dx |\ph_t (x)| \| a_x \cU_N (t;0) \xi_N \| \| a^* (V(x-\cdot )\ph_t) \cU_N (t;0) \xi_N \| \\ &\leq 4\int dx |\ph_t (x)| \| a_x \cU_N (t;0) \xi_N \| \| V(x-\cdot ) \ph_t \|_2 \| (\cN+1)^{1/2} \cU_N (t;0) \xi_N \| \\ &\leq 4 \sup_x \| V(x-\cdot ) \ph_t \|_2 \| (\cN+1)^{1/2} \cU_N (t;0) \xi_N \|^2. \end{split} \]
{F}rom the assumption (\ref{eq:ass}), we obtain
\[ \| V(x-\cdot ) \ph_t \|^2 = \int dy V^2 (x-y) |\ph_t (y)|^2 \leq C \| \ph_t \|_{H^1}^2. \]
Furthermore, the Hartree energy
\begin{equation}
\cE_H (\ph) = \int |\nabla \ph (x)|^2 dx + \frac{1}{2} \int dx dy V(x-y) |\ph (x)|^2 |\ph (y)|^2 
\end{equation}
is conserved along the Hartree evolution, i.\,e.\ for a solution of the Hartree equation $\varphi_t$, $\cE_H(\ph_t)$ is independent of $t$. Together with the assumption (\ref{eq:ass}), this implies that there exists a universal constant $C > 0$ such that
\[ \| \ph_t \|_{H^1} \leq C \| \ph_0 \|_{H^1}. \]
Hence
\[ \sup_x \| V(x-\cdot ) \ph_t \| \leq K \]
for a constant $K > 0$ depending only on the constant $C0$ in (\ref{eq:ass}) and on the $H^1$-norm of the initial wave function $\ph_0$. We conclude that 
\begin{equation}\label{eq:Ibd} |\text{I}| \leq K \langle \cU_N (t;0) \xi_N, (\cN+1) \cU_N (t;0) \xi_N \rangle .
\end{equation}

To control the second term on the r.\,h.\,s.\ of (\ref{eq:ddtN}), we proceed as follows:
\[ \begin{split} 
|\text{II}| &\leq \frac{2}{\sqrt{N}} \int dx  \| a (V(x-\cdot ) \ph_t) a_x \cU_N (t;0) \xi_N\| \| a_x \cU_N (t;0) \xi_N \| \\ &\leq \frac{2}{\sqrt{N}} \sup_x \| V(x-\cdot ) \ph_t \|_2  \int dx \| a_x \cN^{1/2} \cU_N (t;0) \xi_N \| \| a_x \cU_N (t;0) \xi_N \| \\ &\leq \frac{K}{N} \langle \cU_N (t;0) \xi_N, \cN^2 \cU_N (t;0) \xi_N \rangle + K \langle \cU_N (t;0) \xi_N, \cN \cU_N (t;0) \xi_N \rangle .
\end{split} \]
In the first term on the r.\,h.\,s.\ of the last equation, we are going to use the $1/N$ factor to reduce the exponent of the number of particles operator. Since $\cU_N^* (t;0) \cN \cU_N (t;0)$ measures the number of fluctuations, it is heuristically clear that it can be bounded by the total number of particles $N$; more precisely, one can prove the operator inequality (see \cite{BSS}, proof of Prop. 3.1)
\[ \frac{1}{N} \cU_N^* (t;0) \cN^2 \cU_N (t;0) \lesssim \cU_N^* (t;0) \cN \cU_N (t;0) + \frac{\cN^2}{N}. \]
Applying this bound for $\xi_N$, we conclude that
\[ \frac{1}{N} \langle \cU_N (t;0) \xi_N, \cN^2 \cU_N (t;0) \xi_N \rangle \lesssim \langle \cU_N (t;0) \xi_N, \cN \cU_N (t;0) \xi_N \rangle + \frac{1}{N}\langle \xi_{N}, \cN^{2} \xi_{N} \rangle \]
and therefore that
\[ |\text{II}| \leq K \langle \cU_N (t;0) \xi_N, (\cN+1) \cU_N (t;0) \xi_N \rangle \, . \]
Inserting the last bound and (\ref{eq:Ibd}) into (\ref{eq:ddtN}), we find
\[ \left| \frac{d}{dt} \langle \cU_N (t;0) \xi_N, \cN \cU_N (t;0) \xi_N \rangle \right| \leq K \langle \cU_N (t;0) \xi_N, (\cN+1) \cU_N (t;0) \xi_N \rangle. \]
Gronwall's Lemma implies therefore
\[ \langle \cU_N (t;0) \xi_N, \cN \cU_N (t;0) \xi_N \rangle \leq C e^{K|t|} .\]
This concludes the sketch of the proof of Proposition \ref{prop:growth}; more details can be found in \cite{RS,CLS}.

\section{Fluctuations around Hartree dynamics}
\label{sec:fluc-bos}
\setcounter{equation}{0}

The coherent state approach presented in the last section also allows us to describe the fluctuations around the Hartree dynamics. 

\medskip

\emph{Fluctuations for coherent initial states.}
For simplicity, let us consider the evolution of a coherent state initial data $\Psi_N = W(\sqrt{N} \ph) \Omega \in \cF$. (As in the last section, we could also consider approximate coherent states $\Psi_N = W(\sqrt{N} \ph) \xi_N$, with $\xi_N$ satisfying $\langle \xi_N , (\cN^{3/2} + N^{-1/2} \cN^{5/2} +N^{-1} \cN^3) \xi_N \rangle \leq C$.) According to the definition of the fluctuation dynamics (\ref{eq:fldyn}), we can write
\[ e^{-i\cH_N t} W(\sqrt{N} \ph) \Omega = W(\sqrt{N} \ph_t) \cU_N (t;0) \Omega.\]
Recall the generator $\cL_N (t)$ of the fluctuation dynamics $\cU_N (t;s)$ as given by (\ref{eq:cLN}). The last two terms on the r.\,h.\,s.\ of (\ref{eq:cLN}), the cubic and the quartic contribution to $\cL_N (t)$, seem to vanish, in the limit of large $N$. Therefore we define a new time-dependent generator
\begin{equation}\label{eq:Linfty}
\begin{split}
\cL_\infty (t) = &\; \int dx \nabla_x a_x^* \nabla_x a_x + \int dx (V*|\ph_t|^2) (x) a_x^* a_x + \int dx dy V(x-y) \ph_t (x) \overline{\ph}_t (y) a_x^* a_y \\ &+ \int dx dy V(x-y) \left(\ph_t (x) \ph_t (y) a_x^* a_y^* + \overline{\ph}_t (x) \overline{\ph}_t (y) a_x a_y \right) 
\end{split} 
\end{equation} 
keeping only the quadratic part of $\cL_N (t)$. We denote by $\cU_\infty (t;s)$ the evolution generated by $\cL_\infty (t)$, i.e.
\begin{equation}\label{eq:Uinfty} i\partial_t \cU_\infty (t;s) = \cL_\infty (t) \cU_\infty (t;s), \quad \cU_\infty(s;s) = 1.
\end{equation}
For interaction potentials satisfying the condition (\ref{eq:ass0}) one can prove that
\begin{equation}\label{eq:fluc} \left\| (\cU_N (t;0) - \cU_\infty (t;0)) \psi \right\| \leq \frac{C|t|}{\sqrt{N}} \left[ \| \cN^{3/2} \psi \| + \frac{1}{\sqrt{N}} \| \cN^{2} \psi \| \right] 
\end{equation}
for some constant $C>0$. The proof of (\ref{eq:fluc}) is based on the identity
\[ \cU_N (t;0) - \cU_\infty (t;0) = \int_0^t \cU_N (t;s) \left( \cL_N (s) - \cL_\infty (s) \right) \cU_\infty (s;0)\, ds\]
and on the control of the growth of moments of the number of particles operator with respect to the 
limiting fluctuation dynamics $\cU_\infty (s;0)$. The bound (\ref{eq:fluc}) leads to the following 
theorem. 
\begin{theorem}\label{thm:fluc1}
Let $V$ be a measurable function satisfying the operator inequality 
\begin{equation}\label{eq:ass} V^2 (x) \leq C (1-\Delta) 
\end{equation} 
for some $C > 0$, and let $\ph \in H^1 (\bR^3)$. Then there exists a constant $D >0$ such that 
\[ \left\| e^{-i\cH_N t} W(\sqrt{N} \ph) \Omega - W(\sqrt{N} \ph_t) \cU_\infty (t;0) \Omega \right\| \leq \frac{D|t|}{\sqrt{N}}. \]
\end{theorem}
In other words, Theorem \ref{thm:fluc1} states that, if we approximate the many-body evolution of the initial coherent state taking into account also the limiting fluctuation dynamics $\cU_\infty (t;0)$ rather than only the Hartree dynamics of the coherent state, we get convergence in norm as $N \to \infty$ (with a rate proportional to $N^{-1/2}$), while approximation of the many-body evolution just by evolved coherent states only gives convergence in the weaker sense of reduced densities.   

\medskip

\emph{Bogoliubov transformations.} In contrast with $\cU_N (t;s)$ the limit fluctuation dynamics $\cU_\infty (t;s)$ is a simple evolution, because its generator is quadratic. In fact, it acts as a Bogoliubov transformation, a general concept that we introduce now. For $f,g \in L^2 (\bR^3)$, we define 
\[ A(f,g) = a(f) + a^* (\overline{g}) .\]
Then we have
\begin{equation}\label{eq:AA*} A^* (f,g) = a^* (f) + a(\overline{g}) = A (\overline{g} , \overline{f}) = A (\cJ (f,g)) 
\end{equation}
where $\cJ: L^2 (\bR^3) \oplus L^2 (\bR^3) \to L^2 (\bR^3) \oplus L^2 (\bR^3)$ is the antilinear map defined by
\[ \cJ (f,g) = (\overline{g}, \overline{f}) .\]
The canonical commutation relations take the form
\begin{equation}\label{eq:CCR2} \left[ A(f_1, g_1) , A^* (f_2, g_2) \right] = \langle (f_1, g_1), S (f_2, g_2) \rangle 
\end{equation}
with the linear operator $S: L^2 (\bR^3) \oplus L^2 (\bR^3) \to L^2 (\bR^3) \oplus L^2 (\bR^3)$ defined by
\[ S (f,g) = (f, - g) ,\]
and where $\langle \cdot , \cdot \rangle$ denotes the inner product in $L^2 (\bR^3) \oplus L^2 (\bR^3)$. A Bogoliubov transformation is a linear map $\Theta : L^2 (\bR^3) \oplus L^2 (\bR^3) \to L^2 (\bR^3) \oplus L^2 (\bR^3)$ that preserves (\ref{eq:AA*}) and (\ref{eq:CCR2}). In other words, it is a linear map satisfying the two conditions 
\[ \Theta \cJ = \cJ \Theta \quad \text{and } \quad \Theta^* S \Theta = S .\]
It is easy to check that $\Theta$ is a Bogoliubov transformation if and only if it can be written in the block form
\begin{equation}\label{eq:bogo} \Theta = \left( \begin{array}{ll} U & V \\ V & U \end{array} \right) 
\end{equation}
with two linear operators $U,V : L^2 (\bR^3) \to L^2 (\bR^3)$ satisfying $U^* U -V^* V = 1$ and $U^* V+ V^* U = 0$. 

A Bogoliubov transformation $\Theta : L^2 (\bR^3) \oplus L^2 (\bR^3) \to L^2 (\bR^3) \oplus L^2 (\bR^3)$ is called implementable if there exists a unitary map $\Xi : \cF \to \cF$ defined on the bosonic Fock space $\cF$ over $L^2 (\bR^3)$ such that
\[ \Xi^* A(f,g) \Xi = A(\Theta (f,g)) \]
for every $f,g \in L^2 (\mathbb{R}^3)$. The Shale-Stinespring condition states that a Bogoliubov transformation $\Theta$ is implementable if and only if the off-diagonal operator $V: L^2 (\bR^3) \to L^2 (\bR^3)$ is Hilbert-Schmidt (i.\,e.\ if $\tr V^* V < \infty$).

The limit fluctuation dynamics $\cU_\infty (t;s)$, defined in (\ref{eq:Uinfty}), acts for any $t,s \in \bR$ as a Bogoliubov transformation: It is simple to check that there exists a two-parameter group of Bogoliubov transformations 
$\Theta (t;s) : L^2 (\bR^3) \oplus L^2 (\bR^3) \to L^2 (\bR^3) \oplus L^2 (\bR^3)$ such that
\begin{equation}\label{eq:bogo-theta} \cU_\infty^* (t;s) A(f,g) \cU_\infty (t;s) = A(\Theta (t;s) (f,g)). 
\end{equation}
The time-dependent Bogoliubov transformations $\Theta (t;s)$ are easily seen to satisfy the equation
\[ i\partial_t \Theta (t;s) = D(t) \Theta (t;s) \]
with the generator
\[ D(t) = \left( \begin{array}{cc} -\Delta + (V*|\ph_t|^2) + A_1 & A_2 \\
A_2 & -\Delta + (V*|\ph_t|^2) + A_1 \end{array} \right). \]
Here we denoted by $A_1$ and $A_2$ the operators with integral kernels $A_1(x;y)=V(x-y)\varphi_t(x)\overline{\varphi_t(y)}$ and $A_2(x;y)=V(x-y)\varphi_t(x)\varphi_t(y)$ (these are of the form of the exchange operator, but their physical role is different).
Identifying $\cU_\infty (t;s)$ with a Bogoliubov transformation means that the limit fluctuation dynamics can be determined by  solving a partial differential equation on $L^2 (\bR^3) \oplus L^2 (\bR^3)$; in this sense we can think of $\cU_\infty (t;s)$ as a simple effective dynamics (in contrast with $\cU_N (t;s)$ which a true many-body evolution).   

\medskip

\emph{Probabilistic interpretation.} The convergence of the reduced densities can be interpreted, in the language of probability theory, as a  law of large numbers. Let $J^{(1)}$ be a one-particle observable, i.\,e.\ a self-adjoint operator over $L^2 (\bR^3)$. We denote by $J^{(1)}_i$ the operator on $L^2 (\bR^{3N})$ acting as $J^{(1)}$ on the $i$-th particle and as the identity on the other $(N-1)$ particles. W.\,r.\,t.\ a factorized $N$-particle wave function $\ph^{\otimes N}$, the observables $J^{(1)}_i$ define independent and identically distributed random variables, and therefore, for every $\delta > 0$,  
\[ \PP_{\ph^{\otimes N}} \left( \left| \frac{1}{N} \sum_{i=1}^N J^{(1)}_i - \langle \ph , J^{(1)} \ph \rangle \right| \geq \delta  \right) \to 0 \quad (N \to \infty).\]
W.\,r.\,t.\ the evolved wave function $\psi_{N,t}$, solution of the mean-field Schr\"odinger equation
\[
i\partial_{t}\psi_{N,t} = \Big[ \sum_{j=1}^{N} -\Delta_{j} + \frac{1}{N}\sum_{i<j}^{N}V(x_{i} - x_{j}) \Big]\psi_{N,t}
\]
with initial data $\psi_{N,0} = \varphi^{\otimes N}$, the $N$ observables $J^{(1)}_1, \dots , J^{(1)}_N$ are no longer independent. Nevertheless, it is easy to check that the convergence of the reduced densities implies that 
\[ \PP_{\psi_{N,t}} \left( \left| \frac{1}{N} \sum_{i=1}^N J^{(1)}_i - \langle \ph_t , J^{(1)} \ph_t \rangle \right| \geq \delta  \right) \to 0 \quad (N \to \infty).\]

At time $t=0$ we also have a central limit theorem stating that the fluctuations of $\sum_{i=1}^N J^{(1)}_i$, appropriately normalized, are Gaussian in the limit. Does the same hold true for $t \not = 0$? The answer, obtained in \cite{BKS,BSS}, is positive; with respect to the measure induced by $\psi_{N,t}$, we have
\begin{equation}\label{eq:cltt} \frac{1}{\sqrt{N}} \sum_{i=1}^N J^{(1)}_i - \langle \ph_t , J^{(1)} \ph_t \rangle \to \text{Gauss} (0, \sigma_t^2) 
\end{equation}
as $N \to \infty$, in distribution. The variance of the limiting Gaussian variable is given by
\[ \sigma_t^2 =  \frac{1}{2} \left[ \left\langle \Theta (t;0) \left(O\ph_t , J O \ph_t \right) , \Theta (t;0) \left(O \ph_t , J O\ph_t \right) \right\rangle - \left| \left\langle \Theta (t;0) \left(O\ph_t , J O \ph_t \right) , \frac{1}{\sqrt{2}} \, \left(\ph, \overline{\ph} \right) \right\rangle \right|^2 \right]  \]
where $\Theta (t;s)$ denotes the Bogoliubov transformation defined in (\ref{eq:bogo-theta}) to describe the limit fluctuation dynamics. Eq.\ (\ref{eq:cltt}) shows that in the mean-field limit the correlations among the particles are weak enough for the central limit theorem to hold true, but they are strong enough to change the variance of the limiting Gaussian variable. (For a completely factorized wave function $\ph_t^{\otimes N}$, the variance would clearly be $\langle \ph_t, (J^{(1)})^2 \ph_t \rangle - \langle \ph_t, J^{(1)} \ph_t \rangle^2$.)

\medskip

\emph{Fluctuations for data with fixed number of particles.} For $N$-particle initial data, a better approach to study fluctuations around the Hartree dynamics has been proposed in \cite{LNS}, based on ideas developed in \cite{LNSS} to analyze the excitation spectrum of mean-field Hamiltonians. Fix $\ph \in L^2 (\bR^3)$. Then we can write any $N$-particle wave function $\psi_N \in L^2_s (\bR^{3N})$ as\footnote{The symbol $\otimes_s$ denotes the symmetrized tensor product, i.\,e.\ $\varphi_1 \otimes_s \cdots \otimes_s \varphi_N = (N!)^{-1}\sum_{\sigma \in \mathcal{S}_N} \varphi_{\sigma(1)}\otimes \cdots \otimes \varphi_{\sigma(N)}$.} 
\[ \psi_N = \psi^{(0)} \ph^{\otimes N} + \psi^{(1)} \otimes_s \ph^{\otimes (N-1)} + \dots + \psi^{(N-1)} \otimes_s \ph + \psi^{(N)} \]
where the $n$-particle bosonic wave function $\psi^{(n)}$ is assumed to be orthogonal to $\ph$ in all its $n$ entries. This defines a map \[ \begin{split} U_\ph : &\; L^2_s (\bR^{3N}) \to \cF_+ \\  &\; \psi_N \mapsto \{ \psi^{(0)}, \psi^{(1)}, \dots , \psi^{(N)} , 0 , 0 \dots \} \end{split} \] which is linear and isometric. Here $\cF_+$ denotes the bosonic Fock space over the orthogonal complement of $\text{span}\{\ph\}$. We can think of $U_\ph \psi_N \in \cF_+$ as describing the fluctuations around the condensate. We assume that the initial data $\psi_N \in L^2_s (\bR^{3N})$ exhibits complete condensation in $\ph \in L^2 (\bR^3)$ in the sense that the fluctuations $\phi (0) := U_\ph \psi_N \in \cF_+$ can be bounded uniformly in $N$ in the sense that for some $C>0$ we have
\[ \left\langle \phi (0), d\Gamma (1-\Delta) \phi (0) \right\rangle \leq C \] independent of $N$. 
We let $\psi_N$ evolve with the Hamiltonian \[ H_N = \sum_{j=1}^N -\Delta_{x_J} + \frac{1}{N} \sum_{i<j}^N V (x_i -x_j) \, . \] We know that $\psi_{N,t} = e^{-iH_N t} \psi_N$ exhibits condensation in the one-particle state $\ph_t$ evolved according to the Hartree dynamics. To study the evolution of the fluctuations around the Hartree evolution we apply the map $U_{\ph_t}: L^2_s (\bR^{3N}) \to \cF_{+,t}$, defined analogously to $U_\ph$, on $\psi_{N,t}$. (Notice that the image space of $U_{\ph_t}$ depends on time because of the requirement of orthogonality to $\ph_t$.) We find that 
\begin{equation}\label{eq:fluc-LNS} \| U_{\ph_t} e^{-iH_N t} \psi_N - \phi (t)  \| \to 0 \quad (N \to \infty).
\end{equation}
Here $\phi (t) \in \cF_+$ is the solution of the evolution equation 
\begin{equation}\label{eq:FS-ev} i \partial_t \phi (t) = \cL_\infty (t) \phi (t) 
\end{equation}
with initial data $\phi (0)$ and a quadratic time-dependent generator $\cL_\infty (t)$ very similar to (\ref{eq:Linfty}) (although not exactly the same because of the requirement of orthogonality). Eq.\ (\ref{eq:fluc-LNS}) shows that, by taking into account the quadratic dynamics of the fluctuations, we obtain a norm approximation for the full evolution $\psi_{N,t}$. In other words, writing $\phi (t) = \{ \phi^{(1)} (t), \dots , \phi^{(N)} (t) , 0, 0, \dots \}$, we find 
\[ \psi_{N,t} \simeq \phi^{(0)} (t) \ph^{\otimes N}_t + \phi^{(1)} (t) \otimes_s \ph_t^{\otimes (N-1)} + \dots + \phi^{(N)} (t) \]
with an error whose norm tends to zero as $N \to \infty$. (It is also possible to check that the error is actually of order $N^{-1/2}$.)

\section{The Gross-Pitaevskii regime}
\setcounter{equation}{0}
\label{sec:GP}

Another limit in which it is possible to approximate the many-body dynamics of bosonic systems by an effective one-particle equation is the Gross-Pitaevskii regime, which is relevant for the description of trapped Bose-Einstein condensates. At the microscopic level, a trapped Bose-Einstein condensate can be described as a gas of $N$ bosons with Hamilton operator of the form
\begin{equation}\label{eq:Htrap} H_N^{\text{trap}} = \sum_{j=1}^N \left( -\Delta_{x_j} + V_{\text{ext}} (x_j) \right) + \sum_{i<j} N^2 V(N (x_i -x_j)) .
\end{equation}
Here $V_{\text{ext}}$ is an external potential modelling the trap and the interaction is described by a smooth repulsive potential $V$ (we have to assume that $V \geq 0$ pointwise) with short range (for convenience we will assume that $V$ has compact support, although this is not really necessary). 

\medskip

\emph{Scattering length.} In (\ref{eq:Htrap}) the interaction potential scales with the number of particles $N$ so that its scattering length is of order $N^{-1}$. Let us recall that the scattering length of a potential $V$ is defined through the solution of the zero-energy scattering equation 
\begin{equation}\label{eq:scat} \left( -\Delta + \frac{1}{2} V \right) f = 0 
\end{equation}
with the boundary condition $f (x) \to 1$ for $|x| \to \infty$. Under the assumption of compact support for $V$ one can show that, for $|x|$ sufficiently large, 
\[ f(x) = 1- \frac{a_0}{|x|} \]
for an appropriate constant $a_0 > 0$, which is called the scattering length of $V$. It is a simple exercise to show that equivalently the scattering length $a_0$ can also be defined through the integral
\begin{equation}\label{eq:sc-len} 8\pi a_0 = \int V(x) f(x) dx
\end{equation}
where again $f$ denotes the solution of (\ref{eq:scat}). {F}rom the point of view of physics, the scattering length $a_0$ measures the effective range of the interaction potential; two quantum mechanical particles interacting through the potential $V$, when they are far apart, feel the other particle as a hard sphere with radius $a_0$ (in particular, the scattering length of a hard sphere potential coincides with the radius of the sphere).

Notice that if $a_0$ denotes the scattering length of $V$, the scattering length of the rescaled potential $N^2 V(N\cdot)$ is given by $a=a_0/N$. This follows because by simple scaling, from (\ref{eq:scat}) we find  
\[ \left( -\Delta + \frac{N^2}{2} V(N\cdot) \right) f(N\cdot) = 0 \]
and 
\[ f(Nx) = 1- \frac{a_0}{N|x|} = 1- \frac{(a_0/N)}{|x|}. \]

\medskip

\emph{Ground state properties of trapped condensates.} It was proven in \cite{LSY} that the ground state energy per particle for the Hamiltonian (\ref{eq:Htrap}) converges, as $N \to \infty$, towards the minimum of the Gross-Pitaevskii energy functional 
\begin{equation}\label{eq:GPen} \cE_{\text{GP}} (\ph) = \int \left[ |\nabla \ph|^2 + V_{\text{ext}} |\ph|^2 + 4\pi a_0 |\ph|^4 \right] dx 
\end{equation}
over all one-particle wave functions $\ph \in L^2 (\bR^3)$ with $\| \ph \| = 1$. {F}rom (\ref{eq:GPen}) we conclude that in first approximation the ground state energy of the boson gas depends only on the scattering length of the interaction potential, not on its precise profile.  

It was then shown in \cite{LS} that the ground state of $H_N^{\text{trap}}$ exhibits complete condensation in the minimizer of the Gross-Pitaevskii energy functional (\ref{eq:GPen}). More precisely, the one-particle reduced density $\gamma^{(1)}_N$ associated with the ground state of $H_N^\text{trap}$ was proven to satisfy 
\[ \frac{1}{N} \gamma^{(1)}_N \to |\ph \rangle \langle \ph| \quad (N \to \infty).\] The interpretation of this result is straightforward: in the ground state of $H_N^{\text{trap}}$ all particles, up to a fraction vanishing in the limit $N \to \infty$, are condensated in the one-particle state described by the unique  minimizer of the Gross-Pitaevskii functional.

\medskip

\emph{Dynamics of initially trapped condensates.} Since Gross-Pitaevskii theory has proved so successful in the description of the ground state properties of the Hamiltonian (\ref{eq:Htrap}), can it also be used to predict the time-evolution of initially trapped condensates? As in Section \ref{sec:mf} we want to study the reaction of the system to a change of the external fields. At time $t=0$, we assume the boson gas to be prepared in the ground state $\psi_N$ of the trapping Hamiltonian (\ref{eq:Htrap}). Then, the traps are switched off and the condensate starts to evolve by the translation invariant Hamiltonian 
\begin{equation}\label{eq:HN-GP} H_N = \sum_{j=1}^N -\Delta_{x_j} + \sum_{i<j}^N N^2 V(N (x_i -x_j)) .
\end{equation}
The next theorem, taken from \cite{ESY1,ESY2,ESY3}, describes the resulting dynamics. (More recently, a similar statement has been shown in \cite{P}.)
\begin{theorem}\label{thm:GP}
Let $V \geq 0$ be a spherically symmetric, short range, bounded potential with scattering length $a_0$. Consider a sequence $\psi_N \in L^2_s (\bR^{3N})$ satisfying 
\begin{itemize}
\item \emph{Finite energy per particle:} there exists $C > 0$ such that $\langle \psi_N, H_N \psi_N \rangle \leq C N$;
\item \emph{Condensation:} the one-particle reduced density $\gamma^{(1)}_N$ associated with $\psi_N$ is such that 
\[ \frac{1}{N} \gamma_N^{(1)} \to |\ph \rangle \langle \ph| \quad (N \to \infty)\]
for a $\ph \in L^2 (\bR^3)$. 
\end{itemize}
Let $\psi_{N,t} = e^{-iH_N t} \psi_N$ and let $\gamma_{N,t}^{(1)}$ be the one-particle reduced density associated with $\psi_{N,t}$. Then, for every fixed $t \in \bR$, 
\begin{equation}\label{eq:conv-GP}
\frac{1}{N} \gamma_{N,t}^{(1)} \to |\ph_t \rangle \langle \ph_t |
\end{equation} 
as $N \to \infty$, where $\ph_t$ is the solution of the time-dependent Gross-Pitaevskii equation
\begin{equation}\label{eq:GP1} i\partial_t \ph_t = -\Delta \ph_t + 8\pi a_0 |\ph_t|^2 \ph_t 
\end{equation}
with the initial condition $\ph_0 = \ph$. 
\end{theorem}

\emph{Remarks.} Convergence in (\ref{eq:conv-GP}) holds for example in the trace norm. Moreover (\ref{eq:conv-GP}) also implies convergence of higher order reduced densities. If $\gamma_{N,t}^{(k)}$ denotes the $k$-particle reduced density associated with $\psi_{N,t}$, it is shown in \cite{ESY2,ESY3} that 
\[ \frac{1}{\binom{N}{k}} \gamma_{N,t}^{(k)} \to |\ph_t \rangle \langle \ph_t|^{\otimes k} \]
as $N \to \infty$, for any fixed $k \in \bN$ and $t \in \bR$. 

\medskip

\emph{Comparison with mean-field regime.} Let us discuss the relation between the Gross-Pitaevskii limit, characterized by an interaction potential with scattering length of the order $N^{-1}$, and the mean-field regime discussed in Section \ref{sec:mf}. 
Writing the interaction potential in (\ref{eq:HN-GP}) as \[ N^2 V (Nx) = \frac{1}{N} v_N (x) \] with $v_N (x) = N^3 V (N x)$, one could think of the Gross-Pitaevskii limit as a mean-field limit with a potential $v_N$ that converges towards a delta distribution in the limit of large $N$. However, this interpretation of (\ref{eq:HN-GP}) is quite misleading. Formally we have
\[ v_N (x) = N^3 V(Nx)  \to b_0 \delta (x) \]
with $b_0 = \int V(x) dx$. Hence, based on the mean-field interpretation of (\ref{eq:HN-GP}), we should expect the many-body evolution to be approximated, in the limit of large $N$, by the nonlinear Schr\"odinger equation 
\[ i\partial_t \ph_t = -\Delta \ph_t + b_0 (\delta * |\ph_t|^2) \ph_t  = -\Delta \ph_t + b_0 |\ph_t|^2 \ph_t.\]
This equation has the same form as the Gross-Pitaevskii equation (\ref{eq:GP1}) but a different constant in front of the nonlinearity. The reason why the mean-field interpretation leads to a wrong coupling constant is that physically the two regimes are very different. While the mean-field regime is characterized by a large number of very weak collisions, in the Gross-Pitaevskii limit particles interact very rarely (only when they are at distances of the order $N^{-1}$, which is much smaller than the typical distance $N^{-1/3}$ among the particles) but when they do interact, the collisions are very strong. Because of these rare and strong collisions, the solution of the Schr\"odinger equation $\psi_{N,t}$ generated by the Hamiltonian (\ref{eq:HN-GP}) develops a singular correlation structure, varying on the length scale $N^{-1}$, which is then responsible for the emergence of the scattering length in the Gross-Pitaevskii equation (\ref{eq:GP1}). 
Correlations among the particles therefore play a crucial role in the Gross-Pitaevskii regime, while they are negligible in the mean-field limit. 

\medskip

\emph{Correlation structure.} Let us now try to explain how the correlation structure affects the dynamics of the condensate. (Recently the correlation structure developed by the solution of the Schr\"odinger equation in the Gross-Pitaevskii limit has been studied in \cite{CH2}.) Let us normalize the one- and two-particle reduced densities associated with $\psi_{N,t}$ to $\wt{\gamma}_{N,t}^{(1)} = N^{-1} \gamma_{N,t}^{(1)}$ and $\wt{\gamma}_{N,t}^{(2)} = \binom{N}{2}^{-1} \gamma_{N,t}^{(2)}$. They satisfy the differential equation 
\begin{equation}\label{eq:bbgky1} 
i\partial_t \wt{\gamma}_{N,t}^{(1)} = [-\Delta, \wt{\gamma}^{(1)}_{N,t} ] + (N-1) \tr_2 \left[ N^2 V(N (x_1 - x_2)), \wt{\gamma}^{(2)}_{N,t} \right] 
\end{equation}
which is the first of the $N$ coupled equations forming the BBGKY hierarchy (similar to (\ref{eq:BBGKY}) in the mean-field setting). By assumption, at time $t = 0$, $\psi_{N,0}$ exhibits condensation, meaning that $\wt{\gamma}_{N,t=0}^{(1)} \to |\ph \rangle \langle \ph|$ for a $\ph \in L^2 (\bR^3)$. If condensation is approximately preserved by the time-evolution, we should expect that also for $t \not = 0$  
\begin{equation}\label{eq:ans-wt1} \wt\gamma_{N,t}^{(1)} \simeq |\ph_t \rangle \langle \ph_t| .
\end{equation}
As for the two-particle reduced density, again we expect approximate factorization. Here, however, we should also take into account correlations. Assuming that correlations can be described by means of the solution of the zero-energy scattering equation (\ref{eq:scat}), we expect that the integral kernel of $\wt{\gamma}_{N,t}^{(2)}$ can be approximated by 
\begin{equation}\label{eq:ans-wt2} \wt\gamma_{N,t}^{(2)} (x_1, x_2 ; y_1, y_2) \simeq f (N(x_1 - x_2)) f(N (y_1 - y_2)) \ph_t (x_1) \ph_t (x_2) \overline{\ph}_t (y_1) \overline{\ph}_t (y_2) .
\end{equation}
Using this ansatz for $\wt\gamma^{(2)}_{N,t}$, we find that the second term on the r.\,h.\,s.\ of (\ref{eq:bbgky1}) approximately has the integral kernel
\begin{equation}\label{eq:comp-wt2} 
\begin{split} 
(N-1) \tr_2 &\left[ N^2 V(N (x_1 - x_2)), 
\gamma^{(2)}_{N,t} \right] (x;y) \\ &= (N-1) \int dx_2 \left( N^2 V(N (x - x_2)) - N^2 V (N(y- x_2)) \right) 
\gamma^{(2)}_{N,t} (x,x_2 ; y, x_2) \\ &\simeq \int dx_2 \left( N^3 V(N (x -x_2)) - N^3 V(N (y-x_2)) \right) \ph_t (x) \overline{\ph}_t (y) |\ph_t (x_2)|^2 \\ &\hspace{8cm} \times f (N(x-x_2)) f(N (y-x_2))  \\
&\simeq 8\pi a_0 \left( |\ph_t (x)|^2 - |\ph_t (y)|^2 \right) \ph_t (x) \overline{\ph}_t (y) \end{split} 
\end{equation}
where we used (\ref{eq:sc-len}) and the fact that $f (Nx) \to 1$ weakly as $N \to \infty$. Inserting also (\ref{eq:ans-wt1}) in the term on the l.\,h.\,s.\ as well as in the first term on the r.\,h.\,s.\ of (\ref{eq:bbgky1}), we see that the Gross-Pitaevskii equation (\ref{eq:GP1}) arises exactly as the self-consistent equation for $\ph_t$. Notice that the presence of the solution of the zero-energy scattering equation $f(N\cdot)$ in the ansatz (\ref{eq:ans-wt2}) for the two particle reduced density does not contradict the fact that, as $N \to \infty$, $\wt\gamma_{N,t}^{(2)} \to |\ph_t \rangle \langle \ph_t|^{\otimes 2}$; the correlation structure in (\ref{eq:ans-wt2}) is non-trivial only for $|x_1 -x_2| \lesssim N^{-1}$ or $|y_1 - y_2| \lesssim N^{-1}$, and it disappears in the limit $N \to\infty$. Nevertheless it plays a crucial role in (\ref{eq:comp-wt2}) because it is multiplied with a very singular potential, varying on the same short scale.

\medskip

\emph{Energy estimate.} Since the presence of the solution $f (N\cdot)$ of the zero energy scattering equation in  (\ref{eq:ans-wt2}) plays such an important role, a rigorous derivation of the Gross-Pitaevskii equation (\ref{eq:GP1}) requires a proof that the solution $\psi_{N,t}$ of the many-body Schr\"odinger equation really develops a correlation structure and that, in good approximation, this correlation structure can be described by $f (N\cdot)$. To reach this goal, it is useful to prove certain energy estimates, bounding appropriate Sobolev norms of $\psi_{N,t}$ by moments of the Hamiltonian.

To give an example of the analysis involved in this step, we prove a simple energy estimate for the case of small interaction potentials. More precisely, we define the following dimensionless quantity to measure the strength of the radial interaction $V$: 
\[ \rho = \sup_{r \geq 0} r^2 V(r) + \int_0^\infty dr \, r V(r) \]
where we introduced the short hand notation $r = |x|$ and $V(r) \equiv V(x)$ if $|x| = r$. Smallness of $\rho$ also implies that the solution $f$ of the zero-energy scattering equation (\ref{eq:scat}) remains close to $1$ (recall the boundary condition $f(x) \to 1$ as $|x| \to \infty$). In fact, one can show \cite[Lemma D.1]{ESY1} that there exists a constant $c > 0$ with 
\begin{equation}\label{eq:fprop} 
1-c \rho \leq f(x) \leq 1,\qquad |\nabla f (x)| \leq c \frac{\rho}{|x|},\qquad |\nabla^2 f (x)| \leq c \frac{\rho}{|x|^2} 
\end{equation} 
for all $x \in \bR^3$. Using (\ref{eq:fprop}), we can prove the following proposition, which is taken from \cite{ESY2}.  
\begin{proposition}\label{prop:enest}
Let $V \geq 0$ be a spherically symmetric, short range, bounded potential with $\rho > 0$ small enough. Let $f$ denote the solution of the zero-energy scattering equation (\ref{eq:scat}). Then 
\begin{equation}\label{eq:enest} \langle \psi_N, H_N^2 \psi_N \rangle \geq \frac{N^2}{2} \int d \bx \left| \nabla_{x_1} \nabla_{x_2} \frac{\psi_N (\bx)}{f (N (x_1 -x_2))} \right|^2 
\end{equation}
for every $\psi_N \in L^2_s (\bR^{3N})$. Here we used the convention $\bx = (x_1, \dots, x_N) \in \bR^{3N}$. 
\end{proposition}

\emph{Remark.} Let $\psi_{N,t} = e^{-iH_N t} \psi_N$ be the solution of the Schr\"odinger equation with initial data $\psi_N$. Assume additionally (with respect to the assumptions of Theorem \ref{thm:GP}) that $\langle \psi_N, H_N^2 \psi_N \rangle \leq C N^2$; this can be achieved with an approximation argument that we skip. Then (\ref{eq:enest}) implies that
\[ N^2 \int d\bx \left| \nabla_{x_1} \nabla_{x_2} \frac{\psi_{N,t} (\bx)}{f(N(x_1 -x_2))} \right|^2 \leq 2 \langle \psi_{N,t} , H_N^2 \psi_{N,t} \rangle = 2 \langle \psi_N , H_N^2 \psi_N \rangle \leq 2C N^2 \]
and therefore that 
\begin{equation}\label{eq:apries}
\int d\bx \left|\nabla_{x_1} \nabla_{x_2} \frac{\psi_{N,t} (\bx)}{f(N(x_1 -x_2))} \right|^2 \leq C 
\end{equation}
uniformly in $N$ and in $t$. Note that in (\ref{eq:apries}) it is very important that we divide $\psi_{N,t}$ by $f(N(x_1- x_2))$ before we take the two derivatives. Keeping in mind that $f(x) = 1 - a_0/|x|$ for $|x|$ large enough, it is easy to check that
\[ \int dx \, |\nabla^2  (f (Nx))|^2 \simeq N \]
and therefore we can expect that
\[ \int d\bx \left|\nabla_{x_1} \nabla_{x_2} \psi_{N,t} (\bx) \right|^2 \simeq N. \]
Only if we first divide $\psi_{N,t}$ by $f(N (x_1 -x_2))$, removing in this way the singular correlation structure in the $x_1 -x_2$ variable, we obtain in (\ref{eq:apries}) a bound of order one. In other words, (\ref{eq:apries}) implies that, when $|x_1 -x_2| \lesssim N^{-1}$, the evolved wave function $\psi_{N,t}$ can be approximated by $f (N(x_1 -x_2))$ times a function $\psi_{N,t} (\bx)/ f(N(x_1 -x_2))$ which varies on a larger length-scale. In this sense, we can say that the estimate (\ref{eq:apries}) proves that $\psi_{N,t}$ has a short scale correlation structure, which can in good approximation be described by the function $f(N (x_1 - x_2))$. For this reason the bound (\ref{eq:apries}) plays a crucial role in the proof of Theorem \ref{thm:GP} (for the case of weak potential treated in \cite{ESY2}; for large potentials, a different energy estimate has been proved and applied in \cite{ESY3}).

\begin{proof}[Proof of Proposition \ref{prop:enest}]
We write
\[ H_N = \sum_{j=1}^N h_j \qquad \text{with } \quad h_j = -\Delta_{x_j} + \frac{1}{2} \sum_{i \not = j} N^2 V(N (x_i - x_j)). \]
Then, using the permutation symmetry of $\psi_N$, we find 
\[ \langle \psi_N, H_N^2 \psi_N \rangle = N(N-1) \langle \psi_N, h_1 h_2 \psi_N \rangle + N \langle \psi_N , h_1^2 \psi_N \rangle \geq N(N-1) \langle \psi_N, h_1 h_2 \psi_N \rangle. \]
The positivity of the potential and the fact that for all $j \geq 3$
\[ N^2 V(N (x_1 - x_j)) (-\Delta_{x_2}) = \nabla_{x_2}^* N^2 V (N (x_1 - x_j)) \nabla_{x_2} \]
imply that 
\begin{equation}\label{eq:H2-est} \langle \psi_N , H_N^2 \psi_N \rangle \geq N (N-1) \langle \psi_N, \left( -\Delta_{x_1} + \frac{N^2}{2} V(N (x_1- x_2)) \right) \left( -\Delta_{x_2} + \frac{N^2}{2} V(N (x_1 -x_2)) \right) \psi_N \rangle .
\end{equation}
Next we define $\phi_N (\bx) = \psi_N (\bx) / f (N (x_1 -x_2))$ and write
\begin{equation}\label{eq:Dlf} \frac{-\Delta_{x_1} \psi_N}{f (N (x_1 - x_2))} = -\Delta_{x_1} \phi_N - \frac{2N (\nabla f) (N(x_1 - x_2))}{f (N(x_1 -x_2))} \cdot \nabla_{x_1} \phi_N + \frac{N^2 (-\Delta f) (N(x_1 - x_2))}{f (N (x_1 -x_2))}. 
\end{equation}
Hence
\begin{equation}\label{eq:L1} \left(-\Delta_{x_1} + \frac{N^2}{2} V(N (x_1 -x_2)) \right) \psi_N =  f (N(x_1 -x_2)) L_1 \phi_N 
\end{equation}
where we defined the operator
\[ L_1 = -\Delta_{x_1} \phi_N - \frac{2N (\nabla f) (N(x_1 -x_2))}{f (N (x_1 -x_2))} \cdot \nabla_{x_1} \]
The important observation is that the last term on the r.\,h.\,s.\ of (\ref{eq:Dlf}) cancels exactly with the potential, because $f$ satisfies the zero-energy scattering equation. Analogously, we find that 
\[ \left( -\Delta_{x_2} + \frac{N^2}{2} V(N (x_1 -x_2)) \right) \psi_N =  f (N(x_1 -x_2)) L_2 \phi_N 
\]
with 
\[ L_2 = -\Delta_{x_2} \phi_N + \frac{2N (\nabla f) (N(x_1 -x_2))}{f (N (x_1 -x_2))} \cdot \nabla_{x_2}. \]
It is important to observe that $L_1, L_2$ are the Laplace operator w.\,r.\,t.\ the weight $f^2 (N\cdot)$. In other words
\[ \begin{split} \int d\bx \, f^2_N (x_1 -x_2) \overline{(L_1 \chi_N) (\bx)} \xi_N (\bx) &= \int d\bx \, f^2_N (x_1 -x_2) \overline{\chi_N} (\bx) (L_1 \xi_N)(\bx) \\ &= \int d\bx \, f^2 (N(x_1 - x_2)) \overline{\nabla_{x_1}\chi_N} (\bx) \nabla_{x_1} \xi_N (\bx) \end{split} \]
and analogously for $L_2$. 

{F}rom (\ref{eq:H2-est}) we conclude that 
\begin{equation}\label{eq:prop-fin} \begin{split} \langle \psi_N , H_N^2 \psi_N \rangle \geq \; &N(N-1) \int d\bx \, f^2(N (x_1 - x_2)) \overline{L_1 \phi_N} (\bx) L_2 \phi_N (\bx) \\ &=  N (N-1) \int d\bx \, f^2 (N(x_1 -x_2)) \overline{\nabla_{x_1} \phi_N} (\bx) \nabla_{x_1} L_2 \phi_N (\bx) 
\\ =\; &N (N-1) \int d\bx \, f^2 (N(x_1 -x_2)) \, \left| \nabla_{x_2} \nabla_{x_1} \phi_N (\bx) \right|^2 \\ &+ N(N-1) \int d\bx f^2 (N (x_1 -x_2)) \, \overline{\nabla_{x_1} \phi_N} (\bx) [\nabla_{x_1}, L_2] \phi_N (\bx) \\= \; &N (N-1) \int d\bx \, f^2 (N(x_1 -x_2)) \, \left| \nabla_{x_2} \nabla_{x_1} \phi_N (\bx) \right|^2 \\ &+ N(N-1) \int d\bx f^2 (N (x_1 -x_2)) \, \overline{\nabla_{x_1} \phi_N} (\bx) \nabla \frac{N (\nabla f) (N(x_1 -x_2))}{f (N (x_1 -x_2))}  \nabla_{x_2} \phi_N (\bx). \end{split} 
\end{equation}
In the first term on the r.\,h.\,s.\ we bound $f^2 (N (x_1 -x_2)) \geq (1- C \rho)^2$ from below (see (\ref{eq:fprop})). As for the second term on the r.\,h.\,s.\ of (\ref{eq:prop-fin}), we notice that
\[ \nabla \frac{N (\nabla f) (N(x_1 -x_2))}{f (N (x_1 -x_2))} = \frac{N^2 (\nabla^2 f) (N (x_1 -x_2))}{f (N (x_1 -x_2))} + \frac{N^2 (\nabla f)^2 (N (x_1 -x_2))}{f^2 (N (x_1 -x_2))} \, . \]
The bounds (\ref{eq:fprop}) therefore imply that 
\[ \left| \nabla \frac{N (\nabla f) (N(x_1 -x_2))}{f (N (x_1 -x_2))} \right| \leq C \rho \frac{1}{|x_1 - x_2|^2} \]
for small $\rho > 0$. Hence 
\[ \begin{split} &\left|  N(N-1) \int d\bx f^2 (N (x_1 -x_2)) \, \overline{\nabla_{x_1} \phi_N} (\bx) \nabla \frac{N (\nabla f) (N(x_1 -x_2))}{f (N (x_1 -x_2))}  \nabla_{x_2} \phi_N (\bx) \right| \\ &\hspace{6cm} \leq C \rho N(N-1) \int d\bx \frac{1}{|x_1 -x_2|^2} |\nabla_{x_1} \phi_N (\bx)| |\nabla_{x_2} \phi_N (\bx)| \\ &\hspace{6cm} \leq C \rho N(N-1) \int d\bx \, |\nabla_{x_2} \nabla_{x_1} \phi_N (\bx)|^2 \end{split} \]
and thus
\[ \langle \psi_N , H_N^2 \psi_N \rangle \geq  (1-C\rho) N(N-1) \int d\bx \left| \nabla_{x_1} \nabla_{x_2} \phi_N (\bx) \right|^2. \]
For $\rho > 0$ sufficiently small, we obtain the desired bound.
\end{proof}

\medskip

\emph{Strategy of the proof of Theorem \ref{thm:GP}.} As mentioned above the energy estimate (\ref{eq:enest}) and its corollary (\ref{eq:apries}) play a crucial role in the proof of Theorem \ref{thm:GP} because they can be used to identify the short-scale correlation structure characterizing the solution $\psi_{N,t}$ of the many-body Schr\"odinger equation. Let $\{ \wt\gamma_{N,t}^{(k)} \}_{k=1}^N$ be the family of normalized reduced densities associated with $\psi_{N,t}$, satisfying $\tr\, \wt\gamma_{N,t}^{(k)} = 1$ for all $N,k \in \bN$ and $t \in \bR$. Denote by $\{ \wt\gamma_{\infty,t}^{(k)} \}_{k \geq 1}$ a limit point of the sequence $\{ \wt\gamma_{N,t}^{(k)} \}_{k=1}^N$ for $N \to \infty$. The estimate \eqref{eq:apries} implies that, for large but fixed $N \in \bN$ (along the appropriate subsequence), we can approximate 
\begin{equation}\label{eq:appro-gk} \wt\gamma_{N,t}^{(k)} (x_1, \dots , x_k ; y_1, \dots , y_k) \simeq \prod_{i<j}^k f(N(x_i -x_j)) \wt\gamma_{\infty,t}^{(k)} (x_1, \dots , x_k; y_1, \dots , y_k) 
\end{equation}
This allows us to derive, starting from the BBGKY hierarchy for the reduced densities $\{ \wt\gamma^{(k)}_{N,t} \}_{k=1}^N$, the infinite hierarchy 
\begin{equation}
\label{eq:inf-GP}
i\partial_t \wt\gamma_{\infty,t}^{(k)} = \sum_{j=1}^k \left[ -\Delta_{x_j} , \wt\gamma_{\infty,t}^{(k)} \right] + 8 \pi a_0 \sum_{j=1}^k \tr_{k+1} \left[ \delta (x_j - x_{k+1}) , \wt\gamma_{\infty,t}^{(k)}\right].
\end{equation} 
for the limit point $\{ \wt\gamma_{\infty,t}^{(k)} \}_{k \geq 1}$. The presence of the correlation structure in (\ref{eq:appro-gk})   is the reason why the scattering length $a_0$ arises in the infinite hierarchy (\ref{eq:inf-GP}). What is still missing to show Theorem \ref{thm:GP} is a proof of the uniqueness of the solution of the infinite hierarchy (\ref{eq:inf-GP}). Before proving uniqueness, it is important to understand in which space of families of densities $\{ \wt\gamma_{\infty,t}^{(k)} \}_{k \geq 1}$ uniqueness should be shown; proving uniqueness in a smaller space is of course easier, but it also requires to show that every limit point of the sequence $\{ \wt\gamma_{N,t}^{(k)} \}_{k =1}^N$ is contained in that small space. It turns out that a possible choice of the space of densities  is the Sobolev-type space
\[ \cH_1 = \{ \{\gamma^{(k)} \}_{k\geq 1}\! :\ \exists C>0\ \forall k \in \bN\ \tr\, (1-\Delta_{x_1}) \dots (1-\Delta_{x_k}) \gamma^{(k)} \leq C^k\}. \]
Hence, to conclude the proof of Theorem \ref{thm:GP}, one has to prove that every limit point of $\{ \gamma_{N,t}^{(k)} \}_{k\geq 1}$ is in the space $\cH_1$, and further one has to show that, for any initial condition $\wt\gamma_{\infty,0}^{(k)} = |\ph \rangle \langle \ph|^{\otimes k} \in \cH_1$, there is at most one solution $\{ \wt\gamma_{\infty,t}^{(k)} \}_{k \geq 1}$ of the infinite hierarchy (\ref{eq:inf-GP}) which is in $\cH_1$ for all $t \in \bR$. 

The proof that every limit point $\{ \wt\gamma^{(k)}_{\infty,t} \}_{k \geq 1}$ of the sequence $\{ \wt\gamma_{N,t}^{(k)} \}_{k=1}^N$ satisfies the a-priori bounds 
\begin{equation}\label{eq:apri-k} \tr \, (1-\Delta_{x_1}) \dots (1-\Delta_{x_k}) \wt\gamma_{\infty,t}^{(k)} \leq C^k 
\end{equation}
for all $k \geq 1$ is based on energy estimates similar to Proposition \ref{prop:enest}, but involving also higher moments of $H_N$. 
A challenge is the fact that for finite $N$, the reduced densities $\wt\gamma_{N,t}^{(k)}$ cannot satisfy (\ref{eq:apri-k}), at least not uniformly in $N$, because of the presence of the correlation structure (taking derivatives of the correlation functions $f(N (x_i - x_j))$ produces factors of $N$). Only after taking the limit $N \to \infty$, the correlation structure disappears and one can expect (\ref{eq:apri-k}) to hold true. To obtain (\ref{eq:apri-k}), one needs therefore to show estimates of the form
\begin{equation}\label{eq:apri-cut} \int d\bx \, \Theta_{\ell} (\bx) \left| \nabla_{x_1} \dots \nabla_{x_k} \psi_{N,t} (\bx) \right|^2 \leq C^k 
\end{equation}
where the cutoff $\Theta_\ell$ satisfies $\Theta_\ell (\bx) \simeq 0$ if there exists $i \leq k$ and $j \leq N$ with $|x_i -x_j| \leq \ell$, while $\Theta_\ell (\bx) \simeq 1$ if $|x_i -x_j| \gg \ell$ for all $i \leq k$ and $j \in \{1 , \dots , N\}$ with $j\not = i$. If $\ell$ is sufficiently large (it turns out that one needs $N \ell^2 \gg 1$), the cutoff removes all singularities of $\nabla_{x_1} \dots \nabla_{x_k} \psi_{N,t}$ due to the correlation structure (because correlations are only important when particles are close to each other). At the same time, if $N \ell^3 \ll 1$, the effect of the cutoff turns out to be negligible in the limit of large $N$; hence (\ref{eq:apri-k}) follows from (\ref{eq:apri-cut}), choosing $N^{-1/2} \ll \ell \ll N^{-1/3}$.

The proof of the uniqueness of the infinite hierarchy (\ref{eq:inf-GP}) in the class $\cH_1$, given in \cite{ESY1}, is based on an diagrammatic expansion in Feynman graphs. Since the singularity of the interaction (the $\delta$-function (\ref{eq:inf-GP})) cannot be controlled by the kinetic energy (the $\cH_1$-norm), one needs to make use of the dispersive properties of the free evolution generated by the first term on the r.\,h.\,s.\ of (\ref{eq:inf-GP}). To perform the analysis it is convenient to switch to Fourier space. Thanks to the decay in momentum characterizing every family of densities in $\cH_1$ and to the decay of the propagators of the free evolution, the contribution associated to every Feynman graph is convergent in the ultraviolet regime. Further details can be found in \cite{ESY1}. 

A different and shorter approach to prove the uniqueness of the infinite hierarchy (\ref{eq:inf-hier}) was later proposed in \cite{KM} based on certain space-time estimates for the densities $\wt\gamma_{\infty,t}^{(k)}$. Recently \cite{CHPS}, this approach was applied to deduce uniqueness in the space $\cH_1$; an important ingredient was the quantum de Finetti theorem. 

\medskip

\emph{Coherent states approach in the Gross-Pitaevskii regime.} Theorem \ref{thm:GP} proves the convergence towards the Gross-Pitaevskii dynamics without control on the rate of the convergence. Since in real systems the number of particles $N$ is large 
but finite (e.\,g.\ $N \simeq 1000$ in very dilute samples of Bose-Einstein condensates), it is important to know how large $N$ must be in order for the Gross-Pitaevskii equation to become a good approximation of the many-body quantum dynamics. Can the coherent states approach presented in Section \ref{sec:cohe} for the mean-field regime also be applied in the Gross-Pitaevskii limit to obtain an explicit bound on the error?   

Let us first try to proceed naively, following exactly the same strategy as in the mean-field case. We switch to the bosonic Fock space $\cF = \bigoplus_{n \geq 0} L^2_s (\bR^{3n})$ on which we define the Hamilton operator 
\begin{equation}\label{eq:HNF} \cH_N = \int dx \nabla_x a_x^* \nabla_x a_x + \frac{1}{2} \int dx dy N^2 V (N (x-y)) a_x^* a_y^* a_y a_x. 
\end{equation}
We consider a coherent state initial data $\Psi_N = W(\sqrt{N} \ph) \Omega$ with a $\ph \in L^2 (\bR^3)$ with $\| \ph \| = 1$. The expected number of particles in the state $\Psi_N$ is $N$. We let $\Psi_N$ evolve and try to approximate it with a new coherent state. To this end we define the fluctuation vector $\xi_{N,t} \in \cF$ by
\[ e^{-i\cH_N t} W(\sqrt{N} \ph) \Omega = W(\sqrt{N} \ph_t) \xi_{N,t}, \]
which can be rewritten as 
\[ \xi_{N,t} = \cU_N (t;0) \Omega \]
with the fluctuation dynamics
\begin{equation}\label{eq:fldyn-GP} \cU_N (t;s) = W(\sqrt{N} \ph_t)^* e^{-i\cH_N t} W(\sqrt{N} \ph_s) .
\end{equation}
Instead of choosing $\ph_t$ to solve the Gross-Pitaevskii equation (\ref{eq:GP1}), it is convenient to consider the solution of the modified Gross-Pitaevskii equation\footnote{Of course, this equation depends on $N$, and so do its solutions. So when we talk about convergence of this equation to the Gross-Pitaevskii equation we mean convergence of its $N$-dependent solutions to solutions of the Gross-Pitaevskii equation. Despite solutions being $N$-dependent we don't put an extra $N$-index to keep the notation light.}
\begin{equation}\label{eq:mod-GP} i\partial_t \ph_t = -\Delta \ph_t + (N^3 V(N\cdot) f(N\cdot) * |\ph_t|^2) \ph_t 
\end{equation}
with the same initial data $\ph_0=\ph$. As $N \to \infty$, 
\[ N^3 V(Nx) f(Nx) \to 8\pi a_0 \delta (x) \]
and the solution of the modified equation (\ref{eq:mod-GP}) can be easily shown to converge towards the solution of the Gross-Pitaevskii equation (\ref{eq:GP1}), with the rate $N^{-1}$ (at least for sufficiently regular initial data). 

Similarly as in the mean-field case, to show that $e^{-i\cH_N t} W(\sqrt{N} \ph) \Omega$ can be approximated by 
the evolved coherent state it is enough to prove that the fluctuation vector $\xi_{N,t}$ stays close to the vacuum. More precisely, it is enough to control the growth of the expectation of the number of particles operator 
\begin{equation}
\label{eq:grow-GP}
\langle \cU_N (t;0) \Omega, \cN \cU_N (t;0) \Omega \rangle 
\end{equation}
with respect to the fluctuation dynamics $\cU_N (t;s)$. To this end, we compute the (time-dependent) generator 
$\cL_N (t)$ of $\cU_N (t;s)$, defined by the equation
\[ i\partial_t \cU_N (t;s) = \cL_N (t) \cU_N (t;s), \quad \cU_N (s;s) = 1 \]
and thus given by 
\[ \cL_N (t) = \left[ i\partial_t W^* (\sqrt{N} \ph_t) \right] W(\sqrt{N} \ph_t) + W^* (\sqrt{N} \ph_t) \cH_N W(\sqrt{N} \ph_t) .\]
On the one hand we have
\begin{equation}\label{eq:GPL1} \left[ i\partial_t W^* (\sqrt{N} \ph_t) \right] W(\sqrt{N} \ph_t) = C_1 (N,t) - \sqrt{N} \left[ a^* (i\partial_t \ph_t) + a (i\partial_t \ph_t) \right] 
\end{equation}
for an unimportant constant $C_1 (N,t)$. On the other hand, using (\ref{eq:Wshift}), we find
\begin{equation}\label{eq:GPL2} \begin{split} 
W^* (\sqrt{N} \ph_t) &\cH_N W(\sqrt{N} \ph_t) \\ = &\; C_2 (N,t) + \sqrt{N} \left[ a^* (-\Delta_x \ph_t + (N^3 V(N\cdot) * |\ph_t|^2) \ph_t) + a ( -\Delta_x \ph_t + (N^3 V(N\cdot) * |\ph_t|^2) \ph_t) \right] \\ &+ \int dx \nabla_x a_x^* \nabla_x a_x + \int dx (N^3 V(N\cdot)*|\ph_t|^2) (x) a_x^* a_x \\ &+ \int dx dy N^3 V(N(x-y)) \ph_t (x) \overline{\ph}_t (y) a_x^* a_y \\ &+ \frac{1}{2} \int dx dy \, N^3 V(N (x-y)) \left[ \ph_t (x) \ph_t (y) a_x^* a_y^* + \overline{\ph}_t (x) \overline{\ph}_t (y) a_x a_y \right] \\ &+ \frac{1}{\sqrt{N}} \int dx dy N^3 V(N(x-y)) a_x^* (\ph_t (y) a_y^* + \overline{\ph}_t (y) a_y) a_x \\ &+ \frac{1}{2N} \int dx dy N^3 V(N(x-y)) a_x^* a_y^* a_y a_x \end{split} 
\end{equation}
where $C_2 (N,t)$ denotes another unimportant constant. Recall that in the mean-field regime the linear terms (linear in creation and annihilation operators) arising from the two contributions to the generator (\ref{eq:cLN}) canceled exactly due to the choice of $\ph_t$ as a solution of the Hartree equation. \emph{Here in the Gross-Pitaevskii regime, the cancellation is not complete}; in fact the creation and annihilation operators on the r.\,h.\,s.\ of (\ref{eq:GPL1}) have the argument
\[ i\partial_t \ph_t = -\Delta \ph_t + (N^3 V(N\cdot) f(N\cdot) *|\ph_t|^2) \ph_t \]
while the argument of the creation and annihilation operators appearing in the linear terms on the r.\,h.\,s.\ of (\ref{eq:GPL2}) does not contain the solution $f(N\cdot)$ of the zero-energy scattering equation (\ref{eq:scat}). It follows that 
\begin{equation}\label{eq:cLN-wrong} \begin{split} \cL_N (t) = \;& C(N,t) + \sqrt{N} \left[ a^* ((N^3 V(N\cdot) \omega(N\cdot) * |\ph_t|^2)\ph_t) + a ((N^3 V(N\cdot) \omega(N\cdot) * |\ph_t|^2)\ph_t) \right] \\ &+ \int dx \nabla_x a_x^* \nabla_x a_x + \int dx (N^3 V(N\cdot)*|\ph_t|^2) (x) a_x^* a_x \\ &+ \int dx dy N^3 V(N(x-y)) \ph_t (x) \overline{\ph}_t (y) a_x^* a_y \\ &+ \frac{1}{2} \int dx dy \, N^3 V(N (x-y)) \left[ \ph_t (x) \ph_t (y) a_x^* a_y^* + \overline{\ph}_t (x) \overline{\ph}_t (y) a_x a_y \right] \\ &+ \frac{1}{\sqrt{N}} \int dx dy N^3 V(N(x-y)) a_x^* (\overline{\ph}_t (y) a_y + \ph_t (y) a^*_y) a_x \\ &+ \frac{1}{2N} \int dx dy N^3 V(N(x-y)) a_x^* a_y^* a_y a_x \end{split} 
\end{equation}
where we set $\omega = 1 - f$ (and where $C(N,t) = C_1 (N,t) + C_2 (N,t)$ is another unimportant constant). Hence, in the Gross-Pitaevskii regime, the generator of the fluctuation dynamics (\ref{eq:fldyn-GP}) contains a contribution which is linear in creation and annihilation operators and therefore does not commute with the number of particles operator. This contribution is of order $\sqrt{N}$ (at least formally), and for this reason, it is impossible to bound the growth of the expectation of the number of particles operator (\ref{eq:grow-GP}) uniformly in $N$.

\medskip

\emph{Modified fluctuation dynamics.} The reason why this naive approach does not work is the fact that we are trying to approximate the many-body evolution of a  coherent state initial data by an evolved coherent state, completely neglecting the correlation among the particles. (Recall that coherent states in each $n$-particle component are completely factorized and thus do not have any correlations.) Since we know correlations to be very important in the Gross-Pitaevskii regime, it should not surprise us that fluctuations with respect to the evolved coherent state are too large.  

To take into account the \emph{correlation structure} developed by the time-evolution we define the kernel
\begin{equation}\label{eq:kt} k_t (x;y) = -N \omega (N (x-y)) \ph_t (x) \ph_t (y) 
\end{equation}
where, as before, $\ph_t$ denotes the solution of the modified Gross-Pitaevskii equation (\ref{eq:mod-GP}) and $\omega = 1- f$ (where $f$ is the solution of the zero-energy scattering equation (\ref{eq:scat})). Recall that $\omega (x) \simeq a_0/|x|$ for $|x| \gg 1$, while  in accordance with (\ref{eq:fprop}) it is regular for $|x| \ll 1$; it is useful to think of the function $N\omega (N(x-y))$ as $a_0 /|x-y|$ but regularized for $|x-y| \lesssim 1/N$. Using the kernel $k_t$ we define the unitary operator $T_t$, acting on the Fock space $\cF$, by
\begin{equation}\label{eq:Tt} T_t = \exp \left( \frac{1}{2} \int dx dy \left(k_t (x;y) a_x^* a_y^* - \overline{k}_t (x;y) a_x a_y \right) \right). 
\end{equation}
Since the exponent is quadratic in creation and annihilation operators, $T_t$ implements a Bogoliubov transformation. It is even possible to explicitly compute its action on the creation and annihilation operators, yielding 
\begin{equation}\label{eq:bogo-aa} \begin{split} T_t^* a (f) T_t &=  a (\text{cosh}_{k_t} f) + a^* (\text{sinh}_{k_t} \overline{f}) \\ T_t^* a^* (f) T_t &= a^* (\text{cosh}_{k_t} f) + a (\text{sinh}_{k_t} \overline{f}) \end{split} 
\end{equation}
for any $f \in L^2 (\bR^3)$. Here we use the notation $\cosh_{k_t}$ and $\sinh_{k_t}$ for the linear operators on $L^2 (\bR^3)$ given by
\[ \cosh_{k_t} = \sum_{n \geq 0} \frac{1}{(2n)!} (k_t \overline{k}_t)^n, \qquad \sinh_{k_t} = \sum_{n \geq 0} \frac{1}{(2n+1)!} (k_t \overline{k}_t)^n k_t, \]
where products of $k_t$ and $\overline{k}_t$ have to be understood as products of operators (we identify the function $k_t \in L^2 (\bR^3 \times \bR^3)$ with the operator having $k_t$ as its integral kernel). We use the unitary operator $T_t$ to approximate the correlation structure developed by the many-body evolution. We consider the evolution of initial data having the form
\begin{equation}\label{eq:PsiNW} \Psi_N = W(\sqrt{N} \ph) T_0 \xi_N 
\end{equation}
for a $\xi_N$ with only few particles (think of $\xi_N = \Omega$ for simplicity), and we try to approximate its evolution with an evolved state of the same form. As discussed in Appendix \ref{sec:gpenergy}, (\ref{eq:PsiNW}) is a natural class of initial data, approximating the ground state of the Bose-Einstein condensate trapped in a volume of order one (the point here is that the Bogoliubov transformation $T_0$ generates the correct correlation structure, which is crucial on the one hand to reach the ground state energy and, on the other hand, to 
follow the many-body dynamics). 

To approximate the evolution of (\ref{eq:PsiNW}), we define the new fluctuation vector $\xi_{N,t}$ requiring that
\[ e^{-i\cH_N t} W(\sqrt{N} \ph) T_0 \xi_N = W(\sqrt{N} \ph_t) T_t \xi_{N,t}. \]
Equivalently $\xi_{N,t} = \wt{\cU}_{N} (t;0) \xi$ with the modified fluctuation dynamics
\begin{equation}\label{eq:modi-flu} \wt{\cU}_N (t;s) = T_t^* W^* (\sqrt{N} \ph_t) e^{-i\cH_N (t-s)} W(\sqrt{N} \ph_s) T_s .
\end{equation}

Let us compute the kernel of the one-particle reduced density $\gamma_{N,t}^{(1)}$ associated with $\Psi_{N,t} = e^{-i\cH_N t} \Psi_N$:
\[\begin{split} 
\gamma_{N,t}^{(1)} (x;y) = \; & \langle \Psi_{N,t} , a_y^* a_x \Psi_{N,t} \rangle \\  = \; &\langle W(\sqrt{N} \ph_t) T_t \xi_{N,t} , a_y^* a_x W(\sqrt{N} \ph_t) T_t \xi_{N,t} \rangle  \\ = \; &\langle \xi_{N,t}, T_t^* \left( a_y^* + \sqrt{N} \overline{\ph}_t (y) \right) \left( a_x + \sqrt{N} \ph_t (x) \right) T_t \xi_{N,t} \rangle \end{split}. \]
Hence, for any one-particle observable $J$ on $L^2 (\bR^3)$, we find
\[ \tr\, J \left(\gamma_{N,t}^{(1)} - N |\ph_t \rangle \langle \ph_t| \right) = \sqrt{N} \langle \xi_{N,t} , T_t^* \left[ a^* (J\ph_t) + a (J \ph_t) \right] T_t \xi_{N,t} \rangle + \langle \xi_{N,t}, T_t^* d\Gamma (J) T_t \xi_{N,t} \rangle \]
and therefore
\[ \tr \left| \gamma_{N,t}^{(1)} - N |\ph_t \rangle \langle \ph_t| \right| \leq \sqrt{N} \langle \xi_{N,t}, T_t^* \cN^{1/2} T_t \xi_{N,t} \rangle  + \langle \xi_{N,t}, T^*_t \cN T_t \xi_{N,t} \rangle. \]
Using (\ref{eq:bogo-aa}), we conclude that \footnote{It is interesting to note that, while the introduction of $T_t$ changes the energy by a contribution of order $N$ (see Appendix \ref{sec:gpenergy}), Eq. (\ref{eq:trans}) shows that the change in the number of particles is only of order one.}
\begin{equation}\label{eq:transformN}
 T_t^* \cN T_t \leq C (\cN+1)
\end{equation}
and therefore that
\begin{equation}\label{eq:tr-bd-GP} \tr\, \left| \gamma_{N,t}^{(1)} - N |\ph_t \rangle \langle \ph_t| \right| \leq C \sqrt{N} \langle \xi_{N,t}, (\cN+1) \xi_{N,t} \rangle .
\end{equation}
This means that to get a bound on the rate of the convergence of the many-body evolution towards the Gross-Pitaevskii dynamics (proportional to $\sqrt{N}$), it is enough to control the growth of the number of particles with respect to the modified fluctuation dynamics, i.\,e.\ to control the expectation value
\begin{equation}\label{eq:wtUN} \langle \wt\cU_N (t;0) \xi , \cN \wt{\cU}_N (t;0) \xi \rangle 
\end{equation}
uniformly in $N$. 

\medskip

\emph{Generator of modified fluctuation dynamics.} To estimate the expectation (\ref{eq:wtUN}) we compute the generator of $\wt{\cU}_N (t;s)$, defined by the equation
\[ i\partial_t \wt{\cU}_N (t;s) = \wt{\cL}_N (t) \wt{\cU}_N (t;s) \qquad \text{with } \wt{\cU}_N (s;s) = 1 \]
and therefore given by 
\begin{equation}\label{eq:wtcL} \wt{\cL}_N (t) = \left( i\partial_t T_t^* \right) T_t + T_t^* \cL_N (t) T_t,  
\end{equation}
where $\cL_N (t)$ is the generator computed in (\ref{eq:cLN-wrong}). The first term on the r.\,h.\,s.\ of (\ref{eq:wtcL}) is an expression quadratic in creation and annihilation operators and can be bounded as
\[ \pm \left[ \left(i\partial_t T_t^* \right) T_t , \cN \right] \leq C(\cN+1).\]
Hence, its contribution to the growth of the number of particles operator can be controlled similarly as we bounded (\ref{eq:Ibd}) in the mean-field case. Let us now focus on the second term on the r.\,h.\,s.\ of (\ref{eq:wtcL}). As discussed between Eq.s \eqref{eq:GPL2} and \eqref{eq:cLN-wrong}, the generator $\cL_N (t)$ contains a large term (proportional to $\sqrt{N}$) linear in creation and annihilation operators. After conjugation with $T_t$, this term is given by
\begin{equation}\label{eq:lin-large} \sqrt{N} \int dx dy N^3 V(N (x-y)) \omega (N (x-y)) |\ph_t (y)|^2 \ph_t (x) T_t^* a_x^* T_t + \text{h.c.} 
\end{equation}
At the same time, the generator $\cL_N (t)$ \eqref{eq:cLN-wrong} contains a term which is cubic in creation and annihilation operators. After conjugating with $T_t$ we find this contribution to be, introducing the notation $\text{ch}_x (z) = \cosh_{k_t} (z;x)$ and $\text{sh}_x (z) = \text{sinh} (z;x)$, 
\begin{equation}\label{eq:cub1} \begin{split} 
\frac{1}{\sqrt{N}} & \int dx dy N^3 V(N(x-y)) \overline{\ph}_t (y) T_t^* a_x^* a_y a_x T_t + \text{h.c.}  \\ = \; &\frac{1}{\sqrt{N}} \int dx dy N^3 V(N (x-y)) \ph_t (y) T_t^* a_x^* T_t  (a (\text{ch}_y) + a^* (\text{sh}_y))(a (\text{ch}_x) + a^* (\text{sh}_x)) + \text{h.c.}  \\ = \; & \frac{1}{\sqrt{N}} \int dx dy N^3 V(N (x-y)) \ph_t (y) T_t^* a_x^* T_t  a^* (\text{sh}_y) (a (\text{ch}_x) + a^* (\text{sh}_x)) \\  &+ \frac{1}{\sqrt{N}} \int dx dy N^3 V(N (x-y)) \ph_t (y) T_t^* a_x^* T_t a (\text{ch}_y) a (\text{ch}_x) \\ &+ \frac{1}{\sqrt{N}} \int dx dy N^3 V(N (x-y)) \ph_t (y) T_t^* a_x^* T_t  a (\text{ch}_y) a^* (\text{sh}_x) \\ &+\text{h.c.}
\end{split} \end{equation}
In the last summand, the operators $a (\text{ch}_y)$ and $a^* (\text{sh}_x)$ are not in normal order (a product of creation and annihilation operators is said to be normal ordered if all creation operators are to the left of all annihilation operators). Putting them in normal order produces a commutator term which is linear in creation and annihilation operators. Since $\cosh_{k_t} = 1 + O(k_t)^2$ and $\sinh_{k_t} = k_t + O(k_t^3)$, we find 
\begin{equation}\label{eq:expandsinh} \langle \text{ch}_x, \text{sh}_y \rangle = \int dz \, \overline{\cosh_{k_t} (z;x)} \sinh_{k_t} (z;y) \simeq \int dz \delta (x-z) k_t (z;y) = k_t (x;y) \end{equation}
up to terms which are regular in the variable $x-y$ (more precisely, higher powers of $k_0$ have kernels that are regular on the diagonal and therefore their contribution to (\ref{eq:cub1}) can be shown to be negligible; in contrast, $k_t(x;y)$ behaves like $a_0 |x-y|^{-1}$ for $|x-y| \gg 1/N$). With this approximation, we find from (\ref{eq:cub1})
\begin{equation}\label{eq:cub2} \begin{split} \frac{1}{\sqrt{N}} \int dx dy &N^3 V(N(x-y)) \overline{\ph}_t (y) T_t^* a_x^* a_y a_x T_t + \text{h.c.}  \\ 
\simeq \; &  \frac{1}{\sqrt{N}} \int dx dy N^3 V(N (x-y)) k_t (x;y) \overline{\ph}_t (y) T_t^* a_x^* T_t + \cE + \text{h.c.}\\ 
\simeq \; & - \sqrt{N} \int dx dy N^3 V(N (x-y)) \omega (N(x-y)) |\ph_t (y)|^2 \ph_t (x) T_t^* a_x^* T_t+ \cE + \text{h.c.} 
\end{split}
\end{equation}
where $\cE$ includes all terms proportional to $T_t^* a_x^* T_t$ multiplied with a normal ordered quadratic expression in $a$ and $a^*$. It is easy to check that all terms in $\cE$ after expanding $T^*_t a^*_x T_t$ can be written in normal order up to negligible errors, and the normal ordered expression can be bounded using the number operator $\cN$. In contrast, the main term on the r.\,h.\,s.\ of (\ref{eq:cub2}) can not be bounded; however \emph{it cancels exactly the large linear contribution in (\ref{eq:lin-large})}. 

There is another important cancellation produced by the conjugation with the Bogolubov transformation $T_t$. On the one hand, we have the two quadratic contributions
\begin{equation}\label{eq:qu1} \begin{split} \int dx \nabla_x ( a^* (\text{ch}_x) &+ a (\text{sh}_x) \nabla_x (a (\text{ch}_x) + a^* (\text{sh}_x)) \\ &= \int dx dy \nabla_x a^*_x (\nabla_x k_t) (x,y) a^*_y + \cE_1
\\ &= - N^3 \int dx dy (\Delta \omega) (N(x-y)) \ph_t (x) \ph_t (y) a_x^* + \cE_2 \end{split} 
\end{equation}
and 
\begin{equation}\label{eq:qu2} \begin{split} 
\int dx dy N^3 V(N (x-y)) \ph_t (x) \ph_t (y) &(a^* (\text{ch}_x) + a (\text{sh}_x))(a^* (\text{ch}_y) + a (\text{sh}_y))  \\ &= \int dx dy N^3 V(N (x-y)) \ph_t (x) \ph_t (y) a_x^* a_y^* + \cE_3 \end{split} 
\end{equation}
where $\cE_1, \cE_2$ and $\cE_3$ denote error terms, which can be controlled by the expectation of the number of particles operator $\cN$ and of the kinetic energy operator 
\[ \cK = \int dx \nabla_x a_x^* \nabla_x a_x. \]
On the other hand, from the quartic term in (\ref{eq:cLN-wrong}) we find after conjugation with $T_t$
\begin{equation}\label{eq:qu3} \begin{split} \frac{1}{2N} \int dx dy \, N^3 V(N(x-y)) &(a^* (\text{ch}_x) + a (\text{sh}_x) (a^* (\text{ch}_y) + a (\text{sh}_y)) (a (\text{ch}_y) + a^* (\text{sh}_y)) (a (\text{ch}_x) + a^* (\text{sh}_x)) 
\\ &= \frac{1}{2N} \int dx dy N^3 V(N (x-y))  a^* (\text{ch}_x) a^* (\text{ch}_y) a (\text{ch}_y) a^* (\text{sh}_x) + \cE_4 \\ &= \int dx dy N^3 V(N (x-y)) k_t (x;y) a^*_x a^*_y  + \cE_5 \\
&= - \int dx dy N^3 V(N (x-y)) \omega (N(x-y)) \ph_t (x) \ph_t (y) a_x ^* a_y^* + \cE_5 \end{split} 
\end{equation}
where the error term $\cE_5$ can be bounded by the expectation of $\cN$, of $\cN^2/N$ and of the quartic potential term
\[ \int dx dy N^2 V(N (x-y)) a_x^* a_y^* a_y a_x \, . \]
Combining the first terms of the r.\,h.\,s.\ of (\ref{eq:qu1}), (\ref{eq:qu2}) and (\ref{eq:qu3}), we find
\begin{equation}\label{eq:canc2} N^3 \int dx dy \left[ \left(-\Delta + \frac{1}{2} V\right) (1-\omega) \right] (N(x-y)) \ph_t (x) \ph_t (y) a_x^* a_y^* = 0 
\end{equation}
because we chose $f=1-\omega$ as the solution of the zero-energy scattering equation (\ref{eq:scat}). 

Taking into account these two important cancellations, the generator of the modified fluctuation dynamics (\ref{eq:modi-flu}) can be controlled by the number of particles operator $\cN$, by $\cN^2/N$ and by the Hamiltonian (\ref{eq:HNF}), i.\,e.\ 
\begin{equation}\label{eq:wtLN-bd} \pm \wt\cL_N (t) \leq C e^{c|t|} \left( \cN + \frac{\cN^2}{N} + \cH_N \right) 
\end{equation}
for some $C, c>0$ (independent of $N$ and $t$).
The time dependence on the r.\,h.\,s.\ of the last equation arises through high Sobolev norms of the solution $\ph_t$ of the modified Gross-Pitaevskii equation (\ref{eq:mod-GP}). 

\medskip

\emph{Growth of fluctuations in the Gross-Pitaevskii regime.} The final step is to control the growth of the number of particles operator. In a very similar way as used to prove (\ref{eq:wtLN-bd}), we also obtain the two bounds   
\begin{equation}\label{eq:commGP1} \begin{split} \pm \left[ \wt\cL_N (t) , \cN \right] &\leq C e^{c|t|} \left( \cN + \frac{\cN^2}{N} + \cL_N (t) \right), \\
\pm \dot{\wt\cL}_N (t) &\leq C e^{c|t|} \left( \cN + \frac{\cN^2}{N} + \cL_N (t) \right). \end{split} 
\end{equation} 
Furthermore we have a simple bound where the number of fluctuations is just bounded by the total number of particles, the `worst case':
\[ \langle \wt\cU_N (t;0) \xi_N , \frac{\cN^2}{N} \wt\cU_N (t;0) \xi_N \rangle \leq \langle \wt\cU_N (t;0) \xi_N , \cN \wt\cU_N (t;0) \xi_N \rangle + \left\langle \xi_N , \frac{\cN^2}{N} \xi_N \right\rangle \]
With these ingredients the next proposition can be proven by Gronwall's Lemma.
\begin{proposition}\label{prop:growth-GP}
Let $V \geq 0$, spherically symmetric and short range. 
Let $\xi_N \in \cF$, with 
\[ \langle \xi_N, \cN \xi_N \rangle + \frac{1}{N} \langle \xi_N , \cN^2 \xi_N \rangle + \langle \xi_N , \cH_N \xi_N \rangle \leq C \]
for some $C>0$ (independent of $N$).
Then there exists constants $K,c > 0$ (independent of $N$ and $t$) such that 
\begin{equation}\label{eq:Nbd} \langle \wt\cU (t;0) \xi_N , \cN \wt\cU (t;0) \xi_N \rangle \leq K \exp (\exp (c |t|))
\end{equation}
for all $t \in \bR$.
\end{proposition}
The fast deterioration in time of the bound (\ref{eq:Nbd}) is a consequence of the time-dependence in (\ref{eq:wtLN-bd}) and (\ref{eq:commGP1}), arising from bounds for high Sobolev norms of $\ph_t$. If one assumes that $\| \ph_t \|_{H^4} \leq C$ uniformly in time, the bound (\ref{eq:Nbd}) becomes a simple exponential. 

\medskip

\emph{Quantitative convergence towards the Gross-Pitaevskii equation.} {F}rom (\ref{eq:tr-bd-GP}) and Proposition \ref{prop:growth-GP} we obtain the following theorem, which has been proved in \cite{BDS}.
\begin{theorem}
Let $V \geq 0$, spherically symmetric and short range. Let $\ph \in H^4 (\bR^3)$ and $\xi_N \in \cF$ such that 
\[ \langle \xi_N, \cN \xi_N \rangle + \frac{1}{N} \langle \xi_N , \cN^2 \xi_N \rangle + \langle \xi_N , \cH_N \xi_N \rangle \leq C \]
for some constant $C>0$ (independent of $N$). Consider the family of initial data
\[ \Psi_N = W(\sqrt{N} \ph) T_0 \xi_N \]
and denote by $\gamma_{N,t}^{(1)}$ the one-particle reduced density associated with the evolution $\Psi_{N,t} = e^{-iH_N t} \Psi_N$ of $\Psi_N$. Then
\[ \tr \, \left| \gamma^{(1)}_{N,t} - N |\ph_t \rangle \langle \ph_t| \right| \leq C N^{1/2} \exp (\exp (c |t|)) \]
for all $t \in \bR$.
\end{theorem}

\section{Mean-field regime for fermionic systems}
\setcounter{equation}{0}
\label{sec:mf_fermions}

So far we have discussed the time evolution of systems of interacting bosons. In this section, we will focus instead on \emph{fermionic systems}. We will consider systems of $N$ fermions initially confined to a volume of order one by a suitable external potential, e.\,g.\ an electromagnetic trap. We are interested in the evolution of such systems resulting from a change of the external field, e.\,g.\ switching off the trap. 

The Hamilton operator of the trapped $N$ fermion gas has the form 
\be\label{eq:fermi1}
H^{\text{trap}}_{N} = \sum_{j=1}^{N} \big( -\Delta_{x_{j}} + V_{\text{ext}}(x_{j}) \big) + \lambda\sum_{i<j}^{N}V(x_{i} - x_{j})\;.
\ee
It acts on $L^2_a (\bR^{3N})$, the subspace of $L^2 (\bR^{3N})$ consisting of functions which are antisymmetric w.\,r.\,t.\ permutation of the $N$ particles. The interaction potential $V$ varies on the same length scale as the one characterizing the confining potential $V_{\text{ext}}$. For this reason, each particle interacts with all the remaining $(N-1)$ particles, producing a potential energy of the order $\lambda N^2$. The mean-field regime is defined by choosing the coupling constant $\lambda$ so that the potential energy is typically of the same size as the kinetic energy. 
While the potential energy is independent of the statistics of the particles, the antisymmetry of the wave function plays an important role for the kinetic energy. In order to understand this point it is instructive to consider the following simple example.

\medskip

\emph{The free Fermi gas.} Consider a gas of $N$ non-interacting fermions in the three-dimensional torus $\mathbb{T}^{3}$. The Hamiltonian of the system is
\[
H_{N}^{\text{trap}} = \sum_{j=1}^{N} -\Delta_{x_j}\;,
\]
acting on $L^{2}_{a}(\mathbb{T}^{3N})$. The eigenstates of $H_{N}^{\text{trap}}$ can be computed explicitly; they are given by \emph{Slater determinants}:
\[
\psi_{\text{Slater}}(x_{1},\ldots, x_{N}) = \frac{1}{\sqrt{N!}} \det (f_i (x_j))_{1 \leq i,j \leq N} = \frac{1}{\sqrt{N!}}\sum_{\pi\in S_{N}}\s_{\pi} f_{1}(x_{\pi(1)})\cdots f_{N}(x_{\pi(N)})\;,
\]
where 
\[
f_{i}(x) \equiv f_{p_{i}}(x) = e^{ip_{i}\cdot x}\;,\qquad p_{i}\in 2\pi\mathbb{Z}^{3}\;.
\]
For an arbitrary choice of $N$ pairwise distinct momenta $p_{1},\ldots, p_{N} \in 2\pi \bZ^3$, the energy of the corresponding Slater determinant is given by $\sum_{i=1}^{N} p_{i}^{2}$. Thus, the ground state of $H_{N}^{\text{trap}}$ is obtained by choosing $N$ pairwise distinct vectors $\{ p_{i}\}_{i=1}^{N}$ in $2\pi \mathbb{Z}^{3}$ minimizing $\sum_{i=1}^{N}|p_{i}|^{2}$. The condition that all momenta must be distinct follows from the observation that the Slater determinant vanishes if two momenta coincide. This is an expression of the \emph{Pauli exclusion principle}, stating that there cannot be two fermions in the same one-particle state. It follows that the ground state energy of $H^{\text{trap}}_{N}$ is approximated by filling the \emph{Fermi ball} $\{ p\in 2\pi \mathbb{Z}^{3}, |p|\leq cN^{1/3} \}$, where $c = 2\pi/(\frac{4\pi}{3})^{1/3}$. We have (up to errors of lower order) 
\[
\begin{split}
\inf\text{spec}(H^{\text{trap}}_{N}) &\simeq \sum_{|p|\leq cN^{1/3}}|p|^{2} \simeq N^{5/3}\int_{|p|\leq c}dp\, |p|^{2}\qquad \mbox{for $N\gg 1$.}
\end{split}
\]
Hence the ground state energy, which is purely kinetic in this simple example, is of the order $N^{5/3}$ and thus \emph{much larger} than $N$. 

\medskip

\emph{Fermionic mean-field regime.} The fact that the kinetic energy of a Fermi gas scales as $N^{5/3}$ is a general property of $N$ fermion systems trapped in a volume of order one. This is made rigorous by the \emph{Lieb-Thirring kinetic energy inequality}, which states that for any $\psi\in L^{2}_{a}(\mathbb{R}^{3N})$
\begin{equation}\label{eq:LT}
\langle \psi, \sum_{i=1}^{N}-\Delta_{x_{i}} \psi \rangle \geq C\int \rho_{\psi}^{5/3}(x) dx\;,
\end{equation}
for a universal constant $C$. Here we defined the density $\rho_\psi$ associated with $\psi$ through 
\be
\rho_{\psi}(x) = N\int dx_{2}\cdots dx_{N}\, |\psi(x,x_{2},\ldots , x_{N})|^{2}\;.
\ee
Applying (\ref{eq:LT}) to the ground state $\psi_N$ of $H^{\text{trap}}_{N}$, we conclude that, for well-behaved densities $\rho_\psi$, 
\be
\langle \psi_{N}, \sum_{i=1}^{N}-\Delta_{x_{i}} \psi_{N} \rangle \geq C N^{5/3}
\ee
Therefore, in order to have a nontrivial theory in the limit $N\to \infty$, we are led to choose the coupling constant $\lambda$ in Eq.\ (\ref{eq:fermi1}) as $\lambda = N^{-1/3}$. In the same spirit, the external potential should also be scaled with $N$: taking $V_{\text{ext}}(x) = N^{2/3}v_{\text{ext}}(x)$ we make sure that all components of the Hamiltonian scale as $N^{5/3}$. This completes the definition of the fermionic mean-field regime. Notice that if
\[
V_{\text{ext}}(x) = \left\{ \begin{array}{ll} 0 \quad &\text{if } x\in\Lambda \\  \infty \quad &\text{otherwise} \end{array} \right.
\]
for a domain $\Lambda\in\mathbb{R}^{3}$ the rescaling of the external potential is trivial. 

\medskip

\emph{Hartree-Fock theory.} In this regime, the ground state of the system is expected to be well approximated by a Slater determinant $\psi_{\text{Slater}} (x_1, \dots , x_N) = N!^{-1/2} \det (f_i (x_j))$ with a suitable choice of the orthonormal system $\{ f_j \}_{j=1}^N$ in $L^2 (\bR^3)$. The one-particle reduced density associated with $\psi_{\text{Slater}}$ is given by  
\be
\omega_{N} = N\tr_{2,\ldots, N}| \psi_{\text{Slater}} \rangle \langle \psi_{\text{Slater}} | = \sum_{j=1}^{N}| f_{j} \rangle \langle f_{j} |\;.
\ee
Notice that $\omega_N$ is the orthogonal projection onto the subspace of $L^2 (\bR^3)$ spanned by the orbitals $\{ f_j \}_{j=1}^N$. Slater determinants are quasi-free states; that is, they are completely characterized by their one-particle reduced density $\omega_N$ in the sense that all higher order densities can be computed through $\omega_N$ using Wick's rule \eqref{eq:wicksrule}. In particular, the energy of a Slater determinant can be expressed in terms of only $\omega_N$, through the \emph{Hartree-Fock energy functional} 
\be\label{eq:HF}
\begin{split}
\mathcal{E}_{HF}(\omega_N) =& \; \langle \psi_{\text{Slater}}, H^{\text{trap}}_{N} \psi_{\text{Slater}} \rangle \\ =& \; \tr(-\Delta + V_{\text{ext}})\omega + \frac{1}{2N^{1/3}}\int dxdy\, V(x-y) \big[\omega(x;x)\omega(y;y) - |\omega(x;y)|^{2}\big]\;.
\end{split} \ee
Hence, we expect the one-particle reduced density of the ground state of (\ref{eq:fermi1}) to be well approximated, for $\lambda = N^{-1/3}$ and $N$ large, by the minimizer of (\ref{eq:HF}) among all orthogonal projections $\omega_N$ on $L^2 (\bR^3)$ 
with $\tr\, \omega_N = N$. 

\medskip

\emph{Thomas-Fermi theory.} The Hartree-Fock energy functional (\ref{eq:HF}) still depends on $N$. So what happens if we take the limit $N \to \infty$ here? We obtain the next, coarser, degree of approximation which is \emph{Thomas-Fermi theory}. More precisely, after dividing by $N^{5/3}$, the Hartree-Fock ground state energy is expected to be  close to the minimum of the Thomas-Fermi energy functional 
\be\label{eq:TF}
\mathcal{E}_{TF}(\rho) = \frac{3}{5}c_{TF}\int dx\, \rho(x)^{5/3} + \int dx\, v_{\text{ext}}(x)\rho(x) + \frac{1}{2}\int dxdy\, V(x-y) \rho(x)\rho(y)
\ee
over all densities $\rho \in L^{5/3}(\bR^3)$ with $\rho \geq 0$ and $\| \rho \|_{1} = 1$. The minimizer $\rho_{TF}$ satisfies the Thomas-Fermi equation 
\be\label{eq:TFeq}
\begin{split}
c_{TF}\rho_{TF}(x) &= (\mu - \phi_{TF}(x))_{+}^{3/2}
\end{split}
\ee
with the Thomas-Fermi potential $\phi_{TF} = v_{\text{ext}} + V*\rho_{TF}$ generated by the trap and by the self-consistent interaction produced by $\rho_{TF}$. In (\ref{eq:TFeq}) the parameter $\mu$ plays the role of the chemical potential; it has be chosen so that $\| \rho_{TF} \|_{1} = 1$. 

Notice that $\rho_{TF}$ should be interpreted as the (normalized) configuration space density of the fermions in the ground state. With it, we can approximate the minimizer of the Hartree-Fock energy functional (\ref{eq:HF}) by the Weyl quantization 
\begin{equation}\label{eq:weyl} \wt{\omega}_N (x;y) = \frac{1}{(2\pi \eps)^3} \int dv \,  M \left( \frac{x+y}{2}, v \right) e^{iv \cdot \frac{x-y}{\eps}} 
\end{equation}
of the characteristic function $M (x,v) = \chi (|v| \leq c \rho_{TF}^{1/3} (x))$ (the constant $c > 0$ is chosen so that $\tr \, \wt{\omega}_N = N$ and is independent of $N$) \footnote{However $\wt{\omega}_N$ is only approximately an projection.}. Physically (\ref{eq:weyl}) corresponds to the idea that in analogy to the case of free fermions, the ground state of (\ref{eq:HF}) (and of the corresponding many-body Hamiltonian (\ref{eq:fermi1})) can be approximated by filling the Fermi ball; here, however, this procedure is implemented locally, with the number of particles in the local Fermi ball dictated by the Thomas-Fermi density $\rho_{TF}$.

\medskip

\emph{Dynamics.} We consider initial data close to the trapped ground state of (\ref{eq:fermi1}) and its evolution resulting from a change of the external potential. For simplicity we assume that at time $t=0$ the external traps are just switched off. 
The subsequent evolution is governed by the Schr\"odinger equation 
\[
i\partial_{\tau}\psi_{N,\tau} = H_{N}\psi_{N,\tau}\;,\qquad \psi_{N,0} = \psi_{N}
\]
with
\be
H_{N} = \sum_{i=1}^{N} -\Delta_{x_{i}} + \frac{1}{N^{1/3}}\sum_{i<j}V(x_{i} - x_{j})\;.
\ee
We now identify the relevant time scale on which the system undergoes macroscopic changes.
{F}rom the discussion above, we expect the kinetic energy per particle to be of order $N^{2/3}$; hence, the classical velocity of the particles is typically of order $N^{1/3}$. Therefore the natural time scale of the evolution is of order $N^{-1/3}$. After rescaling time by introducing the variable $t = N^{1/3} \tau$, the Schr\"odinger equation takes the form
\be\label{eq:fermi2}
iN^{1/3}\partial_{t}\psi_{N,t} = \Big( \sum_{i=1}^{N} -\Delta_{x_{i}} + \frac{1}{N^{1/3}}\sum_{i<j}V(x_{i} - x_{j}) \Big)\psi_{N,t}\;.
\ee
Setting $\e = N^{-1/3}$ and multiplying the l.\,h.\,s.\ and the r.\,h.\,s.\ of (\ref{eq:fermi2}) by $\e^{2}$, we find 
\be\label{eq:fermi3}
i\e\partial_{t}\psi_{N,t} = \Big( \sum_{i=1}^{N} -\e^{2}\Delta_{x_{i}} + \frac{1}{N}\sum_{i<j}V(x_{i} - x_{j}) \Big)\psi_{N,t}\;.
\ee
In (\ref{eq:fermi3}) we recover the $N^{-1}$ coupling constant characterizing the mean-field limit of bosonic systems. Here, however, the mean-field limit is naturally linked with a \emph{semiclassical limit} where $\eps = N^{-1/3}$ plays the role of Planck's constant. 

We are interested in the evolution \eqref{eq:fermi3} for initial data approximating Slater determinants with reduced density $\omega_N$ minimizing (\ref{eq:HF}). We expect that $\psi_{N,t}$ stays close to a Slater determinant whose one-particle reduced density evolves according to the time-dependent Hartree-Fock equation 
\begin{equation}\label{eq:t-HF}
\begin{split}
i\e\partial_{t}\omega_{N,t} &= [h_{HF}(t), \omega_{N,t}]\;, \quad \text{with } \quad  h_{HF}(t) = -\e^{2}\Delta + V*\rho_{t} - X_{t}\;,
\end{split}
\end{equation}
and with initial data $\omega_{N,0}= \omega_N$. Here, $[A,B] = AB - BA$ denotes the commutator of the operators $A$ and $B$. Moreover, $\rho_t (x) = N^{-1} \omega_{N,t} (x;x)$ is the configuration space density associated with $\omega_{N,t}$, the convolution $V* \rho_t$, called the direct term, acts as a multiplication operator on $L^2 (\bR^3)$ and $X_t$ is the exchange term, an operator with integral kernel given by $X_t (x;y) = N^{-1} V(x-y) \omega_{N,t} (x;y)$. 

Like the Hartree-Fock energy functional, also the time-dependent Hartree-Fock equation (\ref{eq:t-HF}) still depends on $N$. As $N \to \infty$, it is expected to converge towards the classical Vlasov equation. More precisely, let us define the Wigner transform 
\be
W_{N,t}(x,v) = \frac{\e^{3}}{(2\pi)^{3}}\int dy\, \omega_{N,t}\Big( x - \frac{\e y}{2}; x + \frac{\e y}{2} \Big)e^{iv\cdot y}\;
\ee
of the reduced density $\omega_{N,t}$. As $N \to \infty$, $W_{N,t}$ is expected to converge (in an appropriate sense) to a probability density $W_{\infty,t}$ on phase space which satisfies the Vlasov equation 
\[ 
(\partial_{t} + v\cdot \nabla_{x})W_{\infty,t}(x,v) = (\nabla_{x}V* \rho^{\text{vl}}_{t})(x)\cdot \nabla_{v}W_{\infty,t}(x,v)\;,
\]
where $\rho^{\text{vl}}(x) = \int dv\, W_{\infty,t}(x,v)$ is the configuration space density associated to $W_{\infty,t}$. The convergence of the Hartree-Fock (or Hartree\footnote{The exchange term $X_t$ can be neglected with an error smaller than that of the step from many-body quantum mechanics to Hartree-Fock, see \cite{BPS1}.}) dynamics towards the Vlasov evolution has been proved to hold in a weak topology for sufficiently regular initial data $W_{N,0}$ and for a large class of potentials including the Coulomb potential in \cite{LP}. More recently, precise bounds on the rate of the convergence from the Hartree equation to the Vlasov equation for regular initial data and for regular potentials have been established in \cite{APPP,BPSS}. Notice, however, that in all these works the assumptions on the initial data exclude the case that $W_{N,0} (x,v) \simeq \chi (|v| \leq c \rho^{1/3} (x))$ for any reasonable density $\rho$, which would be interesting to consider for equilibrium states at zero temperature. 

\medskip

\emph{Semiclassical structure.} To establish the convergence of the many-body evolution to the time-dependent Hartree-Fock equation (\ref{eq:t-HF}), it is very important to take into account the structure of the initial Slater determinant. We are particularly interested in Slater determinants whose reduced density $\omega_N$ is chosen close to the minimizer of the Hartree-Fock energy (\ref{eq:HF}). In this situation we expect the integral kernel $\omega_N (x;y)$ to decay to zero for $|x - y|\gg \e$. At the same time, we expect it to be regular and to vary on scales of order one in the $x+y$ direction. 

To explain this point better, let us consider again the free Fermi gas on the torus. In this case, the one-particle reduced density of the ground state is given by 
\be
\omega_{N}(x;y) = \sum_{|p|\leq cN^{1/3}}e^{ip\cdot(x - y)}\simeq N\int_{|p|\leq c} dp\, e^{ip\cdot (x - y)/\e} = N\varphi\Big( \frac{x-y}{\e} \Big)\qquad \mbox{as $N\to\infty$,}
\ee
for some function $\varphi$ which can be explicitly calculated. As expected, $\omega_N$ varies on the short scale $\eps$ in the $x-y$ variable. In this elementary case the $(x+y)$-dependence is trivial because $V_{\text{ext}} (x) = 0$ in the box $\Lambda$. 

The formula (\ref{eq:weyl}) suggests that the approximation $\wt{\omega}_N$ of the minimizer of the Hartree-Fock energy functional (\ref{eq:HF}) for an interacting system in the mean-field regime exhibits the same semiclassical structure; it varies on the scale $\eps$ in the $(x-y)$-variable and on a scale of order one in its $(x+y)$-dependence. 

A convenient way to quantify the separation of scales  of $\wt{\omega}_N$ is given by estimates on the commutators
\be\label{eq:commutators}
[\wt\omega_N, x]\quad \text{ and } \quad [\wt\omega_N, \e\nabla]\;.
\ee
The integral kernel of $[\wt{\omega}_N,x]$ is given by 
\[ [\wt{\omega}_N ,x] (x;y) = (x-y) \wt{\omega}_N (x;y) = \frac{i\eps}{(2\pi)^3} \int dv \, (\nabla_v M) \left(
\frac{x+y}{2}, v \right) e^{iv \cdot \frac{x-y}{\eps}} \]
while
\[ [ \wt{\omega}_N, \eps \nabla ] (x;y) = -\eps \left( \nabla_x + \nabla_y \right) \wt{\omega}_N (x;y) = -\frac{\eps}{(2\pi)^3} \int dv  \,(\nabla_x M) \left( \frac{x+y}{2} , v \right) e^{iv \cdot \frac{x-y}{\eps}}. \]
Hence, taking $M(x,v) = \chi (|p| \leq c \rho^{1/3}_{TF} (x))$, semiclassical analysis suggests that 
\begin{equation}\label{eq:comm1} \tr \, \left|[\wt{\omega}_N ,x] \right| \lesssim C N \eps \int dx dv \left| (\nabla_v M) (x,v) \right| \leq C N\eps   
\int dx \rho^{2/3}_{TF} (x) \leq C N\eps 
\end{equation}
and 
\begin{equation}\label{eq:comm2} \tr\, \left| [\wt{\omega}_N,\eps \nabla] \right| \lesssim C N\eps \int dx |\nabla \rho_{TF} (x)| < C N\eps 
\end{equation}
under reasonable assumptions on $V$ and $V_{\text{ext}}$ (to guarantee that $\rho_{TF} \in L^{2/3} (\bR^3)$ and $\nabla \rho_{TF} \in L^1 (\bR^3)$ with $N$-independent norms). 

In the following we will study the time evolution of approximate Slater determinants with one-particle reduced density satisfying the commutator estimates (\ref{eq:comm1}) and (\ref{eq:comm2}). This assumption is motivated by the arguments given above, since we are interested in the evolution of initial data approximating the ground state of (\ref{eq:fermi1}) in the mean-field regime. 

\medskip

\emph{Rigorous results.} The mathematical literature concerning the derivation of effective evolution equations for interacting fermions is much smaller than for the bosonic case. In the mean-field limit we are interested in, the first derivation of the \emph{Vlasov equation} from quantum many-body dynamics has been given in \cite{NS} for analytic interaction potentials. Shortly after, this result has been extended in \cite{Sp} to twice differentiable potentials. The first derivation of the \emph{Hartree-Fock equation} has been given much later in \cite{EESY}, for analytic interaction potentials and short times. 

In the rest of this section we will present the results of \cite{BPS1}, where the convergence towards the Hartree-Fock equation has been established with control of the rate of convergence for a much larger class of potentials and for all times (notice that the results of \cite{BPS1} were also extended to the case of pseudorelativistic particles in \cite{BPS2}). 

Concerning other scaling limits, the Hartree-Fock equation has been derived in \cite{BGGM} and (for a Coulomb potential) \cite{FK} starting from a many-body Hamilton operator of the form (\ref{eq:fermi1}), but with $\lambda= N^{-1}$ (in this case, there is no link to a semiclassical limit).  Recently, convergence in presence of a regularized Coulomb interaction with a $N^{-2/3}$ coupling constant has been given in \cite{PP}.

\medskip

\emph{Fock space.} In order to discuss the results of \cite{BPS1} we need to introduce the fermionic Fock space over $L^2 (\bR^3)$, which is defined as the direct sum \[ \cF = \mathbb{C}\oplus\bigoplus_{n\geq 1}L^{2}_{a}(\mathbb{R}^{3n}) \] where $L^2_a (\bR^{3n})$ denotes the subspace of $L^2 (\bR^{3n})$ consisting of all functions antisymmetric with respect to permutations of the $n$ particles. 

As in the bosonic case, we let $\Omega = \{ 1, 0, 0, \dots \}$ denote the vacuum and define the number of particles operator by $(\cN \Psi)^{(n)} = n \psi^{(n)}$ for all $\Psi = \{ \psi^{(n)} \}_{n \in \bN} \in \cF$. Moreover, for all $f \in L^2 (\bR^3)$, we introduce the creation and annihilation operators $a^* (f)$ and $a(f)$ by setting
\be\label{eq:fermops}
\begin{split}
(a^{*}(f)\Psi)^{(n)}(x_{1},\ldots, x_{n}) &= \frac{1}{\sqrt{n}}\sum_{j=1}^{n}(-1)^{j} f(x_{j})\psi^{(n-1)}(x_{1},\ldots , x_{j-1}, x_{j+1},\ldots, x_{n})\;, \\
(a(f)\Psi)^{(n)}(x_{1},\ldots, x_{n}) &= \sqrt{n+1}\int dx\, \overline{f(x)}\psi^{(n+1)}(x,x_{1},\ldots, x_{n})\; , 
\end{split}
\ee
for $\Psi = \{ \psi^{(n)} \}_{n\in \bN} \in \cF$. It is easy to see that $a^{*}(f)$ is the adjoint of $a(f)$ and that fermionic creation and annihilation operators satisfy the  canonical anticommutation relations
\[
\{a^{*}(f), a^{*}(g)\} = \{a(f), a(g)\} = 0\;,\qquad \{a^{*}(f), a(g)\} = \langle g, f \rangle
\]
for all $f, g\in L^{2}(\mathbb{R}^{3})$. (The anticommutator is defined as $\{A,B\} = AB + BA$ for any two operators $A$ and $B$.) An important consequence of the canonical anticommutation relations is that $a(f)$ and $a^{*}(f)$ are \emph{bounded} operators (unlike in the bosonic case), since 
\be\label{eq:bounded}
\| a(f)\Psi \|^{2} = \langle \Psi, a^{*}(f)a(f) \Psi\rangle = -\langle \Psi, a(f)a^{*}(f) \Psi \rangle + \|f\|^{2}_{2} \leq \|f\|_{2}^{2}\qquad \forall \Psi\in\cF\;.
\ee
It is simple to check that $\| a^* (f) \| = \| a(f) \| = \|f\|_{2}$. 

As in the bosonic case, it is useful to introduce operator-valued distributions $a^{*}_{x}$ and $a_{x}$, which allow us to write  
\be
a(f) = \int dx\, \overline{f(x)} a_{x}\;,\qquad a^{*}(f) = \int dx\, f(x)a^{*}_{x}\;.
\ee
The second quantization of operators on the Fock space is defined exactly as for the bosonic case; see Section \ref{sec:cohe}.

\medskip

\emph{Fermionic Bogoliubov transformations.} It turns out that $N$-particle Slater determinants can be obtained in the Fock space $\cF$ by applying appropriate Bogoliubov transformations on the vacuum vector $\Omega$. To explain this point in more detail, let us define fermionic Bogoliubov transformations, similarly as we did in Section \ref{sec:fluc-bos} for the bosonic case. 

For $f,g\in L^{2}(\mathbb{R}^{3})$, we define 
 \[ A(f,g) = a(f) + a^* (\bar{g}) . \]
Observe that
\begin{equation}\label{eq:AA*-fer} A^* (f,g) = A (\mathcal{J}(f,g))
\end{equation}
where the operator $\mathcal{J}$ is defined by $\mathcal{J} (f,g) = (\bar{g}, \bar{f})$ (as in (\ref{eq:AA*})). In terms of these operators, the canonical anticommutation relations assume the form
\begin{equation}\label{eq:commA} \begin{split} 
\left\{ A (f_1, g_1) , A^* (f_2, g_2) \right\} &= \{ a (f_1) , a^* (f_2) \} + \{ a^* (\bar{g}_1) , a (\bar{g}_2) \} \\ &= \langle f_1, f_2 \rangle + \langle \bar{g}_2 , \bar{g}_1 \rangle = \langle f_1 , f_2 \rangle + \langle g_1 , g_2 \rangle \\ &= \langle (f_1, g_1) , (f_2, g_2) \rangle_{L^2 \oplus L^2}. \end{split} \end{equation}
A fermionic Bogoliubov transformation is a linear map $\nu: L^2 (\bR^3)  \oplus L^2 (\bR^3) \to L^2 (\bR^3) \oplus L^2 (\bR^3)$ preserving (\ref{eq:AA*-fer}) and (\ref{eq:commA}) in the sense that these properties continue to hold for the new field operators $B(f,g) := A (\nu (f,g))$. It is easy to see that $\nu$ is a fermionic Bogoliubov transformation if and only if $\nu$ is unitary and $\mathcal{J}\nu = \nu \mathcal{J}$. Equivalently, a linear map $\nu: L^2 (\bR^3)  \oplus L^2 (\bR^3) \to L^2 (\bR^3) \oplus L^2 (\bR^3)$ is a fermionic Bogoliubov transformation if and only if it has the form
\begin{equation}\label{eq:Nu} \nu =  \left( \begin{array}{ll} u & \overline{v} \\ v & \overline{u} \end{array} \right) \end{equation}
where $u,v: L^2(\mathbb{R}^3) \to L^2(\mathbb{R}^3)$ are linear maps with 
\be\label{eq:bogrelations}
u^* u + v^* v = 1\;, \qquad u^* \overline{v} + v^* \overline{u} = 0. 
\ee
As in the bosonic case, a fermionic Bogoliubov transformation $\nu$ is said to be implementable on $\cF$ if there exists a unitary operator $R_\nu: \cF \to \cF$ such that 
\begin{equation}\label{eq:Bog} R_\nu^* A(f,g) R_\nu = A (\nu (f,g)) \end{equation}
for all $f,g \in L^2 (\bR^3)$. By the Shale-Stinespring condition, the fermionic Bogoliubov transformation $\nu$ is implementable if and only if $v$ is a Hilbert-Schmidt operator (see e.\,g.\ \cite[Theorem 9.5]{Solovej} or \cite{Ru}). 

It is useful to see how the unitary $R_{\nu}$ transforms the fermionic operators and the corresponding operator-valued distributions. We have:
\be\label{eq:Bog2}
R_{\nu}^{*}a(f) R_{\nu} = R_{\nu}^{*}A(f,0)R_{\nu} = a(u f) + a^{*}(\overline{v} \overline{f})
\ee
for all $f\in L^{2}(\mathbb{R}^{3})$. Setting $a(f) = \int dx\, a_{x} \overline{f(x)}$, Eq.\ (\ref{eq:Bog2}) is equivalent to
\be\label{eq:trans}
R^{*}_{\nu}a_{x} R_{\nu} = a(u_{x}) + a^{*}(\overline{v}_{x})\;,
\ee
Here we introduced the notations $u_x (y) := u (y;x)$ and $v_x (y) = v(y;x)$, where $u(x;y)$ and $v(x;y)$ are the integral kernels of the operators $u$ and $v$. 

\medskip

\emph{Quasi-free states.} For any implementable Bogoliubov transformation $\nu : L^2 (\bR^3)  \oplus L^2 (\bR^3) \to L^2 (\bR^3) \oplus L^2 (\bR^3)$, the Fock space vector $R_\nu \Omega$ describes a quasi-free state; all correlation functions of $R_\nu \Omega$ can be computed through one-particle correlations by means of Wick's rule. In this sense a quasi-free state $\Psi \in \cF$ is completely determined by its one-particle reduced density $\gamma^{(1)}_N$, with the kernel 
\[ \gamma^{(1)}_N (x;y) = \langle \Psi, a_y^* a_x \Psi \rangle \]
and by its pairing density $\alpha^{(1)}_N$, defined by 
\[ \alpha^{(1)}_N (x,y) = \langle \Psi, a_y a_x \Psi \rangle \]
 
Let $\omega_N$ denote an orthogonal projection on $L^2 (\bR^3)$ with $\tr \, \omega_N = N$. Then we find an orthonormal system $\{ f_j \}_{j=1}^N$ in $L^2 (\bR^3)$ such that $\omega_N = \sum_{j=1}^N |f_j \rangle \langle f_j|$. We define $u_N = 1- \omega_N$ and $v_N = \sum_{j=1}^N |\bar{f}_j \rangle \langle f_j|$. Observe that $u_N^* u_N + v_N^* v_N = 1$ and $u_N^* \overline{v}_N = v_N^* \overline{u}_N = 0$. Hence (\ref{eq:bogrelations}) is satisfied and
\[ \nu_N =  \left( \begin{array}{ll} u_N & \overline{v}_N \\ v_N & \overline{u}_N \end{array} \right) \]
defines an implementable (since $\tr\, v_N^* v_N = \tr\, \omega_N = N < \infty$) Bogoliubov transformation. This implies that $R_{\nu_N} \Omega$ is a quasi-free state. Its one-particle reduced density has the integral kernel
\[ \begin{split} \gamma_N^{(1)} (x;y) &= \langle R_{\nu_N} \Omega, a_y^* a_x R_{\nu_N} \Omega \rangle \\ &= \langle \Omega, [ a^* (u_{N,y}) + a (\bar{v}_{N,y})] [a (u_{N,x}) + a^* (\bar{v}_{N,x})] \Omega \rangle \\ &= \langle \Omega, \{ a (\bar{v}_{N,y}), a^* (\bar{v}_{N,x}) \} \Omega \rangle = (v_N^* v_N) (x;y) \\ &= \omega_N (x;y) \end{split} \]
The pairing density of $R_{\nu_N} \Omega$ is given by
\[ \begin{split} \alpha^{(1)}_N (x,y) &= \langle R_{\nu_N} \Omega, a_y a_x R_{\nu_N} \Omega \rangle \\ &= \langle \Omega, [a(u_{N,y}) + a^{*}(\overline v_{N,y})][a(u_{N,x}) + a^{*}(\overline v_{N,x})] \Omega \rangle \\ &= \langle \Omega, \{ a(u_{N,y}), a^{*}(\overline v_{N,x}) \} \Omega \rangle \\ &= (\overline{v}_N u_N) (x;y) = 0
\end{split}
\]

Having the same one-particle reduced density and pairing density as the $N$-particle Slater determinant $\psi_{\text{Slater}} (x_1, \dots , x_N) = N!^{-1/2} \det (f_i (x_j))$, we conclude that up to a trivial phase $R_{\nu_N} \Omega = \{ 0, 0, \dots , \psi_{\text{Slater}}, 0, 0 \dots \}$. 

Alternatively the unitary operator $R_{\nu_{N}}$ can be introduced as a \emph{particle-hole transformation} on $\cF$. As above, let $\omega_N = \sum_{j=1}^N |f_j \rangle \langle f_j|$ for an orthonormal system $\{ f_j \}_{j=1}^N$. Moreover, let us complete $\{ f_j \}_{j \in \bN}$ to an orthonormal basis of $L^2 (\bR^3)$. Then, Eq.\ (\ref{eq:Bog2}) implies that
\be\label{eq:particle-hole}
R_{\nu_{N}} a(f_{i}) R_{\nu_{N}}^{*} = R^{*}_{\nu_{N}^{*}}a(f_{i})R_{\nu_{N}^{*}} = R^{*}_{\nu_{N}}a(f_{i})R_{\nu_{N}} = \left\{ \begin{array}{ll} a(f_{i})\quad &\text{if } i > N \\  a^{*}(f_{i}) \quad &\text{if } i \leq N \end{array} \right.. 
\ee
Being unitary, $R_{\nu_{N}}$ defines new fermionic creation and annihilation operators $b^* (f) = R_{\nu_{N}}a^*(f) R^{*}_{\nu_{N}}$ and $b(f) = R_{\nu_{N}}a(f) R^{*}_{\nu_{N}}$, whose vacuum is the Slater determinant $R_{\nu_{N}}\Omega$.

\medskip

\emph{Dynamics of quasi-free states.} As in the bosonic case, we define the Fock space Hamilton operator 
\be\label{eq:fockH}
\mathcal{H}_{N} = \eps^2 \int dx\, \nabla_{x}a^{*}_{x} \nabla_{x}a_{x} + \frac{1}{2N} \int dxdy\, V(x-y) a^{*}_{x}a^{*}_{y}a_{y}a_{x}\;.
\ee
Notice that on the $N$-particle sector of $\cF$, $\cH_N$ coincides with the Hamilton operator appearing on the r.\,h.\,s.\ of the Schr\"odinger equation (\ref{eq:fermi3}). The next theorem describes the time-evolution generated by (\ref{eq:fockH}) on approximate Slater determinants. 
\begin{theorem}\label{thm:HF}
Let $V\in L^{1}(\mathbb{R}^{3})$ such that \begin{equation}\label{eq:Vass0} \int dp\, |\hat{V}(p)|(1 + |p|^{2}) < \infty.
\end{equation}
Let $\omega_{N}$ be a sequence of orthogonal projections on $L^2 (\bR^3)$ with $\tr \, \omega_N = N$ and 
\be\label{eq:semiclassicalmain}
\tr\, |[x, \omega_{N}]| \leq KN\eps, \qquad \tr\, |[\e\nabla, \omega_{N}]| \leq KN\eps\;,
\ee
for some constant $K>0$. Let $\nu_N$ be the fermionic Bogoliubov transformation associated with the projection $\omega_N$ (as constructed in the last paragraph). 

Let $\xi_N$ be a sequence in $\cF$ with $\langle \xi_N , \cN \xi_N \rangle \leq C$ and denote by $\gamma^{(1)}_{N,t}$ the one-particle reduced density associated with the Fock space vector \[ \Psi_{N,t} = e^{-i\cH_{N} t/\eps} R_{\nu_{N}}\xi_N.\] Then there exist constants $C,c>0$ such that
\be\label{eq:mainresult}
\norm{\gamma^{(1)}_{N,t} - \omega_{N,t}}_{\text{HS}}^{2} \leq C\exp(c\exp c|t|) 
\ee
for all $t \in \bR$. Here $\omega_{N,t}$ is the solution of the time-dependent Hartree-Fock equation
\be\label{eq:HFmain}
i\e\partial_{t} \omega_{N,t} = [-\e^{2}\Delta + V*\rho_{t} - X_{t}, \omega_{t}]
\ee
with initial data $\omega_{N,0} = \omega_{N}$ (recall that $\rho_t (x) = N^{-1} \omega_N (x;x)$ and $X_t (x;y) = N^{-1} V(x-y) \omega_N (x;y)$). Under the additional assumptions that $\langle \xi_N, \cN^2 \xi_N \rangle \leq C$ and $d\Gamma (\omega_N) \xi_N = 0$, we also have the trace-norm bound
\begin{equation}\label{eq:main-trace} \tr\, |\gamma^{(1)}_{N,t} - \omega_{N,t}| \leq C N^{1/6}\exp(c\exp c|t|)\; .
\end{equation}
\end{theorem}

\emph{Remarks.} 
\begin{itemize}
\item The r.\,h.\,s.\ of (\ref{eq:mainresult}) and (\ref{eq:main-trace}) should be compared with the size of $\| \gamma^{(1)}_{N,t} \|_{\text{HS}} \simeq N^{1/2}$ and of $\tr \, \gamma_{N,t}^{(1)} = N$. In this sense, (\ref{eq:mainresult}) and (\ref{eq:main-trace}) prove that the Hartree-Fock equation gives a good approximation to the many-body evolution in the fermionic mean-field regime. 
\item The vector $\xi_N$ describes initial deviations from the Slater determinant. Taking $\xi_N = \Omega$, the initial data is exactly the $N$-particle Slater determinant with reduced density $\omega_N$. As long as $\langle \xi_N, \cN \xi_N \rangle \leq C$ uniformly in $N$, the initial data $R_{\nu_N} \xi_N$ is still close to a Slater determinant, at least in the sense of reduced densities (choosing $t=0$ in (\ref{eq:mainresult}) we observe  that $\| \gamma_{N,0}^{(1)} - \omega_N \|_{\text{HS}} \leq C$). 

\item It is simple to extend (\ref{eq:mainresult}) to $\xi_N \in \cF$ with $\langle \xi_N, \cN \xi_N \rangle \leq N^\alpha$ for any $\alpha < 1$ (if $\alpha = 1$, the deviations from the Slater determinant are as big as the Slater determinant). Under this assumption, (\ref{eq:mainresult}) is replaced by the weaker bound
\begin{equation}\label{eq:xi-weak} \| \gamma_{N,t}^{(1)} - \omega_{N,t} \|_{\text{HS}} \leq C N^{\alpha/2} \exp (c \exp (c |t|)) .
\end{equation}

\item It is possible to extend (\ref{eq:xi-weak}) to arbitrary $N$-particle initial data $\psi_N$ with reduced density close to the one of a Slater determinant in the trace norm topology. Let $\omega_N$ be a sequence of orthogonal projections on $L^2 (\bR^3)$ with $\tr\, \omega_N = N$, satisfying (\ref{eq:semiclassicalmain}). Let $\psi_N$ be a sequence of normalized functions in $L^2_a (\bR^{3N})$ with one-particle reduced density $\gamma_{N}^{(1)}$ satisfying 
\be\label{eq:tracenorm} \tr \, | \gamma_{N}^{(1)} - \omega_N | \leq C N^\alpha \ee
for some $\alpha < 1$. 
Set $\xi_N = R_{\nu_N}^* \{ 0, 0, \dots , \psi_N, 0, \dots \}$ and observe that, using (\ref{eq:trans}), 
\begin{equation}\label{eq:xiNxi} \langle \xi_N , \cN \xi_N \rangle =  \langle \psi_N, (\cN -2d\Gamma (\omega_N) +N) \psi_N \rangle\;. 
\end{equation}
Here, we identified $\psi_N$ with the Fock space vector $\{ 0, 0, \dots , \psi_N, 0, \dots \}$ and we used the identity
\begin{equation}\label{eq:rot-N} R_{\nu_N} \cN R_{\nu_N}^* = R_{\nu_N}^* \cN R_{\nu_N} = \cN -2d\Gamma (\omega_N) + N \end{equation}
which follows from (\ref{eq:particle-hole}), applied to $\cN = \sum_{j=1}^\infty a^* (f_j) a(f_j)$. {F}rom (\ref{eq:xiNxi}) we find
\[ 
\begin{split}
\langle \xi_N, \cN \xi_N \rangle &= 2 \langle \psi_N, d\Gamma (1-\omega_N) \psi_N \rangle\\
 &= 2 \tr \, \gamma_N^{(1)} (1-\omega_N) = 2\tr \, (\gamma_N^{(1)} - \omega_N) (1-\omega_N) \leq 2 \tr \, |\gamma_N^{(1)} - \omega_N| 
 \end{split}
 \]
where we used that $\omega_N$ is an orthogonal projection. Hence (\ref{eq:tracenorm}) implies $\langle \xi_N, \cN \xi_N \rangle \leq C N^{\alpha}$ for some $\alpha < 1$, and (\ref{eq:xi-weak}) holds. 
\item Theorem \ref{thm:HF} can be extended to prove the convergence of the higher order correlation functions $\gamma^{(k)}_{N,t}$ to $\omega^{(k)}_{N,t}$, defined by the integral kernel  
\begin{equation}\label{eq:wicksrule}
\omega^{(k)}_{N,t}(x_{1}, \ldots, x_{k}; y_{1},\ldots, y_{k}) = \sum_{\pi \in S_{k}}\s_{\pi} \prod_{j=1}^{k} \omega_{N,t}(x_{j}; y_{\pi(j)})\;. 
\end{equation}
(This is just the $k$-particle reduced density as it looks like for a quasifree state due to the Wick theorem.)
More details about this point can be found in \cite[Theorem 2.2]{BPS1}. 


\item It is well-known that the exchange term is subleading with respect to the direct one. In particular, for the class of potentials we consider in Theorem \ref{thm:HF}, dropping the exchange term does not change the estimates (\ref{eq:mainresult}). In fact, recalling that $X_t$ has the integral kernel $X_t (x;y) = N^{-1} V(x-y) \omega_{N,t} (x;y)$, we write
\begin{equation}\label{eq:deco-ex} 
X_t = \frac{1}{N} \int dp \, \widehat{V} (p) e^{ip \cdot x} \omega_{N,t} e^{-ip \cdot x} 
\end{equation}
and we find
\be\label{eq:Xtmain}
\begin{split}
\tr |[X_{t}, \omega_{N,t}]| &\leq \frac{1}{N}\int dp\, |\hat V(p)|\, \tr\,|[e^{ip\cdot x}\omega_{N,t}e^{-ip\cdot x}, \omega_{N,t}]| \\
&\leq \frac{C}{N} \int dp\,|\hat V(p)|\,\tr\,|[e^{ip\cdot x}, \omega_{N,t}]| \\
&\leq \frac{C}{N} \int dp\, |\hat V(p)||p|\, \tr\,|[x,\omega_{N,t}]|
\end{split}
\ee
where in the last step we used that $\tr\,|[e^{ip\cdot x}, \omega_{N,t}]|\leq |p|\tr\,|[x,\omega_{N,t}]|$. As we shall see later, the bounds (\ref{eq:semiclassicalmain}) can be propagated along the flow of the Hartree-Fock equation; the result is $\tr\,|[x,\omega_{N,t}]| \leq K N\eps\exp(c|t|)$, and hence
\be\label{eq:Xtmain2}
\tr |[X_{t}, \omega_{N,t}]| \leq \tilde C\eps \exp(c|t|) \int dp\,|\hat V(p)||p|\;.
\ee
Roughly speaking, (\ref{eq:Xtmain2}) implies that the contribution of the exchange term to the solution of the Hartree-Fock equation (\ref{eq:HFmain}) is of order $\varepsilon = N^{-1/3}$ (over times of order one), and hence comparable with (or of smaller order than) the r.\,h.\,s.\ of (\ref{eq:mainresult}) and (\ref{eq:main-trace}). The situation is expected to change in the presence of unbounded potentials, e.\,g.\ for a Coulomb potential; in this case, the analogy with the study of the ground state energy of large atoms \cite{B, GS} suggests that neglecting the exchange term will deteriorate the error estimates.
\end{itemize}

\medskip

\emph{Sketch of the proof of Theorem \ref{thm:HF}.} The proof follows a strategy conceptually similar to the coherent state technique discussed in the bosonic case in Section \ref{sec:cohe}. For simplicity, we shall only discuss the proof of convergence in Hilbert-Schmidt norm, that is of the inequality (\ref{eq:mainresult}). Let $\Psi_{N,t} = e^{-i\mathcal{H}_{N}t/\eps}R_{\nu_{N}}\xi_N$ be the evolution in Fock space of the initial state $R_{\nu_{N}}\xi_N$. The kernel of the one-particle reduced density $\gamma^{(1)}_{N,t}$ is
\be
\gamma^{(1)}_{N,t}(x;y) = \langle e^{-i\mathcal{H}_{N}t/\eps}R_{\nu_{N}}\xi_N , a^{*}_{y}a_{x} e^{-i\mathcal{H}_{N}t/\eps}R_{\nu_{N}}\xi_N \rangle\;.
\ee
The idea is to compare the quantum evolution $e^{-i\mathcal{H}_{N}t/\eps}R_{\nu_{N}}\Omega$ with the Hartree-Fock evolution, described in Fock space by $R_{\nu_{N,t}}\xi_N$. This is done by introducing the \emph{fluctuation dynamics} $\cU_{N}(t;s) := R^{*}_{\nu_{N,t}} e^{-i\mathcal{H}_{N}(t-s)/\eps}R_{\nu_{N,s}}$. Using the shorthand notation $u_{t,x} (y) = u_{N,t} (y;x)$ and $v_{t,x} (y) = v_{N,t} (y;x)$, we can rewrite 
the kernel of $\gamma^{(1)}_{N,t}(x;y)$ as
\be\label{eq:UN}
\begin{split}
\gamma^{(1)}_{N,t}(x;y) =\; &\langle \cU_{N}(t;0)\xi_N , R^{*}_{\nu_{N,t}}a^{*}_{y} R_{\nu_{N,t}}R^{*}_{\nu_{N,t}} a_{x} R_{\nu_{N,t}} \cU_{N}(t;0)\xi_N \rangle \\
= \; &\langle \cU_{N}(t;0)\xi_N, [a^{*}(u_{t,y}) + a(\overline{v}_{t,y})][a(u_{t,x}) + a^{*}(\overline{v}_{t,x})] \cU_{N}(t;0)\xi_N \rangle \\
= \; &\omega_{N,t}(x;y) \\ &+ \langle \cU_{N}(t;0)\xi_N , [ a^{*}(u_{t,y})a(u_{t,x}) + a^{*}(u_{t,y})a^{*}(\overline{v}_{t,x}) \\ &\hspace{3cm} + a(\overline{v}_{t,y})a(u_{t,x}) - a^{*}(\overline{v}_{t,x})a(\overline{v}_{t,y}) ] \cU_{N}(t;0)\xi_N \rangle\,.
\end{split} 
\ee
Here we used the transformation property (\ref{eq:trans}) and the fact that $\{ a^{*}(\overline{v}_{t,x}), a(\overline{v}_{t,y}) \} = \omega_{N,t}(x;y)$. For any Hilbert-Schmidt operator $O$ on $L^{2}(\mathbb{R}^{3})$ we find  
\be\label{eq:HFproof1}
\begin{split}
\tr\,O (\gamma^{(1)}_{N,t} - \omega_{N,t}) &= \langle \cU_{N}(t;0)\xi_{N}, \left[ d\Gamma(u_{N,t} O u_{N,t}) - d\Gamma(\overline v_{N,t} \overline{O} \overline{v}_{N,t})\right] \cU_{N}(t;0)\xi_{N}\rangle \\
&\quad + \Big[ \int dxdy\, O(x;y) \langle \cU_{N}(t;0)\xi_{N}, a(\overline{v}_{t,y})a(u_{t,x}) \cU_{N}(t;0)\xi_{N} \rangle + c.c. \Big].
\end{split}
\ee
Consider the first two terms in (\ref{eq:HFproof1}). 
Since for any bounded self-adjoint operator $A$ on $L^2 (\bR^3)$ we have the bound $\pm d\Gamma (A) \leq \| A \| \cN$, where $\| A \|$ is the operator norm of $A$, we find, using that $\| u_{N,t} \|\leq 1$ and $\| v_{N,t} \|\leq 1$,
\be\label{eq:HFerror1}
\Big|\langle \cU_{N}(t;0)\xi_N , \left[ d\Gamma(u_{N,t} O u_{N,t}) - d\Gamma(\overline v_{N,t} \overline{O} \overline{v}_{N,t}) \right] \cU_{N}(t;0)\xi_N \rangle\Big|\leq 2\| O \| \langle \cU_{N}(t;0)\xi_N, \cN\, \cU_{N}(t;0)\xi_N \rangle\;.
\ee
Consider now the second line of (\ref{eq:HFproof1}). In order to estimate it, we shall use the following bound. Let $A$ be a Hilbert-Schmidt operator on $L^{2}(\mathbb{R}^{3})$, with kernel $A(x;y)$. Then
\be\label{eq:HS1}
\begin{split}
\Big\| \int dxdy\, A(x;y) a_{x}a_{y}(\cN + 1)^{-1/2} \psi \Big\| &\leq \int dy \Big\| a(\overline{A}(\cdot; y))a_{y}(\cN + 1)^{-1/2} \psi\Big\| \\
&\leq \int dy\, \| A(\cdot; y) \|_{2}\big\| a_{y}(\cN + 1)^{-1/2} \psi\big\| \\
&\leq \| A \|_{\text{HS}} \int dy\, \big\| a_{y}(\cN + 1)^{-1/2}\psi \big\|\\
&= \|A\|_{\text{HS}}\big\| \cN^{1/2}(\cN + 1)^{-1/2} \psi\big\| \leq \| A \|_{\text{HS}}\|\psi\|\;.
\end{split}
\ee
where in the second line we used the boundedness of the fermionic creation/annihilation operators, Eq.\ (\ref{eq:bounded}). We find
\be\label{eq:HS2}
\Big\| (\cN + 1)^{-1/2} \int dxdy\, \overline{A(x;y)} a^{*}_{y}a^{*}_{x} \psi \Big\| = \Big\| \int dxdy\, \overline{A(x;y)} a^{*}_{y}a^{*}_{x} (\cN + 3)^{-1/2}\psi \Big\| \leq \|A\|_{\text{HS}}\|\psi\|\;.
\ee
The estimates (\ref{eq:HS1}), (\ref{eq:HS2}) can be used to bound the second line of (\ref{eq:HFproof1}). We get
\be\label{eq:HFerror2}
\begin{split} 
\Big| \int dxdy\, O(x;y) \langle \cU_{N}(t;0)\xi_N,& a(\overline{v}_{t,y})a(u_{t,x}) \cU_{N}(t;0)\xi_N \rangle + c.c. \Big|\\ &\leq C\| u_{N,t} O v_{N,t} \|_{\text{HS}} \big\| (\cN + 1)^{1/2} \cU_{N}(t;0)\xi_N \big\|\;.
\end{split} \ee
Since $\| u_{N,t} O v_{N,t} \|_{\text{HS}} \leq \| O \|_{\text{HS}}$ and $\| O \|\leq \| O \|_{\text{HS}}$, (\ref{eq:HFerror1}) and (\ref{eq:HFerror2}) imply
\be
\Big| \tr\, O(\gamma^{(1)}_{N,t} - \omega_{N,t}) \Big| \leq C \| O \|_{\text{HS}} \langle \cU_{N}(t;0)\xi_N, (\cN+1) \cU_{N}(t;0)\xi_N\rangle\;. 
\ee
In particular, choosing $O = \gamma^{(1)}_{N,t} - \omega_{N,t}$, we obtain  
\be\label{eq:HFHSbound}
\| \gamma^{(1)}_{N,t} - \omega_{N,t} \|_{\text{HS}} \leq C\langle \cU_{N}(t;0)\xi_N , (\cN+1) \cU_{N}(t;0)\xi_N \rangle\;.
\ee
Hence, to conclude the proof, we need to control the growth of the expectation of the number of particles operator with respect to the fluctuation dynamics $\cU_N (t;0)$. This is the content of the next proposition. 
\begin{proposition}\label{prop:HFgrowth}
Under the same assumptions as in Theorem \ref{thm:HF}, there exist constants $C, c >0$ such that
\be\label{eq:HFgrowth}
\langle \cU_{N}(t;0)\xi_N , (\cN + 1) \cU_{N}(t;0)\xi_N  \rangle \leq C\exp(c\exp c|t|) \langle \xi_N , (\cN+1) \xi_N \rangle .
\ee
\end{proposition}

\medskip

\emph{Controlling the growth of the fluctuations.} Here we present the main steps in the proof of the key bound Eq.\ (\ref{eq:HFgrowth}). It is based on a Gronwall-type argument. Using the identity (\ref{eq:rot-N}) we write 
\be\label{eq:HFgrowth2}
\begin{split}
\langle \cU_{N}(t;0)\xi_{N}, \cN \cU_{N}(t;0)\xi_{N}\rangle &= \langle e^{-i\cH_{N}t/\eps}R_{\nu_{N}}\xi_{N}, R_{\nu_{N,t}}\cN R_{\nu_{N,t}}^{*} e^{-i\cH_{N}t/\eps}R_{\nu_{N}}\xi_{N}\rangle \\
&= \langle e^{-i\cH_{N}t/\eps}R_{\nu_{N}}\xi_{N}, [\cN + N - 2d\Gamma(\omega_{N,t})] e^{-i\cH_{N}t/\eps}R_{\nu_{N}}\xi_{N}\rangle.
\end{split}
\ee
The expectation of $(\cN+N)$ is independent of time; these terms disappear when we take the derivative. Hence
\be\label{eq:3terms}
\begin{split}
i\e\frac{d}{dt}\langle \cU_{N} &(t;0)\xi_{N}, \cN \cU_{N}(t;0)\xi_{N}\rangle \\ &= -2\big\langle e^{-i\cH_{N}t/\eps}R_{\nu_{N}}\xi_{N}, \big( d\Gamma(i\e\partial_{t}\omega_{N,t}) - [\cH_{N}, d\Gamma(\omega_{N,t})]\big) e^{-i\cH_{N}t/\eps}R_{\nu_{N}}\xi_{N}\big\rangle \\
& = -2\big\langle \cU_{N}(t;0)\xi_{N}, R_{\nu_{N,t}}\big( d\Gamma(i\e\partial_{t}\omega_{N,t}) - [\cH_{N}, d\Gamma(\omega_{N,t})]\big)R_{\nu_{N,t}} \cU_{N}(t;0)\xi_{N}\big\rangle \\
& = -2\big\langle \cU_{N}(t;0)\xi_{N}, R_{\nu_{N,t}}\big( d\Gamma(\rho_{t}*V - X_{t}) - [\mathcal{V}_{N}, d\Gamma(\omega_{N,t})]\big)R_{\nu_{N,t}} \cU_{N}(t;0)\xi_{N}\big\rangle
\end{split}
\ee
where \[ \cV_{N} = \frac{1}{2N} \int dx dy \, V(x-y) a^*_x a_y^* a_y a_x \] is the many-body interaction. In the last step we used that the contribution of the kinetic energy cancels, as follows from the identity $d\Gamma([A,B]) = [d\Gamma(A), d\Gamma(B)]$. Applying the Bogoliubov transformation $R_{\nu_{N,t}}$ and reorganizing all terms in normal order we end up with 
\be\label{eq:HFgenerator}
\begin{split}
i \e \frac{d}{dt} \, \Big\langle &\cU_N (t;0) \xi_{N} , \cN \cU_N (t;0) \xi_{N}\Big\rangle \\ &= \frac{4 i}{N}\, \text{Im} \int  dx dy\, V(x-y) \\ &\quad\times \Big\{ \Big\langle \cU_N (t;0) \xi_{N}, a^*(u_{t,x}) a(\overline{v}_{t,y}) a(u_{t,y}) a(u_{t,x}) \cU_N(t;0)\xi_{N}\Big\rangle \\ & \quad \quad + \Big\langle \cU_N (t;0) \xi_{N}, a(\overline{v}_{t,x}) a(\overline{v}_{t,y}) a(u_{t,y}) a(u_{t,x}) \cU_N (t;0)\xi_{N}\Big\rangle \\ &\quad \quad + \Big\langle \cU_N (t;0) \xi_{N}, a^*(u_{t,y}) a^*(\overline{v}_{t,y}) a^*(\overline{v}_{t,x}) a(\overline{v}_{t,x}) \cU_N (t;0)\xi_{N}\Big\rangle \Big\}. 
\end{split} 
\ee
Consider for example the second term on the r.\,h.\,s.\ of (\ref{eq:HFgenerator}). Expanding the potential in Fourier space, we write 
\[
\begin{split}
\frac{1}{N} &\int dxdy\, V (x-y) \Big\langle \cU_{N}(t;0)\xi_{N}, a(\overline{v}_{t,x}) a(\overline{v}_{t,y}) a(u_{t,y}) a(u_{t,x}) \cU_{N}(t;0)\xi_{N}\rangle \\ =\; & \frac{1}{N}\int dp\,\hat V(p) \\
&\hspace{.5cm} \times \int dr_{1}dr_{2}dr_{3}dr_{4} (v_{t}e^{-ip\cdot x}u_{t} )(r_{1}; r_{2}) (v_{t}e^{ip\cdot x} u_{t})(r_{3}; r_{4}) \Big\langle \cU_{N}(t;0)\xi_{N}, a_{r_{1}}a_{r_{2}}a_{r_{3}}a_{r_{4}} \cU_{N}(t;0)\xi_{N} \Big\rangle .
\end{split}
\]
Inserting $1 = (\cN + 5)^{1/2}(\cN + 5)^{-1/2}$ and applying the Cauchy-Schwarz inequality, we get 
\[
\begin{split}
&\Big| \frac{1}{N}\int dxdy\, V(x-y) \Big\langle \cU_{N}(t;0)\xi_{N}, a(\overline{v}_{t,x}) a(\overline{v}_{t,y}) a(u_{t,y}) a(u_{t,x}) \cU_{N}(t;0)\xi_{N}\rangle \Big| \\
&\leq \frac{1}{N}\int dp\,|\hat V(p)|\big\| (\cN + 5)^{1/2} \cU_{N}(t;0)\xi_{N} \big\| \\
&\quad \times \Big\| (\cN + 5)^{-1/2}\int dr_{1}dr_{2}(v_{N,t}e^{-ip\cdot x}u_{N,t} )(r_{1}; r_{2}) a_{r_{1}}a_{r_{2}} \\ &\hspace{3cm} \times \int dr_{3}dr_{4}(v_{N,t}e^{ip\cdot x} u_{N,t})(r_{3}; r_{4})a_{r_{3}}a_{r_{4}} \cU_{N}(t;0)\xi_{N}  \Big\|\\
&= \frac{1}{N}\int dp\,|\hat V(p)|\big\| (\cN + 5)^{1/2} \cU_{N}(t;0)\xi_{N} \big\| \\
&\quad \times \Big\| \int dr_{1}dr_{2}(v_{N,t}[e^{-ip\cdot x},u_{N,t}] )(r_{1}; r_{2}) a_{r_{1}}a_{r_{2}} \\ &\hspace{3cm} \times \int dr_{3}dr_{4}(v_{N,t}[e^{ip\cdot x}, u_{N,t}])(r_{3}; r_{4})a_{r_{3}}a_{r_{4}} (\cN + 1)^{-1/2}\cU_{N}(t;0)\xi_{N}  \Big\|\\
\end{split}
\]
where in the last step we used that $\cN a(f) = a(f)(\cN - 1)$ and the orthogonality $v_{t}u_{t} = 0$. Applying the estimate (\ref{eq:HS1}) twice and using that $\| v_{N,t} \|\leq 1$ we get:
\be\label{eq:HFgenerator2}
\begin{split}
\Big| \frac{1}{N}\int dxdy\, &V(x-y) \Big\langle \cU_{N}(t;0)\xi_{N}, a(\overline{v}_{t,x}) a(\overline{v}_{t,y}) a(u_{t,y}) a(u_{t,x}) \cU_{N}(t;0)\xi_{N}\rangle \Big| \\ 
&\leq \frac{1}{N} \int dp\, |\hat V(p)| \| [e^{ip\cdot x}, \omega_{N,t}] \|_{\text{HS}}^{2} \Big\langle \cU_{N}(t;0)\xi_{N}, \cN \cU_{N}(t;0)\xi_{N} \Big\rangle .
\end{split}
\ee
Using the naive bound $\| [e^{ip\cdot x}, \omega_{N,t}] \|_{\text{HS}}^{2}\leq 2\| [e^{ip\cdot x}, \omega_{N,t}] \|_{\text{tr}}\leq 4N$ would not be enough here because of the factor $\eps$ on the l.\,h.\,s.\ of (\ref{eq:HFgenerator}). Instead, we have to propagate the semiclassical estimates (\ref{eq:semiclassicalmain}) along the solution of the Hartree-Fock equation. This is the content of Proposition \ref{thm:propagsc} below, which tells us that $\| [e^{ip \cdot x}, \omega_{N,t} ] \|_{\text{tr}} \leq C N \e (1+|p|) \exp (c|t|)$ for all times $t \in \bR$. Inserting this estimate in (\ref{eq:HFgenerator2}), we conclude that
\[ \begin{split}  
\Big| \frac{1}{N}\int dxdy\, V(x-y) &\Big\langle \cU_{N}(t;0)\xi_{N}, a(\overline{v}_{t,x}) a(\overline{v}_{t,y}) a(u_{t,y}) a(u_{t,x}) \cU_{N}(t;0)\xi_{N}\rangle \Big| \\ &\hspace{4cm} \leq C \eps \exp (c |t|) \langle \cU_N (t;0)\xi_{N}, \cN \cU_{N}(t;0)\xi_{N} \Big\rangle.
\end{split} \]

The other terms in (\ref{eq:HFgenerator}) can be bounded in a similar way. The only difference is that instead of (\ref{eq:HS1}) one has to use the estimate 
\be
\Big\| \int dr_{1}d r_{2}\, A(r_{1};r_{2})a^{\sharp}_{r_{1}}a_{r_{2}} \varphi \Big\| \leq \|A\|_{\text{tr}}\|\varphi\| \quad \forall \varphi \in \cF
\ee
valid for any trace-class operator $A$ on $L^{2}(\mathbb{R}^{3})$ with kernel $A(r_{1};r_{2})$. The details can be found in \cite{BPS1}. We conclude that 
\be
\Big| \frac{d}{dt} \, \Big\langle \cU_N (t;0) \Omega , (\cN + 1)\, \cU_N (t;0) \Omega\Big\rangle \Big| \leq C  \exp(c|t|) \Big\langle \cU_N (t;0) \Omega , (\cN + 1)\, \cU_N (t;0) \Omega\Big\rangle\;.
\ee
Gronwall's lemma implies the desired estimate (\ref{eq:HFgrowth}). 

\medskip

\emph{Propagation of the semiclassical structure.} An important ingredient that we used above to control the growth of the expectation of the number of particles with respect to the fluctuation dynamics is the propagation of the semiclassical commutator bounds along the solution of the Hartree-Fock equation. 
\begin{proposition}\label{thm:propagsc}
Let $V$ satisfy (\ref{eq:Vass0}). Let $\omega_N$ be a sequence of orthogonal projections with $\tr \, \omega_N = N$, such that 
\begin{equation}\label{eq:prop-stat} \tr \, |[x,\omega_N]| \leq KN\eps \quad \text{and } \quad \tr \, |[\eps \nabla,\omega_N]| \leq KN\eps 
\end{equation}
for some constant $K > 0$. Let $\omega_{N,t}$ be the solution of the Hartree-Fock equation (\ref{eq:HFmain}) with initial data $\omega_{N,0} = \omega_N$. Then there exist constants $C,c > 0$ such that 
\[ \tr \, |[x,\omega_{N,t}]| \leq C N\eps \exp (c |t|) \quad \text{and } \quad \tr \, |[\eps \nabla, \omega_{N,t}]| \leq CN\eps \exp (c|t|) \]
for all $t \in \bR$. 
\end{proposition} 

\emph{Sketch of the proof of Proposition \ref{thm:propagsc}.} Let $h_{HF}(t) = -\e^{2}\Delta + V*\rho_{t} - X_{t}$ be the Hartree-Fock Hamiltonian. We compute
\[
i\e\frac{d}{dt}[x, \omega_{N,t}] = [x, [h_{HF}(t), \omega_{N,t}]] = [\omega_{N,t}, [h_{HF}(t), x]] + [h_{HF}(t), [x, \omega_{N,t}]].
\]
The last term can be eliminated by conjugating $[x, \omega_{N,t}]$ with the two-parameter group of unitary transformations $W(t;s)$ defined by 
\[
i\e\frac{d}{dt} W(t;s) = h_{HF}(t) W(t;s)\qquad \mbox{with $W(s;s) = 1$ for all $s\in\mathbb{R}$}.
\]
We have
\[
\begin{split}
i\e\frac{d}{dt} W^{*}(t;0) [x, \omega_{N,t}] W(t;0) & = W^{*}(t;0) [\omega_{N,t}, [h_{HF}(t), x]] W(t;0) \\
& = W^{*}(t;0)\Big( [\omega_{N,t}, -2\e^{2}\nabla] - [\omega_{N,t}, [X_{t},x]] \Big)W(t;0)
\end{split}
\]
where we used the identities $[-\e^{2}\D,x]= - 2\e^{2}\nabla$ and $[\rho_{t}*V, x] = 0$. Therefore we get the Duhamel-type formula
\[
\begin{split}
W^{*}(t;0) [x, \omega_{N,t}] W(t;0) 
& = [x, \omega_{N,0}] - \frac{1}{i\e}\int_{0}^{t}ds\, W^{*}(s;0) \Big( [\omega_{N,s}, 2\e^{2}\nabla] + [\omega_{N,s}, [X_{s},x]] \Big) W(s;0).
\end{split}
\]
This implies that
\be
\| [x, \omega_{N,t}] \|_{\text{tr}} \leq \| [x, \omega_{N,0}] \|_{\text{tr}} + \frac{1}{\e}\int_{0}^{t}ds\,\Big( \| [\omega_{N,s}, 2\e^{2}\nabla] \|_{\text{tr}} + \|[\omega_{N,s}, [X_{s},x]]\|_{\text{tr}} \Big).\label{p1}
\ee
To control the second term we use (\ref{eq:deco-ex}). Since $\| \omega_{N,s} \| \leq 1$ we find
\be\label{p2}
\begin{split}
\|[\omega_{N,s}, [X_{s},x]]\|_{\text{tr}} & \leq \frac{1}{N}\int dq\, |\hat V(q)|\, \|[\omega_{N,s}, [e^{ip\cdot x} \omega_{N,s} e^{-ip\cdot x}, x]]\|_{\text{tr}} \\
& \leq \frac{1}{N}\int dq\, |\hat V(q)|\, \|[\omega_{N,s}, e^{ip\cdot x}[\omega_{N,s}, x]e^{-ip\cdot x}]\|_{\text{tr}} \\
& \leq \frac{2}{N}\int dq\, |\hat V(q)|\, \|[\omega_{N,s}, x]\|_{\text{tr}} \\
& \leq \frac{C}{N}\| [\omega_{N,s}, x] \|_{\text{tr}}.
\end{split}
\ee
Inserting (\ref{p2}) in (\ref{p1}) we get
\be
\| [x, \omega_{N,t}] \|_{\text{tr}} \leq \| [x, \omega_{N,0}] \|_{\text{tr}} + C \int_{0}^{t}ds\,\Big( \| [\omega_{N,s},  \e\nabla] \|_{\text{tr}} +  N^{-2/3}\|[\omega_{N,s}, x]\|_{\text{tr}} \Big). \label{p1ba}
\ee
To control $[\omega_{N,s},\e\nabla]$ we start by writing
\[
\begin{split}
i\e\frac{d}{dt}[\e\nabla, \omega_{N,t}] & = [\e\nabla, [h_{HF}(t), \omega_{N,t}]] \\ & = [\omega_{N,t}, [h_{HF}(t), \e\nabla]] + [h_{HF}(t), [\e\nabla, \omega_{N,t}]] \\ 
& = [h_{HF}(t), [\e\nabla, \omega_{N,t}]] + [\omega_{N,t}, [\rho_{t}*V, \e\nabla]] - [\omega_{N,t}, [X_{t}, \e\nabla]].
\end{split}
\]
As before, the first term can be eliminated by conjugation with the unitary operator $W(t;0)$. We find
\be
\| [\e\nabla, \omega_{N,t}] \|_{\text{tr}} \leq \| [\e\nabla, \omega_{N,0}] \|_{\text{tr}} + \frac{1}{\e}\int_{0}^{t}ds\, \Big( \| [\omega_{N,s}, [\rho_{s}*V, \e\nabla]] \|_{\text{tr}} + \| [\omega_{N,s}, [X_{s}, \e\nabla]]\|_{\text{tr}} \Big).\label{p4}
\ee
The first term in the integral can be controlled by 
\be\label{p1bb}
\begin{split}
\| [\omega_{N,s}, [V*\rho_{s}, \e\nabla]] \|_{\text{tr}} & = \e \| [\omega_{N,s}, \nabla V * \rho_{s}] \|_{\text{tr}}\\
& \leq \e\int dq\, |\hat V(q)|\, |q|\, |\hat \rho_{s}(q)|\, \| [\omega_{N,s}, e^{iq\cdot x}] \|_{\text{tr}}\\
&\leq \e\int dq\, |\hat V(q)|\, |q|^{2}\, |\hat \rho_{s}(q)|\, \| [\omega_{N,s}, x] \|_{\text{tr}} \\
& \leq C\e \| [\omega_{N,s}, x] \|_{\text{tr}}
\end{split}
\ee
where we used the bound $\| \hat \rho_{s} \|_{\infty} \leq \|\rho_{s}\|_{1} = 1$ and the assumption on the interaction potential. Finally, consider the second term in the integral on the r.\,h.\,s.\ of (\ref{p4}). Again using (\ref{eq:deco-ex}), we write 
\be\label{p4b}
\begin{split}
\| [\omega_{N,s}, [X_{s}, \e\nabla]] \|_{\text{tr}} &\leq \frac{1}{N}\int dq\,|\hat V(q)|\, \| [\omega_{N,s}, [e^{iq\cdot x}\omega_{N,s}e^{-iq\cdot x}, \e\nabla]] \|_{\text{tr}} \\
&\leq \frac{2}{N}\int dq\,|\hat V(q)|\,\| [e^{iq\cdot x}\omega_{N,s}e^{-iq\cdot x}, \e\nabla ] \|_{\text{tr}} \\
&\leq \frac{2}{N}\int dq\,|\hat V(q)|\,\| [\omega_{N,s}, \e\nabla ] \|_{\text{tr}}
\end{split}
\ee
where we used the identity 
\[
[e^{iq\cdot x}\omega_{N,s}e^{-iq\cdot x},\e\nabla] = e^{iq\cdot x}[\omega_{N,s}, \e(\nabla + iq)]e^{-iq\cdot x} = e^{iq\cdot x}[\omega_{N,s}, \e\nabla]e^{-iq\cdot x}\;.
\]
Inserting the estimates (\ref{p1bb}) and (\ref{p4b}) into  (\ref{p4}), we get
\[
\| [\e\nabla, \omega_{N,t}] \|_{\text{tr}} \leq \| [\e\nabla, \omega_{N,0}] \|_{\text{tr}} + C\int_{0}^{t}ds\,\Big( \| [\omega_{N,s}, x] \|_{\text{tr}} + N^{-2/3}\| [\omega_{N,s}, \e\nabla] \|_{\text{tr}} \Big).
\]
Combining this inequality with (\ref{p1ba}), using the assumptions on the initial data and applying Gronwall's lemma, we obtain (\ref{eq:prop-stat}).

\section{Dynamics of quasi-free mixed states}
\setcounter{equation}{0}
\label{sec:mixed}

In Section \ref{sec:mf_fermions} we discussed the evolution of initial data approximating Slater determinants. Slater determinants are relevant at zero temperature because they provide (or at least they are expected to provide) a good approximation to the fermionic ground state of Hamilton operators like (\ref{eq:fermi1}) in the mean-field limit. At positive temperature, equilibrium states are mixed; in the mean-field regime, they are expected to be approximately quasi-free mixed states. 

\medskip 

\emph{Mixed states.} We introduce the short-hand notation $\h = L^{2}(\mathbb{R}^{3})$. We shall denote by $\cF(\h)$ the fermionic Fock space built over $\h$, that is $\cF(\h) = \bigoplus_{n\geq 0}\h^{\wedge n}$. A general fermionic state is represented by a density matrix on $\cF (\h)$. A density matrix is a non-negative trace class operator $\rho:\cF (\h) \to \cF(\h)$ with $\tr \, \rho = 1$. Notice that the state described by the density matrix $\rho$ is pure if and only if $\rho$ is a rank-one orthogonal projection onto a $\psi \in \cF (\h)$, i.\,e.\ $\rho = |\psi \rangle \langle \psi|$. Otherwise the state is called a mixed state. In general, 
\be\label{eq:densma}
 \rho = \sum_n \lambda_n |\psi_{n} \rangle \langle \psi_{n}| 
 \ee
where $\lambda_n \geq 0$, $\{\psi_n\}$ is an orthonormal family in ${\cal F}(\h)$, and  $\sum_{n} \lambda_n = 1$. Physically, $\rho$ describes an \emph{incoherent} mixture of pure states and $\lambda_n$ is the probability that the system is in the state $\psi_n$. The expectation of an arbitrary operator $A$ on $\cF (\h)$ in the mixed state with density matrix $\rho$ is given by 
\[ \tr \, A \rho = \sum_n \lambda_n \langle \psi_n, A \psi_n \rangle. \] 

Given the density matrix (\ref{eq:densma})
we define the operator $\wt{\kappa}:\cF(\h) \to \cF(\h)$ by 
\[
\wt{\kappa} = \sum_{n} \e_{n}| \psi_{n} \rangle \langle \phi_{n} |
\]
where  $\e_{n}\in \mathbb{C}$ is a sequence satisfying  $|\e_{n}|^{2} = \lambda_{n}$ and $\{\phi_{n}\}$ is an orthonormal family in ${\cal F}(\h)$. Clearly,
\[
\wt{\kappa} \wt{\kappa}^{*} = \rho.
\]
Of course such a decomposition of $\rho$ is far from being unique and later we shall choose a convenient one. Since $\rho$ is trace class it follows that $\wt{\kappa} \in \cL^2 (\cF (\h))$, the set of Hilbert-Schmidt operators on $\cF (\h)$. 

Next, we observe that $\cL^2 (\cF (\h)) \simeq \cF (\h) \otimes \cF (\h)$. This isomorphism is induced by the map $|\psi \rangle \langle \phi| \to \psi \otimes \overline{\phi}$, extended by linearity to the whole space $\cL^2 (\cF (\h))$. The mixed state with density matrix (\ref{eq:densma}) can be described on $\cF (\h) \otimes \cF (\h)$ by the vector
\[ \kappa = \sum_{n} \e_n \psi_n \otimes \overline{\phi}_n. \]
The expectation of the operator $A$ on $\cF (\h)$ in the state $\kappa \in \cF (\h) \otimes \cF (\h)$ is given by
\begin{equation}\label{eq:trAr} \tr \, A \rho = \tr \, A \wt{\kappa} \wt{\kappa}^* = \langle \kappa, (A \otimes 1) \kappa \rangle_{\cF (\h) \otimes \cF (\h)} \, . 
\end{equation}

The doubled Fock space $\cF(\h)\otimes \cF(\h)$ is isomorphic to $\mathcal{F}(\h \oplus \h)$ (see \cite{DG} or any book on mathematical quantum field theory). The unitary map $U$ that implements this isomorphism is known as the \emph{exponential law} and is defined by the relations 
\[
U (\Omega_{\cF(\h)}\otimes\Omega_{\cF(\h)}) = \Omega_{\mathcal{F}(\h \oplus \h)}
\]
and
\be\label{eq:a00}
\begin{split}
U \left[ a(f)\otimes 1 \right] U^{*} &= a(f\oplus 0) =: a_{l}(f) \\ 
U \left[ (-1)^{\cN} \otimes a(f) \right] U^{*} &= a(0\oplus f) =: a_{r}(f) \end{split} \ee
for all $f\in \h$, where $a_{\s}(f)$, $\s = l,\,r$, are called the \emph{left and right representations} of $a(f)$, respectively. By hermitian  conjugation, we also find 
\begin{equation}\label{eq:a00*} \begin{split}
U \left[ a^{*}(f)\otimes 1 \right] U^{*} &= a^{*}(f\oplus 0) =: a^{*}_{l}(f) \\ U \left[ (-1)^{\cN}\otimes a^{*}(f) \right] U^{*} &= a^{*}(0\oplus f) =: a^{*}_{r}(f)
\end{split}
\end{equation}
where $a^{*}_{\s}(f)$, $\s = l,\,r$, are the left and right representations of $a^{*}(f)$. Notice that the presence of the operator $(-1)^{\cN}$ on the second line of (\ref{eq:a00}) and (\ref{eq:a00*}) guarantees that creation operators on the space $\cF (\h \oplus \h)$ satisfy the canonical anticommutation relations (and in particular that $a^\sharp_l (f)$ anticommutes with $a^\sharp_r (g)$, for all $f,g \in \h$). 

It is convenient to introduce the left and right representations of the operator-valued distributions $a_{x}$, $a^{*}_{x}$ by the relations
\be\label{eq:a01}
\begin{split}
a_{l}(f) &= \int dx\, a_{x,l}\overline{f(x)},\qquad a_{r}(f) = \int dx\, a_{x,r}\overline{f(x)}, \\
a^{*}_{l}(f) &= \int dx\, a^{*}_{x,l} f(x),\qquad a^{*}_{r}(f) = \int dx\, a^{*}_{x,r} f(x),
\end{split}
\ee
for all $f\in \h$. We also define the left and right representations of the second quantization of operators on $\h$ by 
\[ \begin{split} 
U \left[ d\G(O)\otimes 1 \right] U^{*} &= d\G(O\oplus 0) =: d\G_{l}(O), 
\\ U \left[ 1\otimes d\G(O) \right] U^{*} &= d\G(0\oplus O) =: d\G_{r}(O)\;.\end{split} \]
The left and right representations of $d\G(O)$ can be written in terms of the operator-valued distributions as
\be
d\G_{l}(O) = \int dxdy\, O(x;y)a^{*}_{x,l}a_{y,l},\qquad d\G_{r}(O) = \int dxdy\, O(x;y)a^{*}_{x,r}a_{y,r}.
\ee

According to (\ref{eq:trAr}), the expectation of the observable $A$ in the mixed state described by the vector $\psi = U \kappa \in \cF (\h \oplus \h)$ is given by 
\[ \langle \kappa, (A \otimes 1) \kappa \rangle_{\cF (\h) \otimes \cF (\h)} = \langle U\kappa, U(A\otimes 1) U^* U \kappa \rangle_{\cF (\h \oplus \h)} = \langle \psi, U (A \otimes 1) U^* \psi \rangle_{\cF (\h \oplus \h)}\;. \]
In particular, the expectation of the second quantization $d\Gamma (O)$ of a one-particle operator $O$ on $\h$ is given by
\[ \langle \psi, U(d\Gamma (O) \otimes 1) U^* \psi \rangle = \langle \psi, d\Gamma_l (O) \psi \rangle = \int dx dy\, O (x;y) \langle \psi , a_{x,l}^* a_{y,l} \psi \rangle = \tr \, O \gamma_\psi^{(1)}, \]
where we defined the one-particle reduced density $\gamma_\psi^{(1)}$ associated with $\psi \in \cF (\h \oplus \h)$ as the non-negative trace class operator on $\h$ having the integral kernel 
\begin{equation}\label{eq:gam1-FF} \gamma_\psi^{(1)} (x;y) = \langle \psi, a_{y,l}^* a_{x,l} \psi \rangle .
\end{equation} 
We also define the pairing density $\alpha_\psi$ associated with $\psi \in \cF (\h \oplus \h)$ as the Hilbert-Schmidt operator on $\h$ with the kernel 
\begin{equation}\label{eq:alpha-FF} \alpha_\psi (x;y) = \langle \psi, a_{y,l} a_{x,l} \psi \rangle . 
 \end{equation}

The above construction allows us to represent mixed states as vectors in the ``larger'' Fock space $\cF(\h\oplus\h)$. This idea is well-known in quantum statistical mechanics and takes the name of \emph{purification}.

\medskip

\emph{Time evolution of mixed states.} The time evolution of the density matrix $\rho$ is given by 
\[ \rho_{t} = e^{-i\mathcal{H}_{N}t/\e}\rho e^{i\mathcal{H}_{N}t/\e},\]
where $\cH_{N}$ is the second quantized Hamiltonian given by (\ref{eq:fockH}). Accordingly, we define the time evolution of $\kappa \in \cF (\h) \otimes \cF (\h)$ by $\kappa_{t} = \left[ e^{-i\mathcal{H}_{N}t/\e} \otimes e^{i\cH_N t/\eps} \right] \kappa$. Let $\psi_t = U \kappa_t$ denote the vector in $\cF (\h \oplus \h)$ describing the mixed state with density matrix $\rho_t$. Then 
\[
\begin{split}
\psi_{t} &= U \kappa_{t} = Ue^{-i(\cH_{N}\otimes 1 - 1\otimes \cH_{N})t/\e}\kappa = e^{-i \cL_{N}t/\e} \psi
\end{split}
\]
where the \emph{Liouvillian} $\cL_{N}$ is defined by 
\[
\cL_{N} = U\big(\mathcal{H}_{N}\otimes 1 - 1\otimes \mathcal{H}_{N}\big)U^{*}.
\]
A more  explicit expression follows from (\ref{eq:a00}) and  (\ref{eq:a01}):
\be\label{a0c}
\begin{split}
\cL_{N} =\;& \e^{2}\int dx\, \nabla_{x}a^{*}_{x,l}\nabla_{x}a_{x,l} + \frac{1}{2N}\int dxdy\, V(x-y)a^{*}_{x,l}a^{*}_{y,l}a_{y,l}a_{x,l}  \\
& - \e^{2}\int dx\, \nabla_{x}a^{*}_{x,r}\nabla_{x}a_{x,r} - \frac{1}{2N}\int dxdy\, V(x-y)a^{*}_{x,r}a^{*}_{y,r}a_{y,r}a_{x,r}.
\end{split}
\ee
Hence, the expectation of an arbitrary operator $A$ on $\cF (\h)$ in the evolved mixed state is given by
\[ \tr\, A \rho_t = \langle \psi_t , U (A\otimes 1) U^* \psi_t \rangle = \langle \psi, e^{i\cL_N t/\e} U (A \otimes 1) U^* e^{-i\cL_N t/\e} \psi \rangle. \]

\medskip

\emph{Araki-Wyss representation.} In Section \ref{sec:mf_fermions}, where we proved that the time-evolution of Slater determinants can be approximated by the Hartree-Fock equation, a crucial ingredient of our analysis was the observation that Slater determinants can be written in the Fock space $\cF (\h)$ in the form $R_{\nu_N} \Omega$, where $\nu_N$ is an appropriate Bogoliubov transformation. In this section, we are interested in the evolution of quasi-free mixed states, which, after purification, can be described by vectors in the Hilbert space $\cF (\h \oplus \h)$. 

Let $\nu : (\h \oplus \h) \oplus (\h \oplus \h) \to (\h \oplus \h) \oplus (\h \oplus \h)$ be an implementable Bogoliubov transformation on the doubled one-particle space. 
Then there exists a unitary map $R_\nu : \cF (\h \oplus \h) \to \cF (\h \oplus \h)$ implementing $\nu$. It is easy to see that vectors of the form $R_\nu \Omega_{\cF (\h \oplus \h)} \in \cF (\h\oplus \h)$ are quasi free in the sense that higher order correlations can be computed through  Wick's rule. 

In the following we will be interested in quasi-free states in $\cF (\h \oplus \h)$ with average number of particles equal to $N$ and with vanishing pairing density. Let $\omega_N$ be a non-negative trace class operator on $\h$ with $0 \leq \omega_N \leq 1$ and $\tr\, \omega_N = N$. Then we define $\nu_N : (\h \oplus \h)\oplus (\h \oplus \h) \to (\h \oplus \h) \oplus (\h \oplus \h)$ by
\begin{equation}\label{eq:def-nuN} \nu_N = \left( \begin{array}{ll} U_N & \overline{V}_N \\ V_N & \overline{U}_N \end{array} \right), 
\end{equation}
where
\be\label{a8}
U_N = \begin{pmatrix} u_N & 0 \\ 0 & \overline{u}_N \end{pmatrix},\qquad V_N = \begin{pmatrix} 0 & \overline{v}_N \\ -v_N & 0 \end{pmatrix},
\ee
and  $u_N = \sqrt{1-\omega_N}$, $v_N = \sqrt{\omega_N}$. The operators $U_N,V_N: \h \oplus \h \to \h \oplus \h$ satisfy the relations (\ref{eq:bogrelations}); hence $\nu_N$ defines a Bogoliubov transformation. Since $\tr\, V_N^{*}V_N = 2 \tr\, \omega_N = 2N<\infty$, the Bogoliubov transformation $\nu_N$ is implementable. Hence $R_{\nu_N} \Omega \in \cF (\h \oplus \h)$ describes a quasi-free mixed states (here and in the following, we denote simply by $\Omega$ the vacuum in $\cF (\h \oplus \h)$). We claim now that $R_{\nu_N} \Omega$ is exactly the quasi-free mixed state with reduced density $\omega_N$ and with vanishing pairing density. In fact, for any $f \in \h$, we find 
\be
\begin{split}
R^{*}_{\nu_N} a_{l}(f) R_{\nu_N} & = a_{l}(u_N f) - a^{*}_{r}(\overline{v}_N \overline{f}), \\ 
R^{*}_{\nu_N} a_{r}(f) R_{\nu_N} &= a_{r}(\overline{u}_N f) + a^{*}_{l}(v_N \overline f).
\end{split}
\ee
Equivalently,
\be\label{a8b}
\begin{split}
R^{*}_{\nu_N} a_{x,l} R_{\nu_N} & = a_{l}(u_{N,x}) - a^{*}_{r}(\overline{v}_{N,x}), \\ 
R^{*}_{\nu_N} a_{x,r} R_{\nu_N} &= a_{r}(\overline{u}_{N,x}) + a^{*}_{l}(v_{N,x}),
\end{split}
\ee
with the usual notation $u_{N,x} (y) = u_N (y;x)$, $v_{N,x} (y)= v_N (y;x)$. Using (\ref{a8b}) and the unitarity of $R_{\nu_N}$ we find, from (\ref{eq:gam1-FF}), 
\be
\begin{split}
\g_{R_{\nu_N} \Omega}(x;y) &= \langle R_{\nu_N}\Omega, a^{*}_{y,l}a_{x,l} R_{\nu_N}\Omega\rangle \\
& = \langle \Omega, R^{*}_{\nu_N}a^{*}_{y,l}R_{\nu_N}R^{*}_{\nu_N}a_{x,l}R_{\nu_N}\Omega\rangle \\
& = \langle \Omega, (a^{*}_{l}(u_{y}) - a_{r}(\overline{v}_{y}))(a_{l}(u_{x}) - a^{*}_{r}(\overline{v}_{x})) \Omega\rangle \\
& = \langle \Omega, a_{r}(\overline{v}_{y})a^{*}_{r}(\overline{v}_{x}) \Omega\rangle\\
& = (v_N^{*}v_N)(x;y) = \omega_{N}(x;y)
\end{split}
\ee
where we used the canonical anticommutation relations and the fact that $a_{l}$ and $a_{r}$ annihilate the vacuum. Also, from (\ref{eq:alpha-FF}), 
\be
\begin{split}
\alpha_{R_{\nu_N} \Omega}(x;y) &= \langle R_{\nu_N}\Omega, a_{y,l}a_{x,l} R_{\nu_N}\Omega\rangle\\
&= \langle \Omega, (a_{l}(u_{y}) - a^{*}_{r}(\overline{v}_{y})(a_{l}(u_{x}) - a^{*}_{r}(\overline{v}_{x}))\Omega\rangle\\
&= \langle \Omega, a_{l}(u_{y})a^{*}_{r}(\overline{v}_{x})\Omega\rangle = 0
\end{split}
\ee
where we used that $\{a_{l}(u_{y}),\, a_{r}^{*}(\overline{v_{x}})\} = 0$.

This representation of quasi-free mixed states is a special example of a well known construction in quantum statistical mechanics, known as the \emph{Araki-Wyss representation} \cite{AW}.

\medskip

\emph{Dynamics of quasi-free mixed states.} We consider the time evolution of initial quasi-free mixed states satisfying certain semiclassical estimates (motivated by the idea that, physically, we are interested in the evolution of equilibrium states at positive temperature). For such initial data, we show that the evolution remains close to a mixed quasi-free state with reduced density evolving according to the Hartree-Fock equation (\ref{eq:HF}). The next theorem is taken from \cite{BJPSS}.
\begin{theorem}\label{thm:mixedHF}
Let $V\in L^{1}(\bR^{3})$ and assume that
\begin{equation}\label{eq:Vass}
\int dp\, (1+|p|^2) |\widehat{V}(p)| < \infty.
\end{equation}
Let $\omega_{N}$ be a sequence of operators on $\h = L^{2}(\bR^{3})$ with $0\leq \omega_{N}\leq 1$, $\tr \, \omega_N = N$ and such that $\tr\,(1-\D)\omega_{N}<\infty$ and
\begin{equation}\label{eq:oass}
\begin{split}
\|[v_N, x]\|_{\text{HS}} &\leq CN^{1/2}\e, \qquad \|[v_N,\e\nabla]\|_{\text{HS}} \leq CN^{1/2}\e, \\
\|[u_{N}, x]\|_{\text{HS}} &\leq CN^{1/2}\e, \qquad  \|[u_N,\e\nabla]\|_{\text{HS}} \leq CN^{1/2}\e,
\end{split}
\end{equation}
with $v_{N} = \sqrt{\omega_{N}}$, $u_{N} = \sqrt{1 - \omega_{N}}$, for a suitable constant $C>0$. Let $\nu_{N}$ denote the Bogoliubov transformations (\ref{eq:def-nuN}), such that $R_{\nu_{N}}\Omega_{\cF (\h\oplus\h )}$ is the quasi-free state on $\cF (\h \oplus \h)$ with one-particle reduced density $\omega_N$ and with vanishing pairing density. 

Let $\g^{(1)}_{N,t}$ be the one-particle reduced density associated with the evolved state
\be\label{eq:qf-evo}
\psi_{N,t} = e^{-i\cL_{N}t/\e}R_{\nu_{N}}\Omega_{\cF(\h\oplus \h)}
\ee
where the Liouvillian $\cL_{N}$ has been defined in (\ref{a0c}). Let $\omega_{N,t}$ be the solution of the time-dependent Hartree-Fock equation
\be\label{eq:HF2}
i\e\partial_{t}\omega_{N,t} = [-\e^{2}\Delta + V*\rho_{t} - X_{t},\omega_{N,t}]
\ee
with the initial data $\omega_{N,0} = \omega_{N}$. Then there exist constants $C,\,c>0$ such that
\be\label{eq:conv-HS}
\norm{\g^{(1)}_{N,t} - \omega_{N,t} }_{\text{HS}}^{2} \leq C \exp(c\exp(c|t|))\;\quad \text{ and } \quad \tr\, \big| \g^{(1)}_{N,t} - \omega_{N,t} \big| \leq C N^{1/2} \exp(c\exp(c|t|)).
\ee
\end{theorem}

\emph{Remarks.} 
\begin{itemize}
\item Similarly to Theorem \ref{thm:HF}, the convergence towards the Hartree-Fock dynamics can be extended to more general initial data than those appearing in (\ref{eq:qf-evo}). Let $\cN = d\Gamma_{l}(1) + d\Gamma_{r}(1)$ be the total number of particles operator in the doubled Fock space and let $\xi_N$ be a sequence in $\cF (\h \oplus \h)$ with $\langle \xi_N, \cN^{10} \xi_N \rangle \leq C$ and such that $\xi_N = \chi (\cN \leq K  N) \xi_N$ (for a sufficiently large constant $K > 0$ independent of $N$). In \cite{BJPSS} it is shown that the one-particle reduced density of $R_{\nu_N} \xi_N$ can be approximated by the solution of the Hartree-Fock equation. 
\item The semiclassical commutator estimates (\ref{eq:oass}) play the same role as (\ref{eq:semiclassicalmain}) in the analysis of the evolution of Slater determinants. There is however an important difference. Eq.\ (\ref{eq:oass}) gives bounds for the Hilbert-Schmidt norm of the commutators, while (\ref{eq:semiclassicalmain}) was expressed in terms of their trace norm. In fact, for Slater determinants we do not expect (\ref{eq:oass}) to be correct. In this case the decay of the  kernel $\omega_N (x;y)$ for $|x-y| \gg \eps$ is quite weak (because $\omega_N$ is a projection and the decay of $\omega_N$ is dictated by the regularity of the Fourier transform in the $(x-y)$-direction) and does not suffice, in general, to prove (\ref{eq:oass}). On the other hand, in Theorem \ref{thm:mixedHF}, $\omega_N$ does not need to be a projection, and one can expect faster decay of the kernel $\omega_N (x;y)$ in the $(x-y)$-direction. This is the reason why the assumption (\ref{eq:oass}) is appropriate for the study of the dynamics at positive temperature.  

For instance, a reasonable approximation for the one-particle reduced density of a thermal state of the Hamiltonian (\ref{eq:fermi1}) describing trapped interacting fermions is given, in the mean-field limit, by the Weyl quantization 
\be
\wt{\omega}_N (x;y) = \frac{1}{(2\pi \eps)^3} \int dv \, M \left( \frac{x+y}{2}, v \right)  e^{i v \cdot \frac{(x-y)}{\eps}} 
\ee
where $M(x,v)$ denotes the phase-space density,
\begin{equation} \label{eq:Mpos} M (x,v) = g_{T,\mu} \left(v^2 - c \rho_{\text{TF}}^{2/3} (x) \right)\;,
\end{equation}
corresponding to the Fermi-Dirac distribution 
\begin{equation}\label{eq:FDdi} 
g_{T,\mu} (E) = \frac{1}{1+ e^{(E-\mu)/T}}\;, 
\end{equation}
depending on the temperature $T>0$ and on the chemical potential $\mu\geq 0$. In (\ref{eq:Mpos}), $\rho_{\text{TF}}$ is the minimizer of the Thomas-Fermi functional (\ref{eq:TF}) and the normalization constant $c>0$ has to be chosen so that 
\[ \int M(x,v) dx dv = 1\;. \]
For the reduced density $\wt{\omega}_N$, we find
\[ [\wt{\omega}_N, x] (x;y) = \frac{i\eps}{(2\pi\eps)^3} \int dv \, (\nabla_v M) \left( \frac{x+y}{2}, v \right) e^{iv \cdot \frac{x-y}{\eps}} \]
and 
\[ [\wt{\omega}_N , \eps \nabla ] (x;y) = - \frac{\eps}{(2\pi\eps)^3} \int dv \, (\nabla_x M) \left( \frac{x+y}{2}, v \right) e^{iv \cdot \frac{x-y}{\eps}} \]
Hence,
\[ \begin{split} \| [\wt{\omega}_N, x] \|^2_{\text{HS}} &\leq \eps^2 N \int dx dv |(\nabla_v M) (x,v)|^2  \leq CN \eps^2, \\  \| [\wt{\omega}_N, x] \|^2_{\text{HS}} &\leq \eps^2 N \int dx dv |(\nabla_x M) (x,v)|^2 \leq C N\eps^2, \end{split}\]
using the regularity of the Fermi-Dirac distribution (\ref{eq:FDdi}) (and assuming some regularity of the Thomas-Fermi density $\rho_{\text{TF}}$). This motivates the assumption (\ref{eq:oass}) for initial data describing thermal states at positive temperature (the condition (\ref{eq:oass}) is actually expressed in terms of $v_N = \sqrt{\omega_N}$ and $u_N = \sqrt{1-\omega_N}$ but the same argument is expected to hold also in this case). 
\item The proof of Theorem \ref{thm:mixedHF} follows a strategy conceptually similar to the one of Theorem \ref{thm:HF}. As for pure states, the rate of the convergence towards the Hartree-Fock evolution can be estimated by controlling the growth of the expectation of the number of particles operator $\cN = d\Gamma_{l}(1) + d\Gamma_{r}(1)$ with respect to the fluctuation dynamics, which is now defined by 
\be
\cU_{N}(t;s) = R^{*}_{\nu_{N,t}}e^{-i\cL_{N} t/\e}R_{\nu_{N,s}}\;.
\ee
It turns out, however, that new ideas are needed here to control the growth of \[ \langle \cU_N (t;0) \Omega, \cN \cU_N (t;0) \Omega \rangle_{\cF (\h \oplus \h)}.\] The main difference with respect to Section \ref{sec:mf_fermions} is that now $u_N = \sqrt{1-\omega_N}$ and $v_N= \sqrt{\omega_N}$ are not orthogonal. To circumvent this problem, one has to use certain cancellations between different terms in the generator of the fluctuation dynamics. Moreover, it is important to introduce an auxiliary dynamics $\wt\cU_N (t;s)$ which, on the one hand, can be proven to stay close to the original fluctuation dynamics, and, on the other hand, only changes the expectation of the number of particles in a controllable way. The details of the proof can be found in \cite{BJPSS}.
\end{itemize}

\appendix
\section{The Role of Correlations in the Gross-Pitaevskii Energy}
\label{sec:gpenergy}

In this appendix we consider bosonic systems in the Gross-Pitaevskii regime, as in Section \ref{sec:GP}. We show that the energy of Fock space states of the form $W(\sqrt{N}\varphi) T_0 \xi_N$, for $\xi_N \in \cF$ with a finite number of particles and finite energy ($\xi_N$ describes excitations of the condensate), is to leading order given by the Gross-Pitaevskii functional (\ref{eq:GPen}) evaluated on the one-particle wave function $\ph \in L^2 (\bR^3)$. This is an instructive calculation since it shows that the Bogoliubov transformation $T_0$ converts part of the many-body kinetic energy into a contribution to the quartic term in the Gross-Pitaevskii functional. Without the Bogoliubov transformation $T_0$ the approximate coherent state $W(\sqrt{N}\varphi) \xi_N$ would have a higher energy, given by a functional similar to (\ref{eq:GPen}) but with coupling constant $b_0 = \int V (x) dx$, which is always larger than $4\pi a_0$. Hence, introducing the Bogoliubov transformation $T_0$ lowers the energy by a quantity of order $N$ (while the change of the number of particles is only of order one; see \eqref{eq:transformN}). This observation supports the claim that states of the form $W(\sqrt{N}\varphi) T_0 \xi_N$ provide a good approximation for the many-body ground state and that $T_0$ implements the correct correlation structure.

\medskip

{\it Energy of Bogoliubov states.}
Consider the Hamilton operator describing a Bose gas in the Gross-Pitaevskii regime, trapped by a confining potential $V_{\text{ext}}$,
\[ \cH_N^{\text{trap}} = \int d x \, a_x^* \left( -\Delta_x + V_{\text{ext}} (x) \right) a_x + \frac{1}{2} \int d x d y \, N^2 V (N (x-y)) a_x^* a_y^* a_y a_x.\]
Let $\xi_N \in \cF$ be such that \begin{equation}\label{eq:assxiN} \left\langle \xi_N, \left(\cN+\frac{1}{N} \cN^2+\cH_N^{\text{trap}}\right)\xi_N \right\rangle \leq C
\end{equation}
uniformly in $N$. As in (\ref{eq:Tt}), we define the Bogoliubov transformation 
\[ T_0 = \exp \left( \frac{1}{2} \int dx dy (k_0 (x;y) a_x^* a_y^* - \overline{k}_0 (x;y) a_x a_y) \right) \]
with the kernel $k_0 (x;y) = - N \omega (N (x-y)) \ph (x) \ph(y)$. We claim that for sufficiently regular $\ph \in L^2 (\bR^3)$ with $\| \ph \|_2 = 1$  
\begin{equation}\label{eq:gpgpgp} \begin{split} \Big\langle W(\sqrt{N} \varphi) T_0\xi_N &, \cH_N^{\text{trap}} W(\sqrt{N}\varphi) T_0\xi_N \Big\rangle \\ &= N \int d x \left( |\nabla \varphi (x)|^2 + V_{\text{ext}} (x) |\varphi (x)|^2 + 4 \pi a_0 |\varphi (x)|^4 \right) + \mathcal{O} (\sqrt{N}) \\ & = N \cE_{\text{GP}} (\varphi) + \mathcal{O} (\sqrt{N}). \end{split} \end{equation}
Taking $\varphi$ to be the normalized minimizer of $\cE_{\text{GP}}$, we conclude from \cite{LSY} that $W(\sqrt{N}\varphi) T_0\xi_N$ has, in leading order, the same energy as the ground state of the restriction of $\cH_N^\text{trap}$ to the $N$-particle sector.

\medskip

{\it Sketch of the proof of (\ref{eq:gpgpgp}).}  Let us consider the case $\xi_N = \Omega$. We compute 
\begin{equation}\label{eq:vac-ex}\Big\langle W(\sqrt{N} \varphi) T_0 \Omega, \cH_N^{\text{trap}} W(\sqrt{N}\varphi) T_0\Omega \Big\rangle.
\end{equation}
conjugating the Hamiltonian first with the Weyl operator (producing a shift of the creation and annihilation operators) and then with the Bogoliubov transformation (which acts on the annihilation and creation operators as given by (\ref{eq:bogo-aa})). At the cost of commutators appearing, all terms can be brought into normal order. At the end, all terms with creation and annihilation operators written in normal order vanish when we consider the vacuum expectation. Hence, we find the following contributions to (\ref{eq:vac-ex}). {F}rom the kinetic energy we have
\begin{equation}\label{eq:kineticexpansion}\begin{split}
\int dx \, &\left\langle \Omega, T_0^* W(\sqrt{N}\varphi)^* \nabla_x a^*_x \nabla_x a_x W(\sqrt{N}\varphi) T_0 \Omega \right\rangle = N\int d x \lvert \nabla \varphi(x)\rvert^2 + \int d x \, \| \nabla_x \text{sh}_x \|^2 \,.   \end{split}
\end{equation}
{F}rom the external potential we obtain 
\begin{equation}\label{eq:externalexpansion}\begin{split}
\int dx V_\text{ext} (x) \, &\left\langle \Omega, T_0^* W(\sqrt{N}\varphi)^*  a^*_x a_x W(\sqrt{N}\varphi) T_0 \Omega \right\rangle = N \int d x \lvert \varphi(x)\rvert^2 V_{\text{ext}}(x) + \int d x V_{\text{ext}}(x) \norm{\text{sh}_x}_2^2 \, .  
\end{split}\end{equation}
Finally, from the interaction, we get 
\begin{equation}\label{eq:quarticexpansion}
\begin{split}
\int dx dy N^2 V(N(x-y)) &\left\langle \Omega, T_0^* W(\sqrt{N}\varphi)^* a^*_x a^*_y a_y a_x W(\sqrt{N}\varphi) T_0 \Omega \right\rangle 
\\ & = \frac{1}{2}\int dxdy N^2 V(N(x-y)) \lvert \langle \text{ch}_y,\text{sh}_x\rangle\rvert^2\\& \quad + \frac{1}{2} \int d x d y N^3V(N(x-y))\left(\langle \text{sh}_x,\text{ch}_y \rangle \varphi(x) \varphi(y) + \text{c.\,c.}\right)
\\&\quad+ \frac{N}{2}\int d x d y N^3 V(N(x-y)) \lvert \varphi(x)\rvert^2 \lvert \varphi(y)\rvert^2 \, .
\end{split}\end{equation}

To evaluate the terms on the r.\,h.\,s.\ of (\ref{eq:kineticexpansion}), (\ref{eq:externalexpansion}) and (\ref{eq:quarticexpansion}) it is useful (like in Section \ref{sec:GP}) to think of $\text{sh}_x(y) \simeq k_0 (y;x)$ and $\text{ch}_x(y) \simeq \delta(x-y)$; more precisely recall that, in contrast to the regular kernels of higher powers of $k_0$, $k_0(x;y)$ itself is singular as $(\lvert x-y\rvert+1/N)^{-1}$ for $N \to \infty$.  We then find (in the last step using Hardy's inequality)
\[ \begin{split} \norm{\text{sh}_x}^2_2 \simeq \int dy \, \lvert k_0(y;x)\rvert^2 &= \lvert \varphi(x)\rvert^2 \int dy \, \lvert \varphi(y)\rvert^2 \, \lvert N \omega(N(x-y))\rvert^2 \\ &\leq C \lvert \varphi(x)\rvert^2 \int dy \, \frac{\lvert \varphi(y)\rvert^2}{\lvert x-y\rvert^2} \leq \lvert \varphi(x)\rvert^2 \norm{\nabla \varphi}_2^2. \end{split} \]
This implies that 
\begin{equation}\label{eq:vextss} \int dx \, V_{\text{ext}}(x) \norm{\text{sh}_x}_2^2 \leq C, 
\end{equation}
i.\,e.\ the contribution is of order one (for sufficiently regular $\ph \in L^2 (\bR^3)$). Hence, this term can be neglected. 

As for the second term on the r.\,h.\,s.\ of (\ref{eq:kineticexpansion}), we notice that 
\[\begin{split}\int d x \| \nabla_x \text{sh}_x \|^2  & = \int d x \langle \text{sh}_x,(-\Delta_x) \text{sh}_x\rangle\\& \simeq - \int dx dy N \omega(N(x-y)) \overline{\varphi(x)} \overline{\varphi(y)} \Delta_x \left[ N \omega\left(N\left(x-y\right)\right) \varphi(x) \varphi(y) \right] \\
&=
 \int dx dy N^4 \omega(N(x-y)) \lvert \varphi(x)\rvert^2 \lvert \varphi(y)\rvert^2 (-\Delta \omega)(N(x-y)) + \mathcal{O}(\sqrt{N}) \end{split}\]
because contributions arising when one or two derivatives act on $\ph$ are of smaller order. 
The key observation now is the fact that $1-\omega$ satisfies the zero-energy scattering equation (\ref{eq:scat}), which implies that
\[(-\Delta \omega) (N(x-y)) = \frac{1}{2} V(N(x-y))(1-\omega(N(x-y))) \]
and thus that 
\begin{equation}\label{eq:Voo} \int d x \, \| \nabla_x \text{sh}_x \|^2 = \frac{1}{2} \int dx dy N^4 V (N (x-y)) \omega (N (x-y)) (1-\omega (N (x-y))) \, |\ph (x)|^2 |\ph (y)|^2 + O (\sqrt{N}). \end{equation}

Thus, from (\ref{eq:kineticexpansion}), (\ref{eq:externalexpansion}) and (\ref{eq:quarticexpansion}), we conclude (using again the approximations $\text{ch}_y (z) \simeq \delta (z-y)$ and $\text{sh}_x (z) \simeq k_0 (z;x)$) that
\[
\begin{split}
 \big\langle W (\sqrt{N} \ph) T_0 &\Omega , \cH_N^\text{trap} \, W(\sqrt{N} \ph) T_0 \Omega \big\rangle \\ = & \; N \int d x \, \lvert \nabla \varphi(x)\rvert^2 + N \int dx\, V_{\text{ext}}(x) \lvert \varphi(x)\rvert^2  \\
&+ \frac{N}{2} \int dx dy \, N^3 V\left(N\left(x-y\right)\right) \omega\left(N\left(x-y\right)\right)\left(1-\omega\left(N\left(x-y\right)\right) \right) \, \lvert \varphi(x)\rvert^2 \lvert \varphi(y)\rvert^2 \\
&+ \frac{N}{2} \int dx dy \, N^3 V(N(x-y)) \omega(N(x-y))^2 \lvert \varphi(x)\rvert^2 \lvert \varphi(y)\rvert^2\\
&- N \int dx dy \, N^3 V(N(x-y)) \omega(N(x-y)) \lvert \varphi(x) \rvert^2 \lvert \varphi(y)\rvert^2\\
&+ \frac{N}{2} \int dx dy \, N^3 V(N(x-y)) \lvert \varphi(x)\rvert^2 \lvert \varphi(y)\rvert^2 + \mathcal{O}(\sqrt{N})).
\end{split} \]
Combining all terms proportional to the interaction potential, we find  
\[ \begin{split} \big\langle W &(\sqrt{N} \ph) T_0 \Omega , \cH_N^\text{trap} \, W(\sqrt{N} \ph) T_0 \Omega \big\rangle \\ & = N \left[\int d x \left( \lvert \nabla\varphi(x) \rvert^2 + V_{\text{ext}}(x)\lvert\varphi(x)\rvert^2\right) \right. \\ &\hspace{3cm} \left. + \frac{1}{2} \int dx dy N^3 V\left(N\left(x-y\right)\right) f(N(x-y)) \lvert \varphi(x)\rvert^2 \lvert \varphi(y)\rvert^2\right] + \mathcal{O}(\sqrt{N}).
  \end{split} \]
Since $\int dx \, V(x) f(x) = 8 \pi a_0$, we obtain 
in the limit $N \to \infty$ (using again some regularity of $\ph$)
 \[ \begin{split} \big\langle W &(\sqrt{N} \ph) T_0 \Omega , \cH_N^\text{trap} \, W(\sqrt{N} \ph) T_0 \Omega \big\rangle= N \left[\int d x \left( \lvert \nabla\varphi(x) \rvert^2 + V_{\text{ext}}(x) + 4\pi a_0 \lvert \varphi(x)\rvert^4 \right) \right] + \mathcal{O}(\sqrt{N})\, ,
  \end{split}  \]
which proves (\ref{eq:gpgpgp}). This computation can be easily extended (estimating all normally ordered terms emerging from \eqref{eq:vac-ex} by \eqref{eq:assxiN}) to arbitrary $\xi_N$ satisfying (\ref{eq:assxiN}).

\end{document}